\documentclass{article}
\usepackage{}
\usepackage{enumerate}
\usepackage{mathbbold}
\usepackage{bbm}
\usepackage{amsfonts}
\usepackage[T1]{fontenc}
\usepackage{latexsym,amssymb,amsmath,amsfonts,amsthm}
\usepackage{graphicx, color}
\usepackage{subfigure}
\usepackage{epsf}
\usepackage{epstopdf}
\usepackage{amssymb}
\usepackage{authblk} 
\usepackage{amsthm}
\usepackage{multirow}
\usepackage{diagbox}
\usepackage{stmaryrd}
\usepackage{mathrsfs}
\usepackage[ruled]{algorithm2e}
\usepackage{hyperref}
\hypersetup{
colorlinks =true, 
linkcolor=black, 
anchorcolor=black,
citecolor=black}

\renewcommand{\theequation}{\arabic{section}.\arabic{equation}}

\newcommand{\R}{{\mathbb R}}
\newcommand{\Z}{{\mathbb Z}}
\newcommand{\N}{{\mathbb N}}

\newcommand{\be}{\begin{eqnarray}}
\newcommand{\ben}{\begin{eqnarray*}}
\newcommand{\en}{\end{eqnarray}}
\newcommand{\enn}{\end{eqnarray*}}
\newcommand{\ba}{\backslash}
\newcommand{\pa}{\partial}

\newcommand{\ov}{\overline}

\newcommand{\Om}{\Omega}

\newcommand{\la}{\lambda}

\newcommand{\wi}{\widetilde}

\newcommand{\hx}{\hat{x}}

\newtheorem{theorem}{Theorem}[section]
\newtheorem{lemma}[theorem]{Lemma}
\newtheorem{assumption}[theorem]{Assumption}

\begin{document}
%\title{\bf Identifying extended sources for biharmonic wave equation from the scattered fields at sparse sensors}
\title{\bf The Radon Transform-Based Sampling Methods for Biharmonic Sources from the Scattered Fields}

\author[1]{Xiaodong Liu}
\author[2]{Qingxiang Shi}
\author[1,3]{Jing Wang} 
\makeatletter % 重新定义 \AB@affilnote 命令，使其仅添加脚注标记 
\renewcommand{\AB@affilnote}[1]{\footnotemark[#1]} % 重新定义 \AB@affillist 命令，使其不输出任何内容 
\renewcommand{\AB@affillist}[1]{}
\renewcommand{\AB@affillist}[2]{} 
\makeatother

\date{}
\maketitle
\footnotetext[1]{State Key Laboratory of Mathematical Sciences, Academy of Mathematics and Systems Science,
Chinese Academy of Sciences, Beijing 100190, China. Email: xdliu@amt.ac.cn} 
\footnotetext[2]{School of Mathematical Sciences, Dalian University of Technology, Dalian 116024, China. Email: sqxsqx142857@dlut.edu.cn}
\footnotetext[3]{Corresponding author. School of Mathematical Sciences, University of Chinese Academy of Sciences, Beijing 100049, China, Email:wangjing23@amss.ac.cn}

\begin{abstract}
%This work is dedicated to uniqueness and numerical algorithms for determining the extended sources of the biharmonic wave equation using scattered field data. First, we establish that the explicit formula for the real-valued source function $S(x)$ can be obtained by only using scattered field data $u^s(x,k)$. Furthermore, we show that an even simpler explicit expression for $S(x)$ can be derived when both $u^s(x,k)$ and $\Delta_x u^s(x,k)$ are employed, which demonstrates certain advantages in numerical experiments.Second, if the support set $\Omega$ of $S(x)$ is composed of annuluses and polygons, we prove that $\Omega$ can be uniquely determined from scattered field data $u^s(x,k)$ at sparse sensors. Numerically, inspired by the uniqueness arguments, we propose three indicator functions, which can efficiently reconstruct the source function $S(x)$ and the boundary of its support $\Omega$.

This paper presents three quantitative sampling methods for reconstructing extended sources of the biharmonic wave equation using scattered field data. The first method employs an indicator function that solely relies on scattered fields $ u^s$ measured on a single circle, eliminating the need for Laplacian or derivative data. Its theoretical foundation lies in an explicit formula for the source function, which also serves as a constructive proof of uniqueness. To improve computational efficiency, we introduce a simplified double integral formula for the source function, at the cost of requiring additional measurements $\Delta u^s$. This advancement motivates the second indicator function, which outperforms the first method in both computational speed and reconstruction accuracy. The third indicator function is proposed to reconstruct the support boundary of extended sources from the scattered fields $ u^s$ at a finite number of sensors. By analyzing singularities induced by the source boundary, we establish the uniqueness of annulus and polygon-shaped sources. A key characteristic of the first and third indicator functions is their link between scattered fields and the Radon transform of the source function. Numerical experiments demonstrate that the proposed sampling methods achieve high-resolution imaging of the source support or the source function itself.

\vspace{.2in}
{\bf Keywords:} extended sources; biharmonic wave equation; sparse data; sampling method.

\vspace{.2in} {\bf AMS subject classifications:}
35R30, 35P25
\end{abstract}

\section{Introduction}
Due to the wide applications of inverse source problems in scientific and industrial fields \cite{App-1, App-2, App-3}, extensive research has been conducted on the mathematical analysis and numerical computation of acoustic waves, electromagnetic waves, and elastic waves \cite{1990application, 2020Identification, JiLiu-nearfield, electromagnetic sources, elastic sources,2025extended-sources-acoustic, 2025extended-sources-elastic+magnetic}. However, there are very few results for biharmonic waves, which model an important scattering problem in an infinitely thin elastic plate. Many challenges arise for the  analysis and computation of the biharmonic waves due to the inherent properties of a solution to the high order equations. Owing to the existence of non-radiating sources \cite{non-radiating}, extended sources cannot be uniquely determined from scattered field data at a single frequency. In this work, we study the inverse biharmonic source problems using multi-frequency scattered fields.

%Following the approach adopted in many studies on inverse source problems of the biharmonic wave equation \cite{Y.Guo, Y.Guo-Phaseless data, non-radiating, 2021application1, 2021application2, 2022application}, we therefore employ multi-frequency data to address the uniqueness analysis and computational reconstruction for inverse source problems involving biharmonic wave equations. 

Throughout this paper, we denote by $B_r(x)$ and $\pa B_r(x)$, respectively, the disk and the circle centered at $x\in\R^2$ with radius $r$. Let $S(x)\in L^{2}(\R^{2})$ be a real-valued source function with compact support  $\Omega \subset B_{R}(0)$. Then the source $S$ gives rise to a scattered field $u^s$ satisfying 
\begin{align}\label{Main_eq}
\Delta_x^{2} u^s(x,k)-k^4u^s(x,k)=S(x)\quad\mathrm{~in~}\mathbb{R}^2
\end{align}
and
\begin{align}\label{se2}
\partial_rw-ikw=o\left(r^{-\frac{n-1}{2}}\right),\ r=|x|\to\infty,\ w=u^s,\Delta_x u^s.
\end{align}
Here, $u^s$ denotes the out-of-plane displacement of the mid-surface of the thin elastic plate, $k>0$ is the wave number.

There are a few works for inverse biharmonic source problems.
Uniqueness analysis and the Fourier method for reconstructing extended sources are investigated in \cite{Y.Guo} and \cite{Y.Guo-Phaseless data}, which utilize multi-frequency scattering data.
We refer to \cite{2021application1, 2021application2} for the stability analyses of the inverse source problem.
It should be noted that all the works \cite{Y.Guo, Y.Guo-Phaseless data,2021application1, 2021application2} rely on both $u^s(x,k)$ and $\Delta_x u^s(x,k)$  collected around the unknown source. In this work, we consider the inverse problems using less data. The first inverse problem is formulated as follows.

\emph{\textbf{IP(1)}: Determine the source function $S(x)$ from multi-frequency scattered fields measured on a circle:
\begin{equation}
\{u^s(x,k)\mid x\in \pa B_R(0),\, k\in \left[k_{-}, k_{+}\right]\}.
\label{data_source_function}
\end{equation}}

%In particular,  we prove that $S(x)$ can be explicitly represented through integration using only $u^s(x,k)$ data. 
The first contribution of this paper is to give a formula for $S(x)$ by simply using the scattered fields given in \eqref{data_source_function}. Such a formula not only gives a constructive uniqueness proof, but also introduces a direct numerical method. In contrast to the approaches in \cite{Y.Guo, Y.Guo-Phaseless data}, our method requires neither the Laplacian data $\Delta_x u^s(x,k)$ nor the computation of the normal derivatives $\partial _{\nu}u^s(x,k)$ and $\partial _{\nu}\Delta_x u^s(x,k)$. Specifically, inspired by \cite{2007Kunyansky-formula-f}, we establish an identity that connects the scattered field data $u^s(x,k)$ with the Radon transform of the real-valued source function $S(x)$. This relation allows us to derive an explicit expression for $S(x)$ using only the scattered fields given in \eqref{data_source_function}.
Given the additional data $\Delta_x u^s(x,k)$, we will also introduce a more concise double integral expression. Numerically, we find that even with sparse sensors, the source function $S(x)$ can be well identified, demonstrating the significant advantage of our proposed formulation for $S(x)$.

In many practical applications, measurements are only available from a finite number of sensors, which naturally leads to the question of what information can be recovered from such sparse data.
In \cite {JiLiu-nearfield, 2025point-sources-biharmonic}, it is shown that the point sources can be identified from multi-frequency scattered field data measured at sparse sensors. 
%Specifically, the positions and number of point sources can be determined using such sparse sensor data, and the minimum required number of sensors is established theoretically. Furthermore, we derive an explicit formula for computing the scattering strengths of the point sources. 
A natural question arises regarding the more complex inverse source problem: can analogous results be established for extended sources?
Inspired by \cite{singularity-radon-transform}, \cite{2007Kunyansky-formula-f} and \cite{2025extended-sources-acoustic}, this work studies whether multi-frequency scattered field data at sparse sensors can be used to recover partial information on extended sources in the biharmonic wave equation, such as the location and geometry of the support $\Omega$. 
Moreover, when only one sensor is available, what information about the source function can be obtained with such limited data? Both the sparsity of the data and the complexity of the fourth-order biharmonic equation pose challenges to the uniqueness analysis of the corresponding inverse source problem.

Based on the above consideration, the second inverse problem of our interest is formulated as follows.

\emph{\textbf{IP(2)}: Reconstruct the support $\Omega$ of the source function $S$ from multi-frequency scattered fields measured at a set of sparse sensors, given by
\begin{equation}
\{u^s(x,k)\mid x\in \Gamma_L,\, k\in \left[k_{-}, k_{+}\right]\},
\label{data_patial_omega}
\end{equation}}
where $\Gamma_L$ is defined as
\begin{equation}
\Gamma_L:=\{x_1,x_2,\ldots,x_L\}\subset \pa B_R(0) 
\nonumber
\end{equation}
denoting the set of $L\in \N$ sensors used for measuring scattered fields.

The other contribution of this paper is to give an indicator function for reconstructing $\pa \Omega$ by only using sparse data \eqref{data_patial_omega}.
The proposed indicator function is proven to be singular, enabling the identification of singularities induced by the boundary 
$\partial \Omega$ of the source support $\Omega$. In particular, when 
$\Omega$  comprises annuluses and polygons, we demonstrate that the boundary 
$\partial \Omega$ of the source support $\Omega$ can be uniquely determined from the data given in \eqref{data_patial_omega}. Furthermore, the lower bound of $L$ is explicitly expressed in terms of the number of annuluses and polygons.
The authors emphasize that the uniqueness analysis in this work, based on scattered fields from the biharmonic wave equation, differs fundamentally from those obtained using far-field data from the Helmholtz equation in \cite{2025extended-sources-acoustic}. Furthermore, the use of scattered fields significantly enhances the reconstruction accuracy of the source boundary.
%Specifically, by constructing an auxiliary function and analyzing its singular behavior, we identify the characteristic singularities induced by the boundary $\partial \Omega$ of the source support $\Omega$. When $\Omega$ consists of annuluses and polygons, we show that the boundary $\pa \Omega$ of the support $\Omega$ can be uniquely determined from the data \eqref{data_patial_omega}. In particular, the lower bound of $L$ will be clarified in terms of the number of annuluses and polygons.
%Based on the theoretical analysis, we construct an indicator function and perform numerical experiments using the direct sampling method. 
The numerical results demonstrate that when only a single sensor is available, we obtain approximate circles that can roughly delineate the range of the extended source. As the number of sensors increases finitely, we are able to reconstruct both the shape and location of $\pa \Omega$ for the extended source.

The rest of this paper is arranged as follows. Section \ref{Uniqueness results} is divided into two subsections, each dedicated to proving one of the two main theorems. 
In subsection \ref{Identify-source-function}, 
by establishing a relationship between the scattered field and the Radon transform of the source function, we derive an explicit expression for the source function, which gives a constructive uniqueness and introduces a quantitative sampling method for the source function.
In subsection \ref{Reconstruct-partial-Omega}, we introduce an indicator function and show that its singularity is closely related to the boundary $\pa \Omega$. This tool is then used, for annular or polygonal sources, to identify the boundary $\pa \Omega$ from the scattered fields taken at sparse sensors. Then three numerical algorithms are presented in subsections \ref{indicator-1}, \ref{indicator-2} and \ref{indicator-3}, respectively. 
%In subsection \ref{indicator-3}, we also demonstrate that a simpler expression for the source function can be obtained by utilizing both the scattered field data $u^s(x,k)$ and $\Delta u^s(x,k)$, leading to an associated numerical algorithm. 
Finally, in the last section, several numerical examples are provided to illustrate the efficiency and robustness of the three numerical algorithms.

\section{Indicators and their applications for uniqueness}
\label{Uniqueness results}
\setcounter{equation}{0}
As with many inverse problems, the issue of uniqueness is of primary concern. This section focuses on two inverse source problems, designated as \textbf{IP(1)} and \textbf{IP(2)}. 
For \textbf{IP(1)}, we make use of full aperture scattering measurements \eqref{data_source_function}.
In this case, we derive an explicit expression for the source function $S(x)$, which allows not only the reconstruction of the support set $\Omega$, but also the values of $S(x)$ for all $x\in \mathbb{R}^2$.
For \textbf{IP(2)}, we aim to establish uniqueness under a sparse sensor configuration. We seek to determine a lower bound on the number of sensors $L$ such that the support of the source $S(x)$ can be uniquely recovered from the multi-frequency sparse scattering data \eqref{data_patial_omega}.

Denote by $J_{n}$ and $N_{n}$, respectively, the Bessel and Neumann functions of order $n$. Then the fundamental solution of the biharmonic wave equations \eqref{Main_eq} is given by
\begin{equation}
\Phi_k(x,y):=  \frac{i}{8k^2}\left(H_0^{(1)}(k|x-y|)+\frac{2i}{\pi}K_0(k|x-y|)\right),\quad  x,y \in \R^2,\, x\neq y,
   \label{equation_FS}
\end{equation}
where 
\ben
H_0^{(1)}(t):=J_{0}(t)+iN_{0}(t)\quad \mbox{and}\quad
K_0(t):=2e^{-t}\sum_{j=0}^{\infty}\frac{(2j)!}{2^{j}(j!)^3}t^{j},\quad t\in (0, \infty)
\enn
are the Hankel function of the first kind and the Macdonald function of order 0, respectively. Then, the scattered field $u^s$ is given by
\begin{align}\label{Scattered field}
u^s(x,k)=\int_{\R^2}\Phi_k(x,y)S(y)dy,\quad x\in\mathbb{R}^2.
\end{align}
With the help of \eqref{equation_FS} and \eqref{Scattered field}, we derive that 
% \begin{equation}
% \begin{aligned}
\be
I_x(r)&:=& \int_0^{+\infty}8k^3 r\Im( u^s(x,k)) J_0(kr)dk\cr
&=&\int_{\R^2}\int_0^{+\infty}J_0(k|x-y|)krJ_0(kr)dkS(y)dy\cr
&=&\int_{\R^2}\delta(|x-y|-r)S(y)dy\cr
&=&\int_{|y-x|=r}S(y)ds(y),\quad r\in [0,+\infty),
% \end{aligned}
\label{I(r)}
\en
% \end{equation}
where the third equality follows from equation (3.3) of \cite{2025point-sources-biharmonic}.
Due to the fact that the scattered field depends analytically on the wave number $k$, the multi-frequency scattered fields $u^s(x,k), k\in [k_{-}, k_+]$ define $I_x$ by \eqref{I(r)} for fixed $x:=(x_1,x_2)\in \Gamma_L$.
We also want to remark that the right-hand side of \eqref{I(r)} is essentially the Radon transform of the function $S(x)$ over circular paths. 

%%%%%%%%%%%%%%%%%%%%%%%%%%%%%%%%%%%%%%%%%%%%%%%%%%%%%%%%%%%%%%%%%%%%%%%%%%%%%%%%%%%%%%%%%%%%%%%%%%%%%%%%%%%%%%%%%%%%%%%%%%%%%%
%%%%%%%%%%%%%%%%%%%%%%%%%%%%%%%%%%%%%%%%%%%%%%%%%%%%%%%%%%%%%%%%%%%%%%%%%%%%%%%%%%%%%%%%%%%%%%%%%%%%%%%%%%%%%%%%%%%%%%%%%%%%%%%
\subsection{Identification of the source function with multi-frequency scattered fields}\label{Identify-source-function}
%In this section, by establishing a relation between the scattered field $u^s(x,k)$ and the circular Radon transform, and drawing on the results from Section 3 of \cite{2007Kunyansky-formula-f} , we derive an explicit expression for the real-valued source function by using only the scattered field data:
%\begin{equation}
%\{u^s(x,k)| x\in C(0,R),\, k\in \left[k_{-}, k_{+}\right]\}.
%\nonumber
%\end{equation}
%To the best of the authors' knowledge, no previous work has achieved such a direct analytical reconstruction of the source function from such limited data.

To begin with, we introduce a function $I_S(z)$ defined on the plane as follows:
\begin{equation}
\begin{aligned}
I_S(z):=&\frac{R}{\pi}\int_0^{2\pi} \Bigg[  \frac{z-x}{|z-x|}\cdot (\cos\theta, \sin\theta)\Bigg]\int _{0}^{+\infty}\int _0^{2R} \la ^2 r\Big[N_1(\la |z-x|)J_0(\la r)\\
&-J_1(\la |z-x|)N_0(\la r)\Big] \int_0^{+\infty}k^3 \Im(u^s(x,k))J_0(k r) dk dr d\la d\theta, \quad z\in B_R(0),
\end{aligned}
\label{07-I_S}
\end{equation}
where $x=R(\cos\theta, \sin\theta)$. 
We want to emphasize that we have used only the multi-frequency scattering fields $\Im(u^s(x,k))$ taken on the measurement surface $\pa B_R(0)$. 
%Recalling the equation $(7)$ in \cite{2007Kunyansky-formula-f}, we have the following lemma.
The following lemma gives a representation for the source function $S(z)\in L^2_{comp}(\R^2)$. We refer to \cite{2007Kunyansky-formula-f} for a proof.
\begin{lemma}\label{lemma-2007-S(z)-general-radon}
For a function $S(z)\in L^2_{comp}(\R^2)$ with compact support in $B_R(0)$. We have the following formula about $S(z)$:
\begin{equation}
\begin{aligned}
S(z)=&\frac{-1}{8\pi}{\rm div}\int_{\pa B_R(0)}\nu(x)\int_{0}^{+\infty}\Bigg[N_0(\la |z-x|)\left(\int_0^{2R}J_0(\la r)g(x,r)dr\right)\\
&-J_0(\la |z-x|)\left(\int_0^{2R}N_0(\la r)g(x,r)dr\right)\Bigg]\la d\la ds(x),\quad z\in B_R(0),
\label{2007-S(z)-general-radon}
\end{aligned}
\end{equation}
where $\nu(x)$ is the unit outward normal vector at $x\in\pa B_R(0)$ and 
\ben
g(x,r):=\int_{|y-x|=r}S(y)ds(y),\quad x\in\pa B_R(0), \, r\in[0, +\infty).
\enn
\end{lemma}

Based on the above lemma, we show that \eqref{07-I_S} is exactly a formula for computing the source function.
\begin{theorem}
Let $S(z)\in L^2_{comp}(\R^2)$ be a real-valued source function supported in $\Omega\subset B_R(0)$, then
\begin{equation}
I_S(z)=S(z),\quad z\in B_R(0).
   \label{07-I_S=S(z)} 
\end{equation}
    \label{source function-07}
\end{theorem}
\begin{proof}
Recall the equality \eqref{I(r)}, we have
\begin{equation}
\begin{aligned}
\int_0^{+\infty}8k^3 r\Im(u^s(x,k))J_0(kr)dk
=\int_{|x-y|=r}S(y)ds(y),\quad r\in [0, +\infty), x\in B_R(0).
\label{radon-g}
\end{aligned}
\end{equation}
Inserting this into \eqref{07-I_S}, using Lemma \ref{lemma-2007-S(z)-general-radon} and the differentiation formulas
\begin{equation}
\nabla J_0(|z|)=-\frac{z}{|z|}J_1(|z|),\quad \nabla N_0(|z|)=-\frac{z}{|z|}N_1(|z|),\quad z\in \R^2,
\nonumber
\end{equation}
we derive
{\allowdisplaybreaks
\begin{align*}
I_S(z)=&\frac{R}{8\pi}\int_0^{2\pi} \Bigg[  \frac{z-x}{|z-x|}\cdot (\cos\theta, \sin\theta)\Bigg]\int _{0}^{+\infty}\int _0^{2R} \la ^2 \Big[N_1(\la |z-x|)J_0(\la r)\\
&-J_1(\la |z-x|)N_0(\la r)\Big]\int_{|y-x|=r}S(y)ds(y) dr d\la d\theta\\
=&\frac{R}{8\pi}\int_0^{2\pi} \Bigg[ \frac{z-x}{|z-x|}\cdot (\cos\theta, \sin\theta)\Bigg]\int _{0}^{+\infty}\int _0^{2R} \la ^2 \Big[N_1(\la |z-x|)J_0(\la r)\\
&-J_1(\la |z-x|)N_0(\la r)\Big] g(x,r) dr d\la d\theta\\
=&\frac{-1}{8\pi}\text{div}\int_{\pa B_R(0)}\nu(x)\int _{0}^{+\infty} \Bigg[N_0(\la |z-x|)\left(\int_0^{2R}J_0(\la r)g(x,r)dr\right)\\
&-J_0(\la |z-x|)\left(\int_0^{2R}N_0(\la r)g(x,r)dr\right)\Bigg]\la d\la ds(x)\\
=& S(z),\quad z\in B_R(0).
\end{align*}
}
\end{proof}

%%%%%%%%%%%%%%%%%%%%%%%%%%%%%%%%%%%%%%%%%%%%%%%%%%%%%%%%%%%%%%%%%%%%%%%%%%%%%%%%%%%%%%%%%%%%%%%%%%%%%%%%%%%%%%%%%%%%%%%%%%%%%%
%%%%%%%%%%%%%%%%%%%%%%%%%%%%%%%%%%%%%%%%%%%%%%%%%%%%%%%%%%%%%%%%%%%%%%%%%%%%%%%%%%%%%%%%%%%%%%%%%%%%%%%%%%%%%%%%%%%%%%%%%%%%%%%
\subsection{Reconstruction of the support \texorpdfstring{$\Omega$}{Ω} using the multi-frequency sparse scattered fields}
\label{Reconstruct-partial-Omega}
%By the analyticity of the scattered field, complete frequency domain information can be determined from measurements within any finite interval $K:=(k_{-}, k_{+})$. Our goal is to reconstruct the support $\Omega$ using the scattered field patterns $\{u^s(x,k)\mid x\in \Gamma_L,\, k\in \left[k_{-}, k_{+}\right]\}$. 
The uniqueness of the source function cannot be ensured under sparse sensor measurement conditions. In contrast to identifying the source function itself, our focus is exclusively on determining the source support. To achieve this, we make the following assumptions.
\begin{assumption} \label{assumptiom}
Let $S(x)$ be a real-valued function. Denote by $\Omega=\bigcup\limits_{m=1}^{M}\Omega_m$ the support of $S(x)$, where each $\Omega_m$ is a compact annular or polygonal region, and $\Omega_m \cap\Omega_n=\emptyset,\, \forall m\neq n$. Assume further that $S\big|_{\Omega_m}\in C^{1}(\Omega_m)$, and $S(x)\neq 0$ for all $x\in \pa\Omega$.
\end{assumption}
%Let $R[S(x)]$ and $L[S(x)]$ denote the right and left limits of the source function $S(x)$, respectively, and define the jump discontinuity $\llbracket S(x) \rrbracket=R[S(x)]-L[S(x)]$, $x\in \mathbb{R}^2$.

%The Radon transform on circles has been studied by many researchers, but to the best of our knowledge, its application in region reconstruction using sparse scattering field data remains unexplored. 
%The following Lemma \ref{discontinuity-I'}, concerning the discontinuity of $I_x'(r)$, is derived from \cite{singularity-radon-transform}:

%For a function $f$ defined in a neighborhood of $t_0\in \R$, we define 
%\ben
%\llbracket f \rrbracket(t_0):=\lim_{t\rightarrow t_0^+}f(t)-\lim_{t\rightarrow t_0^-}f(t).
%\enn
\begin{lemma}\label{discontinuity-I'}
Let Assumption \ref{assumptiom} hold.  For any sensor $x\in \pa B_R(0)$, $I_x'$ does not exist at $r_0\in \R^{+}$ if and only if there exists constants $ c_{j,i}\in \R, j=1,2,3, i=1,2,\ldots,l $, such that 
\begin{equation}
\llbracket I_x^{\prime} \rrbracket(r_0)
:=\lim_{r\rightarrow r_0^+} I_x^{\prime}(r)-\lim_{r\rightarrow r_0^-} I_x^{\prime}(r)
=\displaystyle\sum\limits_{i=1}^{l}\Bigg[c_{1,i}+\lim\limits_{r\rightarrow r_0^+}\frac{c_{2,i}}{\sqrt{r-r_0}}+\lim\limits_{r\rightarrow r_0^-}\frac{c_{3,i}}{\sqrt{r_0-r}}\Bigg]\neq 0.
\label{asymptotic-expression}  
\end{equation}

\end{lemma}
\begin{proof}
\begin{figure}[h!]
    \centering
    \begin{tabular}{cc}
    \subfigure[The case of a single circular arc]{
    \includegraphics[width=0.40\textwidth]{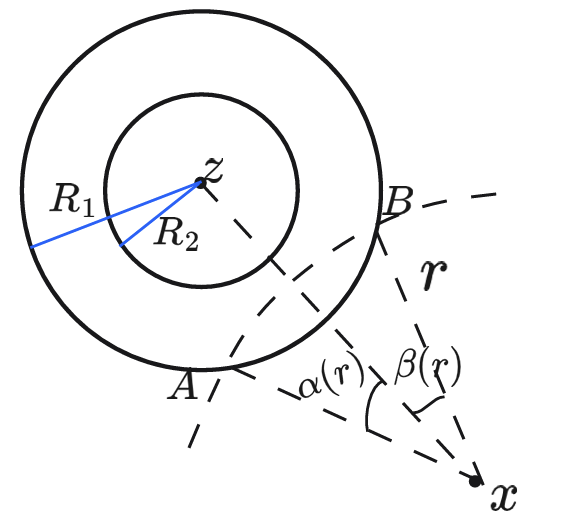}}\hspace{0em} &
    \subfigure[The case of two circular arcs]{
\includegraphics[width=.40\textwidth]{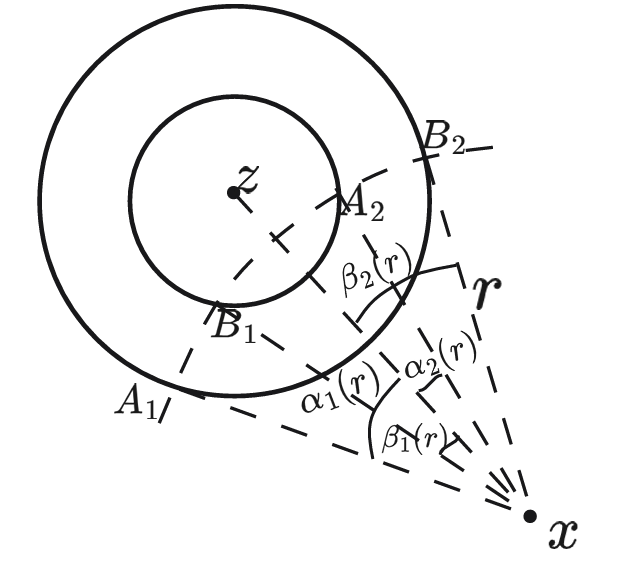}
}
    \end{tabular}
    \caption{Illustrations for the \textbf{Case 1}.}
    \label{Simplification-annular}
\end{figure}
By the definition of $\llbracket I_x^{\prime} \rrbracket(r_0)$, the proof of necessity is straightforward. 
If $I_x'$ does not exist at $r_0\in \R^{+}$, by Section 4 of \cite{singularity-radon-transform}, the circle $\partial B_{r_0}(x)$ passes through a vertex of some polygon or is tangent to $\Omega_m$, for some $m=1,2,\ldots, M$. Therefore, by virtue of linearity, we establish sufficiency by analyzing the following four cases:
\begin{itemize}
    \item Case 1: the circle $\partial B_{r_0}(x)$ is tangent to the annulus $\Omega_m$;
    \item Case 2: the circle $\partial B_{r_0}(x)$ is tangent to exactly one edge of the polygon $\Omega_m$;
    \item Case 3: the circle $\partial B_{r_0}(x)$ passes through exactly one vertex of the polygon $\Omega_m$;
    \item Case 4: the circle $\partial B_{r_0}(x)$ passes through exactly one vertex of the polygon $\Omega_m$ and is tangent to exactly one edge of $\Omega_m$ at that vertex.
\end{itemize}
For simplicity, we assume mutual exclusivity among these cases.
% The first statement has been proved in section 4 of \cite{singularity-radon-transform}.  We only prove the second statement. 

\textbf{Case 1: the circle $\partial B_{r_0}(x)$ is tangent to the annulus $\Omega_m$.}
 
Denote such an annulus by $\ov{B_{R_1}(z)\ba B_{R_2}(z)}$ for some $z\in\R^2$. Note that a circle centered at $x\in \pa B_R(0)$ may intersect the annulus in one of the two configurations shown in Figure \ref{Simplification-annular}, forming one arc and two arcs, respectively. 
% see Figure \ref{Simplification-annular}, then, for any observation point $x$, the integral evaluates to:
% \begin{equation}
% I_x(r)=\int_{\overset{\Large\frown}{AB}(r,y)}S(y)ds(y),
% \nonumber
% \end{equation}
% where $\overset{\Large\frown}{AB}$ is the intersection of the circle arc $C(x,r)$ and the annulus $\Omega$. 
We define the ray $\overrightarrow{xz}$ as the polar axis with an initial angle of $0$, adopting the counterclockwise direction as the positive angular orientation. The intersection points can then be parameterized as $A=(r,\alpha(r))$, $B=(r,\beta(r))$, $A_1=(r,\alpha_1(r))$, $B_1=(r,\beta_1(r))$, $A_2=(r,\alpha_2(r))$, and $B_2=(r,\beta_2(r))$,  as shown in Figure \ref{Simplification-annular}.
For simplicity, we define 
\begin{equation*}
\begin{aligned}
   & \Sigma_1^{(1)}:=(|x-z|-R_1,|x-z|-R_2];\\
   & \Sigma_2^{(1)}:=(|x-z|-R_2,|x-z|+R_2);\\
   & \Sigma_3^{(1)}:=[|x-z|+R_2,|x-z|+R_1).
\end{aligned}
\end{equation*}
Then, we have the explicit expressions
\begin{equation*}
\begin{aligned}
    &\alpha (r)=-\beta(r)=\arccos \left(\frac{|x-z|^2+r^2-R_1^2}{2|x-z|r}\right),\quad r\in \Sigma_1^{(1)}\cup\Sigma_3^{(1)};\\
    &\alpha_1 (r)=-\beta_2(r)=\arccos \left(\frac{|x-z|^2+r^2-R_1^2}{2|x-z|r}\right),\quad r\in \Sigma_2^{(1)};\\
    &\beta_1 (r)=-\alpha_2(r)=\arccos \left(\frac{|x-z|^2+r^2-R_2^2}{2|x-z|r}\right),\quad r\in \Sigma_2^{(1)}.
\end{aligned}
\end{equation*}
Taking the derivative with respect to $r$ on both sides of \eqref{I(r)}, we have
\begin{equation}
I_x'(r)=\begin{cases}
\int_{\beta(r)}^{\alpha(r)}\frac{\pa S(r,\theta)}{\pa r}d\theta+\alpha'(r)S(r,\alpha(r))-\beta'(r)S(r,\beta(r)),&\quad r\in \Sigma_1^{(1)}; \\
\sum\limits_{i=1}^2 \Big[\int_{\beta_i(r)}^{\alpha_i(r)}\frac{\pa S(r,\theta)}{\pa r}d\theta+\alpha_i'(r)S(r,\alpha_i(r))-\beta_i'(r)S(r,\beta_i(r))\Big],&\quad r\in \Sigma_2^{(1)};\\
\int_{\beta(r)}^{\alpha(r)}\frac{\pa S(r,\theta)}{\pa r}d\theta+\alpha'(r)S(r,\alpha(r))-\beta'(r)S(r,\beta(r)),&\quad r\in \Sigma_3^{(1)}.
\label{calculate-annulas-I'}
\end{cases}
\end{equation}
Under the Assumption \ref{assumptiom}, we find that the discontinuity of $I_x'(r)$ in \eqref{calculate-annulas-I'} originates from functions $\alpha'(r)$, $\alpha_i'(r)$, $\beta'(r)$ and $\beta_i'(r)$, $i=1,2$. Straightforward calculations show that
\be
\label{alpha'_beta'_1}
&&\alpha'(r)=-\beta'(r)=\frac{r^2+R_1^2-|x-z|^2}{r\sqrt{R_1^2-(|x-z|-r)^2}\sqrt{(|x-z|+r)^2-R_1^2}}, \quad r\in \Sigma_1^{(1)}\cup\Sigma_3^{(1)};\quad\\
\label{alpha'_beta'_2}
&&\alpha_1'(r)=-\beta_2'(r)=\frac{r^2+R_1^2-|x-z|^2}{r\sqrt{R_1^2-(|x-z|-r)^2}\sqrt{(|x-z|+r)^2-R_1^2}},\quad r\in \Sigma_2^{(1)};\\
\label{alpha'_beta'_3}
&&\beta_1'(r)=-\alpha_2'(r)=\frac{r^2+R_2^2-|x-z|^2}{r\sqrt{R_2^2-(|x-z|-r)^2}\sqrt{(|x-z|+r)^2-R_2^2}},\quad r\in \Sigma_2^{(1)}.
\en

From \eqref{calculate-annulas-I'}-\eqref{alpha'_beta'_3}, it can be directly observed that the discontinuities of $I_x'(r)$ occur only at $r=|x-z|\pm R_i, i=1,2$.
For the critical point $r=|x-z|\pm R_1$, we deduce that
\be
%\begin{aligned}
\llbracket I_x^{\prime} \rrbracket(|x-z|\pm R_1)
 &=&\mp2 S_{\pm}^{R_1} \lim_{r\rightarrow (|x-z|\pm R_1)^\mp}\frac{|x-z|^2-r^2-R_1^2}{r\sqrt{R_1^2-(|x-z|-r)^2}\sqrt{(|x-z|+r)^2-R_1^2}}  \cr
 &=&\frac{\sqrt{2R_1}S_{\pm}^{R_1}}{\sqrt{|x-z|^2\pm|x-z|R_1}}\lim_{\epsilon_{R_1}^{\pm}\rightarrow 0^{+}}\frac{1}{\sqrt{\epsilon_{R_1}^{\pm}}}.
%\end{aligned}
\label{annular-dis1}
\en
Similarly, we have
\begin{equation}
 \llbracket I_x^{\prime} \rrbracket(|x-y|\pm R_2)=\frac{-\sqrt{2R_2}S_{\pm}^{R_2}}{\sqrt{|x-z|^2\pm|x-z|R_2}}\lim_{\epsilon_{R_2}^{\pm}\rightarrow 0^{+}}\frac{1}{\sqrt{\epsilon_{R_2}^{\pm}}}.
\label{annular-dis2}
\end{equation}
% \begin{equation}
% \begin{aligned}
%  \llbracket I_x^{\prime} \rrbracket(|x-y|\pm R_2)&=\pm 2 S_{\pm}^{R_2} \lim_{r\rightarrow (|x-z|\pm R_2)^{\mp}}\frac{|x-z|^2-r^2-R_2^2}{r\sqrt{R_2^2-(|x-z|-r)^2}\sqrt{(|x-z|+r)^2-R_2^2}}  \\
%  &=\frac{\sqrt{2R_2}S_{\pm}^{R_2}}{\sqrt{|x-z|^2-|x-z|R_2}}\lim_{\epsilon_{R_2}^{\pm}\rightarrow 0^{+}}\frac{1}{\sqrt{\epsilon_{R_2}^{\pm}}}.
% \end{aligned}
% \label{annular-dis2}
% \end{equation}
Here, $S_{\pm}^{R_i}:=S|_{\pa \Omega \cap \pa B_{|x-y|\pm R_i}(x)}$ and
$\epsilon_{R_i}^{\pm}=R_i\pm (|x-y|-r)$ define the asymptotic parameters.

\textbf{Case 2: the circle $\partial B_{r_0}(x)$ is tangent to exactly one edge of the polygon $\Omega_m$.}

%Similarly, we consider the case where $\Omega$ is a simply connected polygonal source, as partially illustrated in Figure \ref{Simplification-polygon}(showing a corner point of the polygon). We analyze separately the contributions of the edge $e_1$ and the vertex $v_1$ to the discontinuities of $I_x^{\prime}$. 
We denote such an edge by $e_1:=\overline{v_1v_3}$, as shown in Figure \ref{Simplification-polygon-edge}. Let $l_1$ and $l_2$ be the lines that are perpendicular to the edge $e_1$ and pass through $v_1$ and $v_3$, respectively, as shown in Figure \ref{range-e1}. The lines $l_1$ and $l_2$ intersect the circle $\pa B_R(0)$ at points $A_1$, $A_2$, $B_1$ and $B_2$. 
We claim that if there exist $x$ and $r_0$ such that Case 2 holds for some $m\in \{1,2,\ldots,M\}$, then $\llbracket I_x^{\prime} \rrbracket(r_0)$ is unbounded.
% $\forall x\in \overset{\frown}{A_1B_1}\cup \overset{\frown}{A_2B_2}$ (the red open circular arcs in Figure \ref{range-e1}), there exists $r_0>0$ such that $\llbracket I_x^{\prime} \rrbracket(r_0)$ is unbounded.
In fact, for every point $x\in \Gamma_{e_1}:=\overset{\frown}{A_1B_1}\cup \overset{\frown}{A_2B_2}$, we can draw a perpendicular from $x$ to edge $e_1$, intersecting $e_1$ at point $C$.  The distance from $x$ to $e_1$ is then denoted by $\tilde{h}_1$, as shown in Figure \ref{e1-discontinuity}.
We define the ray $\overrightarrow{xC}$ as the polar axis, with an initial angle of $0$ and the counterclockwise direction taken as the positive angular orientation. The intersection points can then be parameterized as $\tilde{A}=(r,\tilde{\alpha}(r))$ and $\tilde{B}=(r,\tilde{\beta}(r))$, as illustrated in Figure \ref{e1-discontinuity}.
\begin{figure}
      \centering
    \begin{tabular}{cc}
    \subfigure[The discontinuity from the edge $e_1$]{
    \label{e1-discontinuity}
    \includegraphics[width=0.36\textwidth]{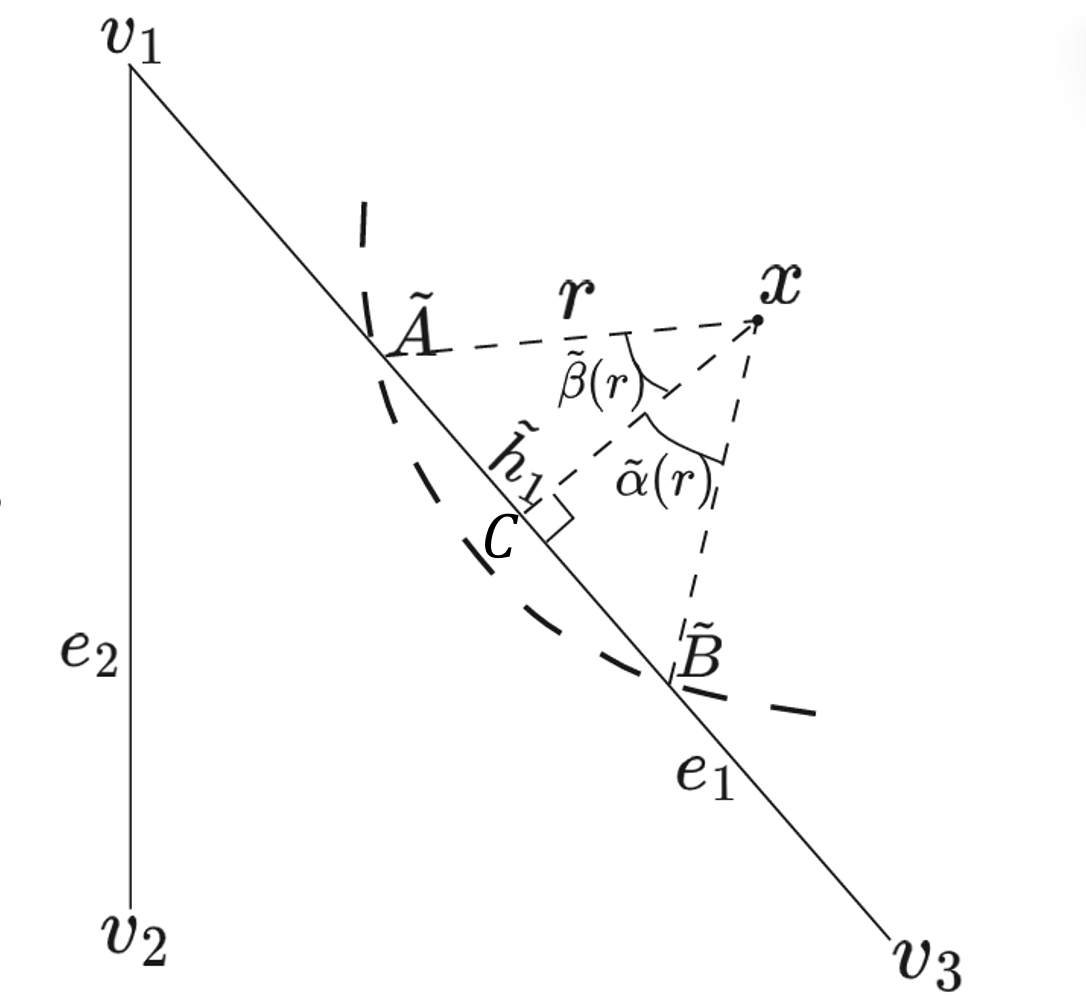}}\hspace{0em} &
    \subfigure[The range of sensors for edge $e_1$]{
    \label{range-e1}
\includegraphics[width=.36\textwidth]{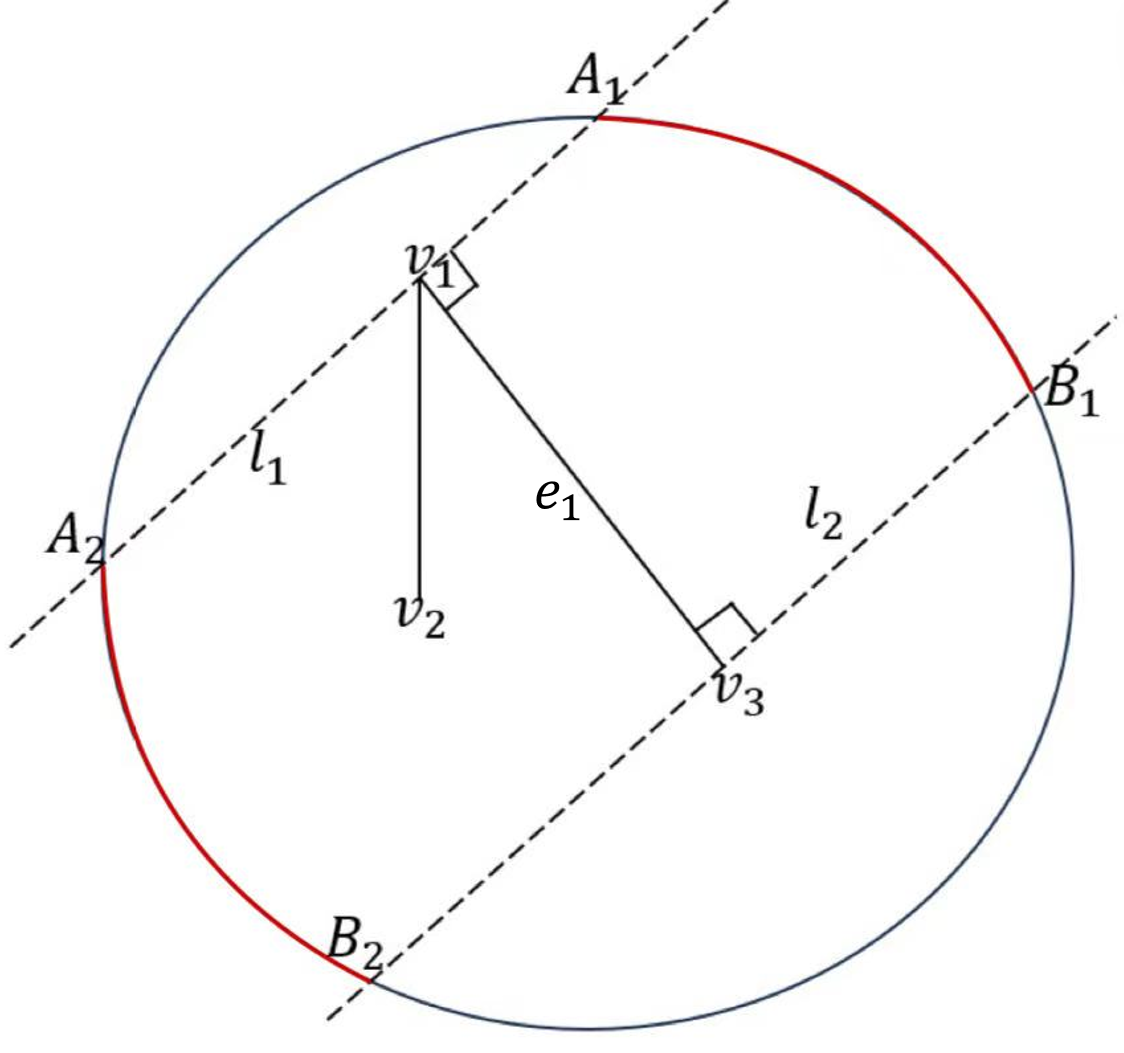}
}
    \end{tabular}
    \caption{Illustrations for the \textbf{Case 2}.}
    \label{Simplification-polygon-edge}
\end{figure}
For simplicity, we denote 
\begin{equation*}
\begin{aligned}
    \Sigma_1^{(2)}:=(\tilde{h}_1-\epsilon_{\tilde{h}_1},\tilde{h}_1)\quad \mbox{and}\quad
    \Sigma_2^{(2)}:=(\tilde{h}_1,\tilde{h}_1+\epsilon_{\tilde{h}_1}),
\end{aligned}
\end{equation*}
where $\epsilon_{\tilde{h}_1}>0$ is a sufficiently small number.
A direct and straightforward computation yields 
\begin{equation}
\tilde{\alpha} (r)=-\tilde{\beta} (r)=\tilde{\alpha}^{\prime}(r)=-\tilde{\beta}^{\prime}(r)=0,\quad r\in \Sigma_1^{(2)},
\label{polygon_edge_alpha_beta2}
\end{equation}
\begin{equation}
\tilde{\alpha} (r)=-\tilde{\beta} (r)=\arccos\left(\frac{\tilde{h}_1}{r}\right),\quad r\in \Sigma_2^{(2)},
\nonumber
\end{equation}
and
\begin{equation}
\begin{aligned}
\tilde{\alpha}^{\prime}(r)=-\tilde{\beta}^{\prime}(r)=\frac{d}{dr}\left(\arccos \left(\frac{\tilde{h}_1}{r}\right)\right)=\frac{\tilde{h}_1}{r\sqrt{r^2-\tilde{h}_1^2}},\quad r\in \Sigma_2^{(2)}
\end{aligned}
\label{polygon_edge_alpha_beta1}
\end{equation}
Differentiating both sides of the equation \eqref{I(r)} with respect to $r$, we have
\begin{equation}
I_x'(r)=\begin{cases}
\int_{\tilde{\beta}(r)}^{\tilde{\alpha}(r)}\frac{\pa S(r,\theta)}{\pa r}d\theta+\tilde{\alpha}'(r)S(r,\tilde{\alpha}(r))-\tilde{\beta}'(r)S(r,\tilde{\beta}(r)),&\quad r\in \Sigma_2^{(2)}; \\
0,&\quad r\in \Sigma_1^{(2)}.
\label{calculate-edge-I'}
\end{cases}
\end{equation}
Therefore, using equations \eqref{polygon_edge_alpha_beta2}-\eqref{polygon_edge_alpha_beta1} and substituting $r=\tilde{h}_1$ into \eqref{calculate-edge-I'}, we have 
\begin{equation}
\begin{aligned}
 \llbracket I_x^{\prime} \rrbracket(\tilde{h}_1)
 =2 S_{\tilde{h}_1} \lim_{r\rightarrow \tilde{h}_1^{+}}\frac{\tilde{h}_1}{r\sqrt{r^2-\tilde{h}_1^2}} 
 =S_{\tilde{h}_1}\sqrt{\frac{2}{\tilde{h}_1}}\lim_{\epsilon_{\tilde{h}_1}\rightarrow 0^{+}}\frac{1}{\sqrt{\epsilon_{\tilde{h}_1}}},
\end{aligned}
\label{polygon-dis1}
\end{equation}
where $S_{\tilde{h}_1}=S|_{e_1 \cap C(x,\tilde{h}_1)}$ and $\epsilon_{\tilde{h}_1}= r-\tilde{h}_1$.

% When the observation point $x\in C(0,R)$ lies within a strip perpendicular to an edge of the polygon, the contribution of points on that edge to the discontinuity of $I_x^{\prime}$  also becomes unbounded.
% For example, in Figure \ref{kuandai}, two lines $l_1$ and $l_2$ are constructed perpendicular to the edge $\overline{v_1v_2}$ at the vertices $v_1$ and $v_2$, respectively. The intersection of the strip between these lines with the circle $C(0,R)$ defines the set of sensor locations that induce singularities along the edge $\overline{v_1v_2}$(red arcs in Figure \ref{kuandai}). Let $h_1$ denote the distance from $x$ to the edge $e_1$.  
% Then, we can compute the functions $\tilde{\alpha} (r)$ and $\tilde{\beta} (r)$ as follows:
% \begin{equation}
% \tilde{\alpha} (r)=-\tilde{\beta} (r)=\arccos\left(\frac{\tilde{h}_1}{r}\right).
% \nonumber
% \end{equation}
% Differentiating yields:
% \begin{equation}
% \begin{aligned}
% \tilde{\alpha}^{\prime}(r)=-\tilde{\beta}^{\prime}(r)=\frac{d}{dr}\left(\arccos \left(\frac{\tilde{h}_1}{r}\right)\right)=\frac{\tilde{h}_1}{r\sqrt{r^2-\tilde{h}_1^2}},
% \end{aligned}
% \label{polygon_edge_alpha_beta}
% \end{equation}
% which leads to the jump magnitude:

\textbf{Case 3: the circle $\partial B_{r_0}(x)$ passes through exactly one vertex of the polygon $\Omega_m$.}

\begin{figure}
      \centering
    \begin{tabular}{cc}
    \subfigure[The discontinuity from the vertex $v_1$]{
    \label{v1-discontinuity}
    \includegraphics[width=0.36\textwidth]{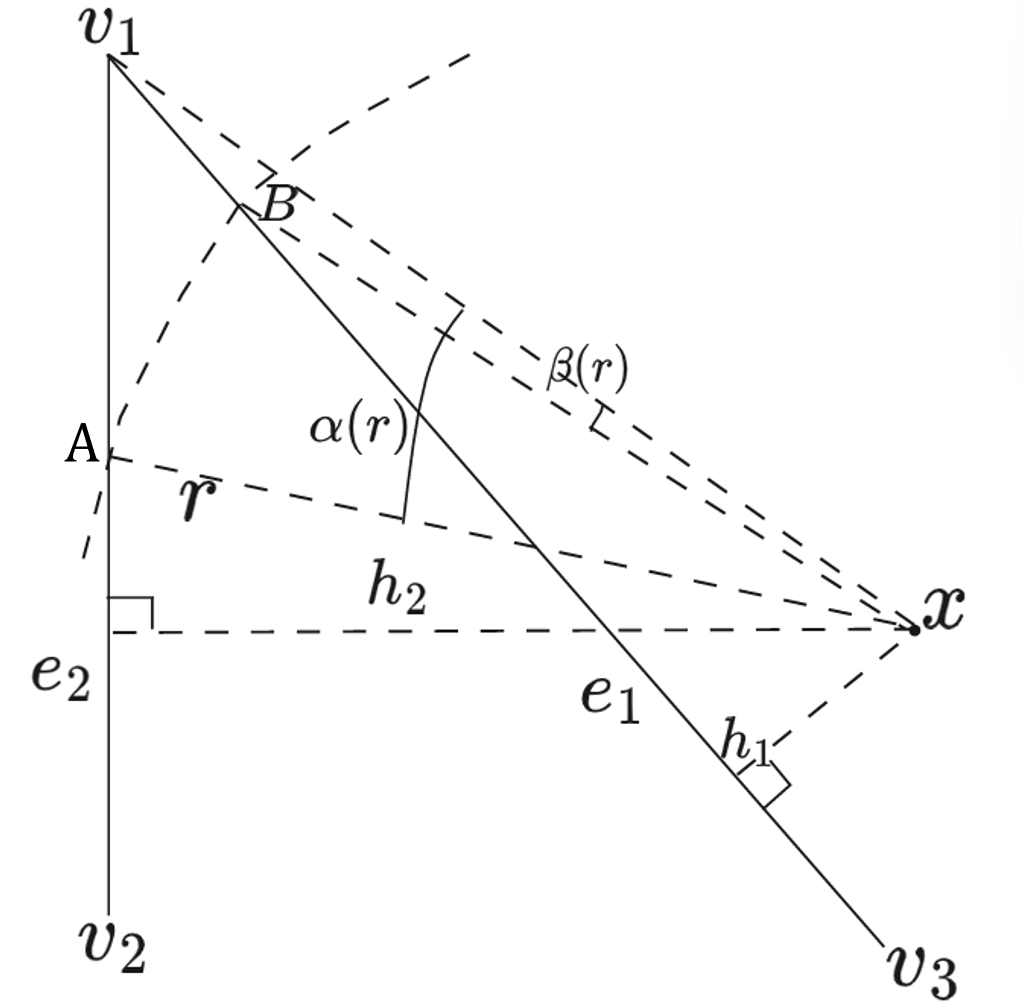}}\hspace{0em} &
    \subfigure[The range of sensors for vertex $v_1$]{
    \label{range-v1}
\includegraphics[width=.36\textwidth]{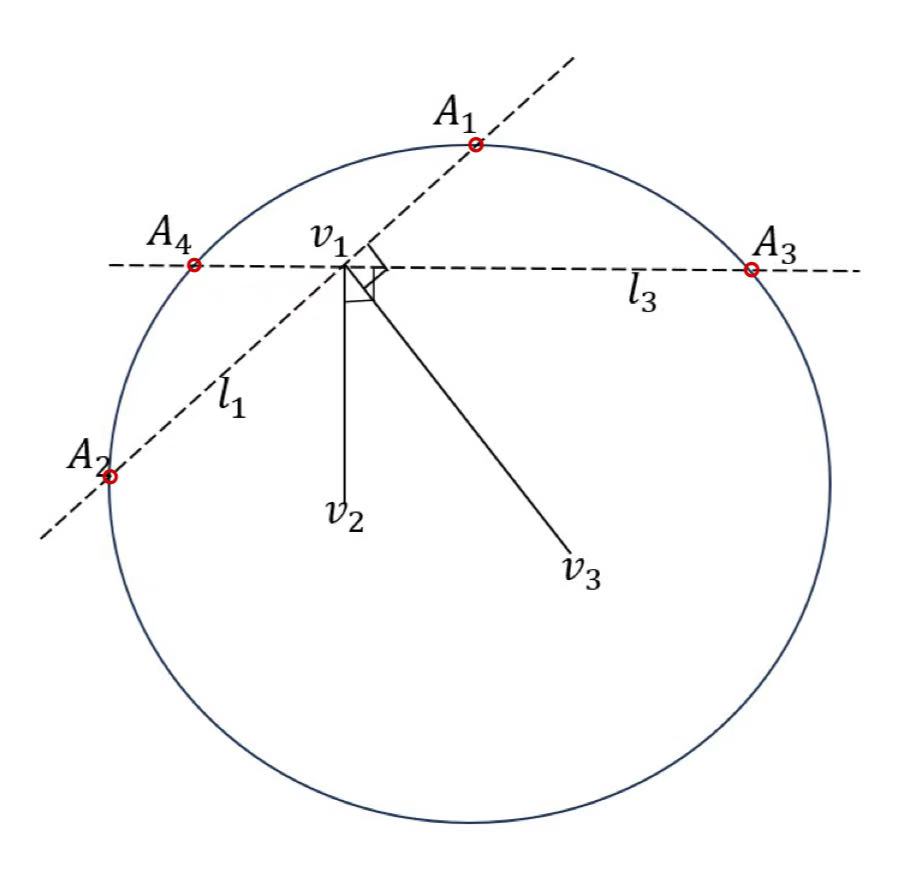}
}
    \end{tabular}
    \caption{Illustrations for the \textbf{Case 3}.}
    \label{Simplification-polygon-vertex1}
\end{figure}
Denote such a vertex by $v_1$. 
As shown in Figure \ref{range-v1}, at the vertex $v_1$, we construct perpendicular lines $l_1$ and $l_3$ to the edges $\overline{v_1v_3}$ and $\overline{v_1v_2}$, respectively, and define $\Gamma_{v_1}:=(l_1\cup l_3)\cap \pa B_R(0)=\{A_1,A_2,A_3,A_4\}$. 
We claim that if there exist $x$ and $r_0$ such that for some $m\in \{1,2,\ldots,M\}$, Case 3 holds, then the jump $\llbracket I_x^{\prime} \rrbracket(r_0)$ is bounded and non-zero.
% In this case, we claim that $\forall x\in \pa B_R(0)\backslash \Gamma_{v_1}$, there exists $r_0>0$ such that the jump $\llbracket I_x^{\prime} \rrbracket(r_0)$ is bounded and nonzero.
For any $x\in \pa B_R(0)\backslash \Gamma_{v_1}$, we define the ray $\overrightarrow{xv_1}$ as the polar axis with angle $0$, taking the counterclockwise direction as the positive angular direction.
Then the points $A$ and $B$ can be parameterized as $A=(r,\alpha(r))$ and $B=(r,\beta(r))$. Let $h_1$ and $h_2$ denote the distances from point $x$ to edges $\overline{v_1v_3}$ and $\overline{v_1v_2}$, respectively, with $h_1<h_2$. 
For simplicity, we define 
\begin{equation*}
\begin{aligned}
    \Sigma_1^{(3)}:=(|x-v_1|-\epsilon_1,|x-v_1|)\quad\mbox{and}\quad
    \Sigma_2^{(3)}:=(|x-v_1|,|x-v_1|+\epsilon_1),
\end{aligned}
\end{equation*}
where $\epsilon_1$ denotes a sufficiently small positive constant.
Then we compute the angular functions $\alpha(r)$, $\beta(r)$ and their derivatives $\alpha^{\prime}(r)$, $\beta^{\prime}(r)$, 
% \begin{equation}
% \alpha (r)=\arccos\left(\frac{h_2}{|x-v_1|}\right)-\arccos \left(\frac{h_2}{r}\right),\quad |x-v_1|-\epsilon<r<|x-v_1|,
% \nonumber
% \end{equation}
% \begin{equation}
% \alpha (r)=0,\quad |x-v_1|<r<|x-v_1|+\epsilon,
% \nonumber
% \end{equation}
% \begin{equation}
% \beta (r)=\arccos\left(\frac{h_1}{|x-v_1|}\right)-\arccos \left(\frac{h_1}{r}\right),\quad |x-v_1|-\epsilon<r<|x-v_1|,
% \nonumber
% \end{equation}
% \begin{equation}
% \beta (r)=0,\quad |x-v_1|<r<|x-v_1|+\epsilon,
% \nonumber
% \end{equation}
\begin{equation}
\alpha (r)=\arccos\left(\frac{h_2}{|x-v_1|}\right)-\arccos \left(\frac{h_2}{r}\right),\,\alpha^{\prime}(r)
=\frac{h_2}{r\sqrt{r^2-h_2^2}},\quad r\in \Sigma_1^{(3)},
\label{polygon_vertex_alpha_beta1}
\end{equation}
%\begin{equation}
%\alpha (r)=0,\,\alpha^{\prime}(r)=0,\quad r\in \Sigma_2^{(3)},
%\label{polygon_vertex_alpha_beta11}
%\end{equation}
\begin{equation}
\beta (r)=\arccos\left(\frac{h_1}{|x-v_1|}\right)-\arccos \left(\frac{h_1}{r}\right),\,\beta^{\prime}(r)=\frac{h_1}{r\sqrt{r^2-h_1^2}},\quad r\in \Sigma_1^{(3)},
\label{polygon_vertex_alpha_beta2}
\end{equation}
and
\begin{equation}
\alpha (r)= \alpha^{\prime}(r)=\beta (r)= \beta^{\prime}(r)=0,\quad r\in \Sigma_2^{(3)}.
\label{polygon_vertex_alpha_beta22}
\end{equation}
Next, differentiating both sides of equation \eqref{I(r)}, we obtain
\begin{equation}
I_x'(r)=\begin{cases}
\int_{\beta(r)}^{\alpha(r)}\frac{\pa S(r,\theta)}{\pa r}d\theta+\alpha'(r)S(r,\alpha(r))-\beta'(r)S(r,\beta(r)),&\quad \quad r\in \Sigma_1^{(3)};\\
0,&\quad \quad r\in \Sigma_2^{(3)}.
\label{calculate-vertex-I'}
\end{cases}
\end{equation}
Finally, 
we estimate $|\llbracket I_x^{\prime} \rrbracket(|x-v_1|)|$ from equations \eqref{polygon_vertex_alpha_beta1}- \eqref{calculate-vertex-I'} as follows
% \begin{equation}
% \begin{aligned}
\be
 |\llbracket I_x^{\prime} \rrbracket(|x-v_1|)|&=&\frac{h_2|S(v_1)|}{|x-v_1|\sqrt{|x-v_1|^2-h_2^2}}-\frac{h_1|S(v_1)|}{|x-v_1|\sqrt{|x-v_1|^2-h_1^2}}\cr
 &&\leq \frac{2|S(v_1)|\max (h_1,h_2)}{|x-v_1|\sqrt{|x-v_1|^2-(\max (h_1,h_2))^2}}\cr
 &&< \infty,
% \end{aligned}
\label{polygon-dis2}
% \end{equation}
\en
where $|x-v_1|>h_2>h_1$ for the selected $x$.

\textbf{Case 4: the circle $\partial B_{r_0}(x)$ passes through exactly one vertex of the polygon $\Omega_m$ and is tangent to exactly one edge of $\Omega_m$ at that vertex.}
    
\begin{figure}
      \centering
    \begin{tabular}{cc}
    \subfigure[The discontinuity from the vertex $v_1$]{
    \label{v1--e1-discontinuity}
    \includegraphics[width=0.36\textwidth]{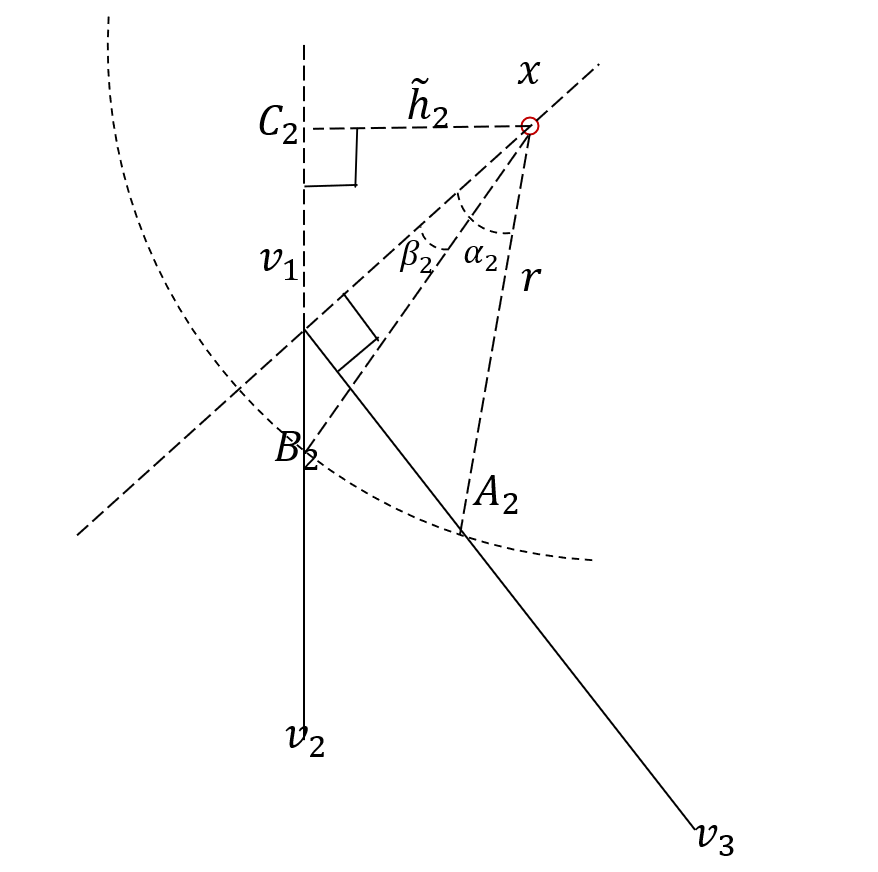}}\hspace{0em} &
    \subfigure[The range of sensors for vertex $v_1$]{
    \label{range-v1-e1}
\includegraphics[width=.36\textwidth]{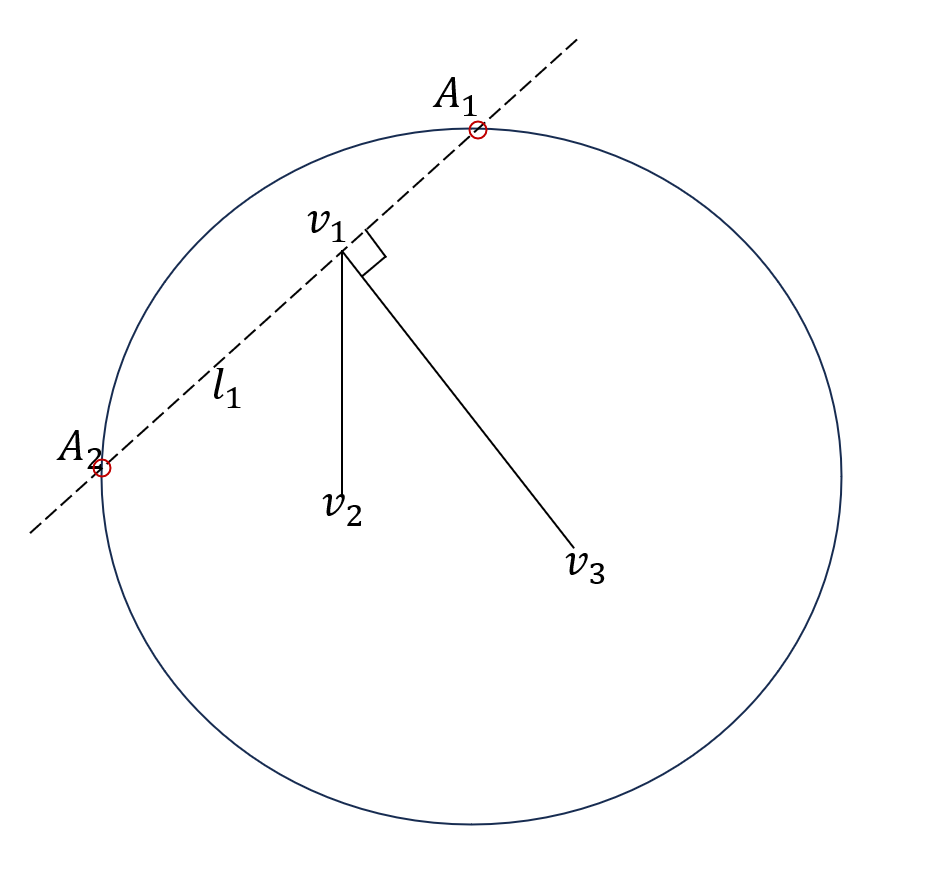}
}
    \end{tabular}
    \caption{Illustrations for the \textbf{Case 4}.}
    \label{Simplification-polygon-vertex2}
\end{figure}

We finally consider the case where both the vertex $v_1$ and the edge $\ov{v_1v_3}$ simultaneously contribute to the discontinuity of $I_x^{\prime}$.
Following the notations in Case 3, as shown in Figure \ref{range-v1-e1}, we observe that both the circles $\pa B_{|A_1-v_1|}(A_1)$ and $\pa B_{|A_2-v_1|}(A_2)$ are tangent to the edge $\ov{v_1v_3}$  at the vertex $v_1$. Without loss of generality, taking $x=A_1$. We claim that the jump $\llbracket I_x^{\prime} \rrbracket(|x-v_1|)$ consists of a bounded term and an unbounded term. As shown in Figure 
\ref{v1--e1-discontinuity}, we drop a perpendicular to the line containing $\ov{v_1v_2}$, denote by $C_2$ the foot and define by $\tilde{h}_2:=|x-C_2|$ the corresponding distance.
We still follow the definition of the polar axis in Case 3, and define
\begin{equation*}
\begin{aligned}
    \Sigma_1^{(4)}:=(|x-v_1|-\epsilon_2,|x-v_1|)\quad\mbox{and}\quad
    \Sigma_2^{(4)}:=(|x-v_1|,|x-v_1|+\epsilon_2),
\end{aligned}
\end{equation*}
where $\epsilon_2$ is a sufficiently small number.
A straightforward computation yields
\begin{equation}
\alpha_2(r)=\alpha_2^{\prime}(r)=\beta_2(r)=\beta_2^{\prime}(r)=0,\quad r\in \Sigma_1^{(4)}.
\label{v_1-e_1-alphabeta-left}
\end{equation}
\begin{equation}
\alpha_2(r)=\arccos{\frac{|x-v_1|}{r}},\,\alpha_2^{\prime}(r)=\frac{|x-v_1|}{r\sqrt{r^2-|x-v_1|^2}},\quad r\in \Sigma_2^{(4)}
\label{v_1-e_1-alpha-right}
\end{equation}
and
\begin{equation}
\beta_2(r)=\arccos{\frac{\tilde{h}_2}{r}}-\arccos{\frac{\tilde{h}_2}{|x-v_1|}},\, \beta_2^{\prime}(r)=\frac{\tilde{h}_2}{r\sqrt{r^2-\tilde{h}_2^2}},\quad r\in \Sigma_2^{(4)},
\label{v_1-e_1-beta-right}
\end{equation}
Differentiating both sides of equation \eqref{I(r)}, we have
\begin{equation}
I_x'(r)=\begin{cases}
\int_{\beta_2(r)}^{\alpha_2(r)}\frac{\pa S(r,\theta)}{\pa r}d\theta+\alpha_2'(r)S(r,\alpha_2(r))-\beta_2'(r)S(r,\beta_2(r)),& \quad r\in \Sigma_2^{(4)};\\
0,& \quad r\in \Sigma_1^{(4)}.
\label{calculate-vertex-special-I'}
\end{cases}
\end{equation}
Therefore, using equations \eqref{v_1-e_1-alphabeta-left}-\eqref{calculate-vertex-special-I'} and the fact that $|x-v_1|>\tilde{h}_2$, we have
% \begin{equation}
% \begin{aligned}
\be
\llbracket I_x^{\prime} \rrbracket(|x-v_1|)
&=&S(v_1) \Bigg[\lim_{r\rightarrow |x-v_1|^{+}}\frac{|x-v_1|}{r\sqrt{r^2-|x-v_1|^2}}-\frac{\tilde{h}_2}{|x-v_1|\sqrt{|x-v_1|^2-\tilde{h}_2^2}} \Bigg]\cr
&=&S(v_1)\Bigg[\frac{1}{\sqrt{2|x-v_1|}}\lim_{\epsilon_{2}\rightarrow 0^{+}}\frac{1}{\sqrt{\epsilon_{2}}}-\frac{\tilde{h}_2}{|x-v_1|\sqrt{|x-v_1|^2-\tilde{h}_2^2}}\Bigg],
% \end{aligned}
\label{vertex-edge-dis}
% \end{equation}
\en
where $\epsilon_2:= r-|x-v_1|$. This implies that $I_x'$ does not exist at $r_0:=|x-v_1|$.
The proof is complete.
\end{proof}

As demonstrated in \eqref{annular-dis1}, \eqref{annular-dis2}, \eqref{polygon-dis1}, \eqref{polygon-dis2}, and \eqref{vertex-edge-dis}, our analysis reveals a fundamental distinction: annular sources generate unbounded discontinuities in $I_x^{\prime}$, whereas polygonal sources produce unbounded jumps at edges but only bounded discontinuities or a combination of bounded and unbounded terms at vertices.

With the help of Lemma \ref{discontinuity-I'} and the corresponding results in \cite{2025extended-sources-acoustic}, we emphasize that the discontinuity analyses of the function $I_x^{\prime}$ in near-field and far-field configurations are different. 
This distinction thus reveals a key advantage of near-field data: it allows for more accurate characterization of polygonal edges and superior reconstruction results, with verification provided in subsequent uniqueness proofs and numerical experiments.
In the next two subsections, we analyze two special cases: annular sources (subsection \ref{2.1.1annular}) and polygonal sources (subsection \ref{2.1.2polygon}). Subsequently, we examine the hybrid case combining annular and polygonal structures (subsection \ref{2.1.3mixed}).

For notation convenience, let $\mathcal{J}$ be the space containing all the extended functions $J: \R^3\rightarrow \R\cup\{\infty\}$
\begin{equation*}
J(a,b,c)=a+b\cdot\lim_{\epsilon\rightarrow0^+}1\big/\sqrt{\epsilon}+c\cdot\lim_{\epsilon\rightarrow0^-}1\big/\sqrt{-\epsilon},\quad (a,\,b,\,c)\in \R^3.
\end{equation*}
%It is clear that the vectors of the form \eqref{asymptotic-expression} belong to the vector space $\mathcal{J}$.
Furthermore, we define four operators $P_i: \mathcal{J}\rightarrow\R$, $i=1,2$,  $F_{\text{vertex}}: \R^2\rightarrow\N$ and $F_{\text{annular}}: \R^2\times\R\rightarrow \N$, respectively, by
\begin{equation*}
P_1(J):=a,\quad P_2(J):=b+c,\quad J\in\mathcal{J},
\end{equation*}
\begin{equation*}
F_{\text{vertex}}(y):=\# \{x\in \Gamma_L \mid P_1(\llbracket I_x^{\prime} \rrbracket(|x-y|))\neq 0\},\quad y\in \R^2,
\end{equation*}
and
\begin{equation}
F_{\text{annular}}(y,d):=\sum_{j=0,1}\#\left\{x\in \Gamma_L \mid \llbracket I_x^{\prime} \rrbracket (|x-y|+(-1)^jd)\neq 0 \right\}, \quad (y,d)\in \R^2\times\R.
\label{F-annular}
\end{equation}

% {\color{red}In the next two subsections, we analyze two special cases: annular sources (Subsection \ref{2.1.1annular}) and polygonal sources (Subsection \ref{2.1.2polygon}). Subsequently, we examine the hybrid case combining annular and polygonal structures (Subsection \ref{2.1.3mixed}). The uniqueness proofs rely on Lemma  \ref{discontinuity-I'}, with all derivations extending the Helmholtz equation framework in \cite{2025extended-sources-acoustic}. However, near-field measurements introduce additional complexities, particularly in modeling circle-circle intersections. }

\subsubsection{Annular sources}\label{2.1.1annular}
In this subsection, we consider a special case
\ben
\Omega_m=\overline{B_{R_m}(y_m)\ba B_{r_m}(y_m)},\quad m=1,2,\ldots M,
\enn
where $0\leq r_m<R_m$. We also denote by $d_m=r_m, R_m$.
%The boundary of $\Omega$ is composed of the circles $C(y_m,r_m)$ and $C(y_m,R_m)$, $m=1,2,\dots,M$. Before establishing the uniqueness theorem for such annular sources, we first prove the following lemma, which states that certain circular arcs $C(x, r)$ centered at the sensor $x\in C(0,R)$ and the radius $r$ can be uniquely determined by the corresponding scattered field measurements $\{u^s(x,k)\mid k>0\}$.
Lemma \ref{discontinuity-I'} immediately implies the following important uniqueness result using the multi-frequency scattering fields taken at a single sensor.
\begin{lemma}
Let Assumption \ref{assumptiom} hold. We also assume that $\Omega=\bigcup\limits_{m=1}^{M}\overline{B_{R_m}(y_m)\ba B_{r_m}(y_m)}$. Fix a sensor $x\in \pa B_R(0)$.
If there exists a unique index $m_0\in \{1,2,\ldots,M\}$ such that the circular $\pa B_r(x)$ is exclusively tangent to either the inner boundary $\pa B_{r_{m_0}}(y_{m_0})$ or the outer boundary $\pa B_{R_{m_0}}(y_{m_0})$ of the annular component $\Omega_{m_0}$, then we have
\begin{equation*}
P_2\left(\llbracket I_x^{\prime} \rrbracket(r)\right)\neq0.
\end{equation*}
\label{lemma-annulas}
\end{lemma}
\begin{proof}
This claim follows directly from Case 1 of Lemma \ref{discontinuity-I'}.
\end{proof}
% \begin{proof}
% {\color{red}delete?} The tangent condition implies that $C(x,r)$ intersects $\partial \Omega$ only along the boundary of $\Omega_{m_0}$. By Lemma \ref{discontinuity-I'}, the resulting discontinuity in $I'(r)$ is not canceled by contributions from other annular components $\Omega_m$, $\forall m\neq m_0$. Thus, the circle arcs $C(x,r)$ can be uniquely determined by $\{u^s(x,k)\mid k>0\}$.
% \end{proof}
The following theorem gives a characterization of the boundary $\pa \Om$ using scattering fields taken at finitely many sensors.
\begin{theorem}
\label{Annuluses-proof}
The support $\Om$ of annular sources can be uniquely determined by the multi-frequency sparse scattered fields
$\{u^s(x,k)\mid x\in \Gamma_L,\, k\in \left[k_{-}, k_{+}\right]\}$,  
provided that
\begin{equation}
L> 16M-8.
\label{Annuluses-condition}
\end{equation}
\end{theorem}
\begin{proof}
% For any fixed $x\in \Gamma_L$, we define the function $F_{\text{annular}}: \R^2\times\R\rightarrow \N$ by
% \begin{equation}
% F_{\text{annular}}(y,d)=\#\left\{x\in \Gamma_L \mid \llbracket I_x^{\prime} \rrbracket (|x-y|+d)\neq 0 \right\}+\#\left\{x\in \Gamma_L \mid \llbracket I_x^{\prime} \rrbracket (|x-y|-d)\neq 0 \right\}.
% \nonumber
% \end{equation}
For any $\pa B_{d_m}(y_m)\subset \pa \Omega$, there exist at least $2L$ circles centered at $x\in \Gamma_L$ that are tangent to the circle $\pa B_{r_m}(y_m)$ or $\pa B_{R_m}(y_m)$.
If $P_2(\llbracket I_x^{\prime} \rrbracket(|x-y_m|+d_m))=0$ or $P_2(\llbracket I_x^{\prime} \rrbracket(|x-y_m|-d_m))=0$, then there exist $\tilde{m}\neq m$ and $\pa B_{d_{\tilde{m}}}(y_{\tilde{m}})$ such that 
\begin{equation}
|x-y_m|\pm d_m=|x-y_{\tilde{m}}|\pm d_{\tilde{m}}, \quad x\in \R^2.
\label{annulues-P2-hyperbolas}
\end{equation}
This implies that the sensor $x$ must lie on one of the two pairs of hyperbolas. Each hyperbola intersects $\pa B_R(0)$ at four points. The boundary $\pa \Omega \Big\backslash \Big(\pa B_{r_{m}}(y_{m})\bigcup\pa B_{R_{m}} (y_{m})\Big)$ contains $2(M-1)$ circles, yielding 
\begin{equation}
F_{\text{annular}}(y_m,d_m)\geq 2L-16(M-1).
\nonumber
\end{equation}

For any $\pa B_d(y)\not\subset \pa \Omega$, if $P_2(\llbracket I_x^{\prime} \rrbracket(|x-y|+d))\neq0$ or $P_2(\llbracket I_x^{\prime} \rrbracket(|x-y|-d))\neq0$, then there exist $m_0\in \{1,2,\ldots, M\}$ and $\pa B_{d_{m_0}}(y_{m_0})\subset \pa \Omega$ such that 
\begin{equation}
|x-y|\pm d=|x-y_{m_0}|\pm d_{m_0},\quad x\in \R^2.
\nonumber
\end{equation}
Therefore, we deduce 
\begin{equation}
F_{\text{annular}}(y,d)\leq 16M.
\nonumber
\end{equation}
Under the condition \eqref{Annuluses-condition}, we have
\begin{equation}
F_{\text{annular}}(y_m,d_m)> F_{\text{annular}}(y,d),\quad \forall \pa B_{d_m}(y_m)\subset \pa \Omega,\,\pa B_d(y)\not\subset \pa \Omega,
\nonumber
\end{equation}
which implies that all circles $\pa B_{d_m}(y_m)\subset \pa \Omega$ can be uniquely determined from the data.
The theorem now follows exactly from similar arguments in the proof of Theorem 3.2 in \cite{2025extended-sources-acoustic}.

\end{proof}

\subsubsection{Polygonal sources}\label{2.1.2polygon}
This subsection is dedicated to the case where the source support $\Om$ is a union of simply connected polygons:
\ben
\Om=\bigcup\limits_{m=1}^{M}\Om_m.
\enn
Let $v_1,v_2,\ldots, v_N$ and $e_1,e_2,\ldots,e_N$ denote, respectively, the vertices and edges of $\Om$. 
Due to the unbounded singularities introduced by the edges of the polygon, our proof differs from the sources of the Helmholtz equation \cite{2025extended-sources-acoustic}. It is precisely this distinct approach that allows for the superior reconstruction performance observed in our numerical experiments.
% Unlike the sources for the Helmholtz equation \cite{2025extended-sources-acoustic}, the edges present an additional singularity of the indicator here, which poses theoretical challenges for the reconstruction of the boundary. Our subsequent proof overcomes this challenge, and this approach also offers advantages for numerical reconstruction.

We begin with a characteristic of the circles passing through the vertices using the multi-frequency scattering fields at a single sensor.

%The following lemma establishes that certain circular arcs $C(x, r)$ centered at the sensor $x\in C(0,R)$ with radius $r$, can be uniquely determined from the scattered field data $\{u^s(x,k)\mid k>0\}$.

\begin{lemma}
Let Assumption \ref{assumptiom} hold. For a fixed sensor $x\in \pa B_R(0)$ and some $r>0$, if there exists a unique index $n_0\in \{1,2,\ldots,N\}$ such that the circle $\pa B_r(x)$ passes through the vertex $v_{n_0}$, then 
\begin{equation*}
P_1\left(\llbracket I_x^{\prime} \rrbracket(r)\right)\neq 0.
\end{equation*}
\label{lemma-polygon-vertex}
\end{lemma}
\begin{proof}
The proof can be obtained directly from Case 2 of Lemma \ref{discontinuity-I'}.
\end{proof}

\begin{theorem}\label{polygon_corners}
The vertices $v_1,v_2,\ldots, v_N$ of $\Om$ can be uniquely determined from the multi-frequency sparse
scattered fields
$\{u^s(x,k)\mid x\in \Gamma_L,\, k\in \left[k_{-}, k_{+}\right]\}$,  
provided that
\begin{equation}
L> 4N-2.
\label{polygons-condition}
\end{equation}.
\end{theorem}
\begin{proof}
% For any fixed $x\in \Gamma_L$, we define the function $F_{\text{vertex}}: \R^2\rightarrow\N$ by 
% \begin{equation}
% F_{\text{vertex}}(y)=\# \{x\in \Gamma_L \mid P_1(\llbracket I_x^{\prime} \rrbracket(|x-y|))\neq 0\}.
% \nonumber
% \end{equation}
For any $v_0\in \{v_1,v_2,\ldots, v_N\}$, if $P_1\left(\llbracket I_x^{\prime} \rrbracket(|x-v_0|)\right)=0$, then there exists some $\tilde{v}\in\{v_1,v_2,\dots,v_N\}\backslash \{v_0\}$ such that $|x-v_0|=|x-\tilde{v}|$.
This implies that the sensor $x$ must lie on the perpendicular bisector of the segment $\overline{v\tilde{v}}$. Since all sensors $x$ lie on the circle $\pa B_R(0)$, we then deduce that 
\begin{equation*}
    F_{\text{vertex}}(v_0)\geq L-2(N-1).
\end{equation*}

For any $v\in \R^2 \backslash \{v_1,v_2,\ldots, v_N\}$, if $P_1(\llbracket I_x^{\prime} \rrbracket(|x-v|))\neq 0$, then there exists some $\overline{v}\in \{v_1,v_2,\ldots, v_N\}$ such that $|x-\overline{v}|=|x-v|$. Hence, we have
\begin{equation*}
   F_{\text{vertex}}(v)\leq 2N.
\end{equation*}

Under the condition \eqref{polygons-condition}, we have 
\begin{equation*}
    F_{\text{vertex}}(v_0)> F_{\text{vertex}}(v),\quad\forall v_0\in \{v_1,v_2,\ldots, v_N\},\, v\in \R^2 \backslash \{v_1,v_2,\ldots, v_N\},
\end{equation*}
which implies the statement of the theorem. The proof is complete.
\end{proof}

Once the vertices are determined, the uniqueness of the edges can be established as follows. 
% We emphasize that the following proof is different from the corresponding proof in \cite{2025extended-sources-acoustic}. In fact, due to changes in the singularity analysis of edges, we achieve a better characterization of edges.

\begin{theorem}\label{polygon_edges}
Given all the vertices $v_1, v_2, \cdots, v_N$, there exists a set  $\Gamma_{\tilde{L}}$ consisting of finite sensors located on $\pa B_R(0)$, such that all edges $e_1,e_2,\ldots,e_N$ and the source values $S(v_1),S(v_2),\ldots,S(v_N)$ can be uniquely determined from the scattered field data
$\{u^s(x,k)\mid  x\in \Gamma_{\tilde{L}},\, k\in \left[k_{-}, k_{+}\right]\}$.
\end{theorem}
\begin{proof}
Since the positions of the $N$ vertices are known, there are finitely possible combinations of edges, which we denote as the set $E_{\tilde{L}}$.
For any fixed $x\in \Gamma_L$ and $e_0:=\overline{v_1v_2}\in E_{\tilde{L}}$, if there exist $\tilde{e}\in E_{\tilde{L}}\backslash\{e_0\}$ such that
\begin{equation}
\operatorname{dist}(x,e_0)=\operatorname{dist}(x,\tilde{e}),\quad x\in \R^2,
\label{dist-distinguish}
\end{equation}
then the sensor $x$ must lie on the angle bisector formed by the extensions of the two edges $e_0$ and $\tilde{e}$.
If the edges $e_0$ and $\tilde{e}$ are parallel, the bisector is defined as the equidistant line.
Note that each angle bisector intersects the circle $\pa B_R(0)$ at two distinct points, thus we obtain that at most finite sensors satisfy equation \eqref{dist-distinguish}.
Therefore, there must exist some $x_0\in \Gamma_{e_0}$ such that 
\begin{equation}
\operatorname{dist}(x_0,e_0)\neq \operatorname{dist}(x_0,e), \quad \forall e\in E_{\tilde{L}}\backslash\{e_0\},
\nonumber
\end{equation}
where $\Gamma_{e_0}$ is the strip region corresponding to the edge $e_0$, defined similarly to Case 2 in Lemma \ref{discontinuity-I'}. Furthermore, based on the singularity analysis of polygon edges in Lemma \ref{discontinuity-I'}, if $P_2(\llbracket I_x^{\prime} \rrbracket(\operatorname{dist}(x,e_0)))\neq0$, then $e_0$ is a genuine edge of the polygon. Otherwise, if $P_2(\llbracket I_x^{\prime} \rrbracket(\operatorname{dist}(x,e_0)))=0$, then $e_0$ is not a true edge. Therefore, we need at most $\tilde{L}$ sensors to uniquely determine the $N$ edges of the polygon.

Once the $N$ edges are determined, we proceed to recover the source values at the vertices. 
For any vertex $v_i\in \{v_1,v_2,\ldots, v_N\}$, choose $x_i\in \pa B_R(0)$ such that $|x_i-v_i|\neq |x_i-v_j|$, $\forall v_j\neq v_i$, $j=1,2,\ldots, N$.
That is to say, the circle $\pa B_{|x_i-v_i|}(x_i)$ passes only through the vertex $v_i$.
From the proofs of Cases 2, 3, and 4 in Lemma \ref{discontinuity-I'}, it follows that $\llbracket I_{x_i}^{\prime} \rrbracket(|x_i-v_i|)$ must contain a non-zero bounded term related to $S(v_i)$. Therefore, $S(v_i)$ can be recovered by computing the value of this bounded term.
\end{proof}

\subsubsection{A combination of annular and polygonal sources}\label{2.1.3mixed}
Finally, we consider the identification of mixed-structured sources, where each $\Omega_m$ is either an annulus or a simply connected polygon. 
% To identify mixed-structured sources, we begin by analyzing the respective contributions of annular and polygonal sources to the discontinuities of $I_x^{\prime}$, respectively.

\begin{theorem}\label{structure-sources}
Assume that $\Omega$ consists of $M$ annuluses $A_1,A_2,\ldots,A_M$ and a finite number of polygons with $N$ vertices $v_1,v_2,\ldots,v_N$.
% {\color{red}and edges $e_1,e_2,\ldots,e_N$}. 
Then the $M$ annuluses and the $N$ vertices can be uniquely determined from the scattered field data
$\{u^s(x,k)\mid x\in \Gamma_{L},\, k\in \left[k_{-}, k_{+}\right]\}$,
provided that
\begin{equation}
L>\max\{16M+4N-8,4N-2\}.
\label{structure-condition}
\end{equation}
\end{theorem}
\begin{proof}
Let first $\pa B_{d_m}(y_m)\subset \pa A$, where $A:=\mathop{\bigcup} \limits_{m}A_m$. If
\ben
P_2(\llbracket I_x^{\prime} \rrbracket(|x-y_m|+d_m))=0\quad \mbox{or}\quad P_2(\llbracket I_x^{\prime} \rrbracket(|x-y_m|-d_m))=0,
\enn
then we derive that (1) there exist $\tilde{m}\neq m$ and $\pa B_{d_{\tilde{m}}}(y_{\tilde{m}})$ such that 
\begin{equation}
|x-y_m|\pm d_m=|x-y_{\tilde{m}}|\pm d_{\tilde{m}}, \quad x\in \R^2.
\label{annulues-P2-hyperbolas4}
\end{equation}
or (2) there exists a line $l_n: \,a_n\wi{x_1}+b_n\wi{x_2}+c_n=0, (\wi{x_1}, \wi{x_2})\in \R^2$ containing some edge $e_n\in \{e_1,e_2,\ldots,e_N\}$ such that 
\begin{equation}
|x-y_m|\pm d_{m}=\left(a_nx_1+b_nx_2+c_n\right)\Big/\left(\sqrt{a_n^2+b_n^2}\right), \quad x\in \R^2.
\label{annulues-P2-parabolas2}
\end{equation}
In other words, the sensor $x$ must lie on four hyperbolas \eqref{annulues-P2-hyperbolas4} or two parabolas \eqref{annulues-P2-parabolas2}. Following the arguments of the proof of Theorem \ref{Annuluses-proof} and considering the fact that the vertices of both parabolas lie within $B_R(0)$, we therefore conclude that 
\begin{equation*}
    F_{\text{annular}}(y_m,d_m)\geq 2L-[16(M-1)+4N].
\end{equation*} 
Let now $\pa B_d(y)\not\subset \pa A$, following similar arguments we have
\begin{equation*}
    F_{\text{annular}}(y,d)\leq 16M+4N.
\end{equation*}
Hence, the uniqueness of the $M$ annuluses holds if $2L-[16(M-1)+4N]>16M+4N$, i.e., $L>16M+4N-8$.
%Thus, all $M$ annuluses can be uniquely determined from the data set $\{u^s(x,k)\mid x\in \Gamma_{L},\, k\in \left[k_{-}, k_{+}\right]\}$ under the condition $L>16M+4N-8$.

Similar to the arguments in Theorem \ref{polygon_corners}, for any fixed $x\in \Gamma_L$, any vertex $v_0\in \{v_1,v_2,\ldots,v_N\}$, and any point $v\in \R^2\backslash\{v_1,v_2,\ldots,v_N\}$, we have 
\begin{equation*}
    F_{\text{vertex}}(v_0)\geq L-2(N-1)>2N\geq F_{\text{vertex}}(v).
\end{equation*}
Therefore, the uniqueness of the $N$ vertices holds provided $L>4N-2$. The proof is complete.

\end{proof}

Finally, once the vertices of the polygons are determined, an argument analogous to that in Theorem \ref{polygon_edges} ensures that there exists a finite set $\Gamma_{\tilde{L}}$ of observation points from which the edges of the polygons and the source values $S(v_n)$, $n=1,2,\ldots,N$ can be uniquely determined from the data $\{u^s(x,k)\mid x\in \Gamma_{\tilde{L}},\, k\in \left[k_{-}, k_{+}\right]\}$.  
%%%%%%%%%%%%%%%%%%%%%%%%%%%%%%%%%%%%%%%%%%%%%%%%%%%%%%%%%%%%%%%%%%%%%%%%%%%%%%%%%%%%%%%%%%%%%%%%%%%%%%%%%%%%%%%%%%%%%%%%%%%%%%
%%%%%%%%%%%%%%%%%%%%%%%%%%%%%%%%%%%%%%%%%%%%%%%%%%%%%%%%%%%%%%%%%%%%%%%%%%%%%%%%%%%%%%%%%%%%%%%%%%%%%%%%%%%%%%%%%%%%%%%%%%%%%%%

\section{Numerical algorithms}
\label{Numerical methods}
\setcounter{equation}{0}
Based on the constructive uniqueness proofs in the previous section, we introduce the following two indicators to reconstruct the support $\Omega$ of the source function $S(z)$, $z\in \R^2$. We begin with \textbf{IP(2)} to reconstruct the boundary $\pa \Omega$ of the source function $S(z)$, $z\in \R^2$. Then we consider the more complex \textbf{IP(1)} for identifying the source function $S(z)$, $z\in \R^2$. 

\subsection{The indicator function \texorpdfstring{$I_{\partial \Omega}$}{I\_{\textpartial Ω}}} \label{indicator-1}
Based on the uniqueness proof in subsection \ref{Reconstruct-partial-Omega}, we define the indicator function $I_{\partial \Omega}$ with
\begin{equation}
\begin{aligned}
I_{\partial \Omega}(z):=&\Bigg|\sum\limits_{x\in \Gamma_L}I_{\text{integral}}(x,z) \Bigg|\Bigg/\max\limits_{z\in D}\left(\Bigg|\sum\limits_{x\in \Gamma_L}I_{\text{integral}}(x,z)\Bigg|\right) ,\quad z\in D,
\end{aligned}
\label{Ipa-Omega}
\end{equation}
where 
\begin{equation}
I_{\text{integral}}(x,z):=\int_{k_{-}}^{k_{+}}8k^3\Im(u^{s}(x,k))\Big[J_0(k|x-z|)
-k|x-z|J_1(k|x-z|)\Big]dk,
\nonumber
\end{equation}
and $D$ is the sampling domain containing $\Om$.
As demonstrated in subsection \ref{Reconstruct-partial-Omega}, for sufficiently small $k_-$, large $L$ and $k_+$ , the indicator function $I_{\pa \Omega}(z)$ approximates $I'(|x-z|)$ for $x\in \Gamma_L$. Consequently, $I_{\partial \Omega}(z)$ exhibits rapid variation or large values at points $z$ where the derivative $I'(|x-z|)$ does not exist.
By analyzing the contributions of $z\in \partial \Omega$ 
to the discontinuities of $I'(|x-z|)$ across sensors $x\in \Gamma_L$, we conclude that for sufficiently large $L$, the indicator function $I_{\partial \Omega}(z)$ reliably identifies the supports of both annular and polygonal sources. Algorithm \ref{partial_Omega} details the reconstruction procedure for  $\partial \Omega$ using the direct sampling method.

\begin{algorithm}[h!]
\label{partial_Omega}
\caption{Qualitative sampling method for the boundary $\partial \Omega$ of the source support}% 算法名字
\LinesNumbered %要求显示行号

1. Collect scattered field patterns $u^s(x,k)$ for all $x\in \Gamma_L$, $k\in [k_-,k_+]$;\\
2. Compute the indicator function $I_{\pa \Omega}(z)$ for all sampling points $z\in D$;\\
3. Plot the indicator function $I_{\pa \Omega}(z)$.\\
\end{algorithm}
%%%%%%%%%%%%%%%%%%%%%%%%%%%%%%%%%%%%%%%%%%%%%%%%%%%%%%%%%%%%%%%%%%%%%%%%%%%%%%%%%%%%%%%%%%%%%%%%%%%%%%%%%%%%%%%%%%%%%%%%%%%%%%
%%%%%%%%%%%%%%%%%%%%%%%%%%%%%%%%%%%%%%%%%%%%%%%%%%%%%%%%%%%%%%%%%%%%%%%%%%%%%%%%%%%%%%%%%%%%%%%%%%%%%%%%%%%%%%%%%%%%%%%%%%%%%%%
\subsection{The indicator function \texorpdfstring{$I^{(1)}_{S}$}{}}\label{indicator-2}
We now consider the significantly more challenging case with  source reconstructions. Motivated by the uniqueness analysis in subsection \ref{Identify-source-function}, we define the second indicator function $I^{(1)}_{S}(z)$ by
\begin{equation}
\begin{aligned}
I^{(1)}_{S}(z):=& \frac{R}{\pi}\int_0^{2\pi} \Bigg[ \frac{z-x}{|z-x|}\cdot (\cos\theta,\sin \theta)\Bigg]\int _{0}^{\lambda_+}\int _0^{2R} \la ^2 r\Big[N_1(\la |z-x|)J_0(\la r)\\
&-J_1(\la |z-x|)N_0(\la r)\Big] \int_{k_-}^{k_+}k^3 \Im(u^s(x,k))J_0(k r) dk dr d\la d\theta, \quad z\in D,
\end{aligned}
\label{I-Omega-(1)}
\end{equation}
where $x=R(\cos \theta, \sin \theta)$.
According to Theorem \ref{source function-07}, for sufficiently small $k_-$, large $L$, $k_+$ and $\lambda_+$, the indicator function $I^{(1)}_{S}(z)$ approximates the source function $S(z)$.

\begin{algorithm}[h!]
\label{I-S-(1)}
\caption{Qualitative sampling method for source functions using $I^{(1)}_S(z)$}
\LinesNumbered 

1. Collect scattered field patterns $u^s(x,k)$ for all $x\in \Gamma_L$, $k\in(k_{-},k_{+})$; \\
2. Compute the indicator function $I^{(1)}_{S}(z)$ for all sampling points $z\in D$;\\
3. Plot the indicator function $I^{(1)}_{S}(z)$.\\
\end{algorithm}

%%%%%%%%%%%%%%%%%%%%%%%%%%%%%%%%%%%%%%%%%%%%%%%%%%%%%%%%%%%%%%%%%%%%%%%%%%%%%%%%%%%%%%%%%%%%%%%%%%%%%%%%%%%%%%%%%%%%%%%%%%%%%%
%%%%%%%%%%%%%%%%%%%%%%%%%%%%%%%%%%%%%%%%%%%%%%%%%%%%%%%%%%%%%%%%%%%%%%%%%%%%%%%%%%%%%%%%%%%%%%%%%%%%%%%%%%%%%%%%%%%%%%%%%%%%%%%
\subsection{The indicator function \texorpdfstring{$I^{(2)}_{S}$}{}}\label{indicator-3}
%In this subsection, we employ Green's formula for the biharmonic wave equation and the properties of Bessel functions to derive an explicit formula for calculating the source function $S(y)$. In contrast to Theorem \ref{source function-07}, we present a more {\color{red}concise???} representation involving only a double integral.

In this subsection, we present a simplified formula for the source function, expressed as a double integral, at the cost of requiring additional measurements $\Delta u^s$.

Before introducing the following lemma and theorem, we recall the fundamental solution \eqref{equation_FS} of the biharmonic wave equation. Straightforward calculations yield 
\begin{equation}
\Delta_x\Phi_k(x,z):=  -\frac{1}{8}\left(iH_0^{(1)}(k|z-x|)+\frac{2}{\pi}K_0(k|z-x|)\right),\quad  x,z \in \R^2,\, x\neq z,
\label{laplace-phi}
\end{equation}
and the Bessel function $J_0$ satisfies the Helmholtz equation
\begin{equation}
\Delta_x J_0(k|z-x|)=-k^2J_0(k|z-x|), \quad  x,z \in \R^2.
\label{laplace-J_0}
\end{equation}

%We first apply the Bessel addition theorem to establish several integral identities. 
Recall the normalized complex circular harmonics
\begin{equation}
Y_n(\hat{x})=\frac{1}{\sqrt{2\pi}}e^{in\theta},\quad n\in \mathbb{Z},
\nonumber
\end{equation}
where $\hat{x}:=(\cos\theta,\sin \theta)^{\top}$, $\theta\in [0,2\pi)$. 
%Unlike the proof provided in the appendix of \cite{2007Kunyansky-formula-f}, we present a simpler alternative approach in the following lemma by employing the complex circular harmonics as basis functions.
Then we have the following identities.

\begin{lemma}
For any circular harmonic $Y_n$, we have the following three identities,
\begin{equation}
\begin{aligned}  
&\int_{\mathbb{S}^1}J_0(k|x-y|)Y_n(\hat{y})ds(\hat{y})=2\pi J_n(k|x|)J_n(k|y|)Y_n(\hat{x}),\\
&\int_{\mathbb{S}^1}N_0(k|x-y|)Y_n(\hat{y})ds(\hat{y})=2\pi N_n(k|x|)J_n(k|y|)Y_n(\hat{x}),\\
&\int_{\mathbb{S}^1}H_0^1(k|x-y|)Y_n(\hat{y})ds(\hat{y})=2\pi H_n^1(k|x|)J_n(k|y|)Y_n(\hat{x}),
\end{aligned}
\label{Funk-heck-J0-Y0-H}
\end{equation}
where $\hat{x}:=x/|x|, |x|>|y|$, $x,y\in \R^2$, $n\in \mathbb{Z}$.
\end{lemma}
\begin{proof}
Denote by $\alpha\in [0,2\pi)$ the angle between $x$ and $y$. Then for $\hx=(\cos\theta,\sin \theta)^{\top}$, we have 
$\hat{y}:=y/|y|=(\cos(\alpha+\theta),\sin (\alpha+\theta))^{\top}$. 
With the help of the addition theorem \cite{CK}, 
\begin{equation}
H^1_0(k|x-y|)=\sum_{m=-\infty}^\infty H^1_m(k|x|)J_m(k|y|)e^{im\alpha},\quad  |x|>|y|, x,y\in \R^2,
\nonumber   
\end{equation}
using the orthogonality of the circular harmonics, the identities 
\begin{equation*}
    H^1_{-n}(t)=(-1)^nH^1_n(t)\quad\mbox{and}\quad J_{-n}(t)=(-1)^nJ_n(t), \quad t\in R,
\end{equation*}
we derive that
% \begin{equation}
% \begin{aligned}
\begin{equation*}
\begin{aligned}  
   \int_{\mathbb{S}^1}H^1_0(k|x-y|)Y_n(\hat{y})ds(\hat{y})
   &=\frac{e^{in\theta}}{\sqrt{2\pi}}\int_{\mathbb{S}^1}H^1_0(k|x-y|)e^{in\alpha}ds(\hat{y}')\cr
&=\frac{e^{in\theta}}{\sqrt{2\pi}}\int_0^{2\pi}\sum_{m=-\infty}^\infty H^1_m(k|x|)J_m(k|y|)e^{i(m+n)\alpha}d\alpha\cr
&=2\pi H^1_n(k|x|)J_n(|y|)Y_n(\hat{x}), \quad |x|>|y|, x,y\in \R^2. 
\end{aligned}
\end{equation*}
Thus, the third identity of \eqref{Funk-heck-J0-Y0-H} is proved.
The other two identities can be obtained similarly.

\end{proof}

We now define
\begin{equation}
I(y,z):=\int_{\pa B_R(0)}\Bigg[\Delta_x\Phi_k(x,y)\frac{\pa J_0(k|z-x|)}{\pa \nu_x}+\Phi_k(x,y)\frac{\pa \Delta_x J_0(k|z-x|)}{\pa \nu_x}\Bigg]ds(x),\quad y,z\in B_R(0).
\label{I}
\end{equation}
Then the following reciprocity relation holds.
\begin{lemma}\label{I(y,z)=I(z,y)}
\begin{equation}
I(y,z)=I(z,y),\quad y,z\in B_R(0).
\nonumber
\end{equation}
\end{lemma}
\begin{proof}
Using equations \eqref{Scattered field}, \eqref{laplace-phi}, and \eqref{laplace-J_0}, we can rewrite $I(y,z)$ as
\begin{equation}
I(y,z)=-\frac{i}{4}\int_{\pa B_R(0)}H_0^1(k|x-y|)\frac{\pa J_0(k|z-x|)}{\pa \nu_x}ds(x), \quad y,z\in B_R(0).
\label{simply-I}
\end{equation}
Note that $Y_n(\hat{y})\overline{Y_{n^{'}}(\hat{z})}, n, n^{'}\in \Z$ form the orthonormal basis of $L^2(\mathbb{S}^1)\times L^2(\mathbb{S}^1)$, we have 
\begin{equation}
I(y, z)=\sum_{n,n'=-\infty}^\infty a^{n}_{n'}(|y|,|z|)Y_n(\hat{y})\overline{Y_{n^{'}}(\hat{z})},\quad y,z\in B_R(0)
\nonumber
\end{equation}
with
\begin{equation}
a^{n}_{n'}(|y|,|z|)=\int_{\mathbb{S}^1}\int_{\mathbb{S}^1} I(y,z)\overline{Y_n(\hat{y})}Y_{n^{'}}(\hat{z})d\hat{y}d\hat{z}.
\label{coefficient}
\end{equation}
%To prove $I(y,z)=I(z,y)$, it suffices to prove that $a^{n}_{n'}(|y|,|z|)=a^{n}_{n'}(|z|,|y|)$. 
Inserting \eqref{simply-I} into \eqref{coefficient} and applying the identities \eqref{Funk-heck-J0-Y0-H}, we obtain 
% \begin{equation}
% \begin{aligned}
\be
a^{n}_{n'}(|y|,|z|)&=&-\frac{i}{4}\int_{\mathbb{S}^1}\int_{\mathbb{S}^1} \Bigg[\int_{\pa B_R(0)}H_0^1(k|x-y|)\frac{\pa J_0(k|z-x|)}{\pa \nu_x}ds(x)\Bigg]\overline{Y_n(\hat{y})\overline{Y_{n^{'}}(\hat{z})}}d\hat{y}d\hat{z}\cr
&=&-\frac{i}{4}\int_{\pa B_R(0)}\int_{\mathbb{S}^1}Y_{n^{'}}(\hat{z})\Bigg[\int_{\mathbb{S}^1}H_0^1(k|x-y|)\overline{Y_n(\hat{y})}d\hat{y}\Bigg]\frac{\pa J_0(k|z-x|)}{\pa \nu_x}d\hat{z}ds(x)\cr
&=&-\frac{\pi i}{2}J_n(k|y|)H_n^1(k|R|)\int_{\pa B_R(0)}\Bigg[\frac{\pa}{\pa \nu_x}\int_{\mathbb{S}^1}J_0(k|z-x|)Y_{n^{'}}(\hat{z})d\hat{z}\Bigg]\overline{Y_n(\hat{x})}ds(x)\cr
&=&-\pi ^2RkiJ_n(k|y|)H_n^1(k|R|)J_n(k|z|)J_n^{'}(kR)\int_{\mathbb{S}^1}Y_n(\hat{x})\overline{Y_{n^{'}}(\hat{x})}ds(\hat{x})\cr
&=&-\pi ^2RkiH_n^1(k|R|)J_n^{'}(kR)J_n(k|y|)J_n(k|z|)\delta^n_{n'}, \quad y,z\in B_R(0).
\label{coefficient-proof}
% \end{aligned}
% \end{equation}
\en
where $\delta^n_{n'}$ is the Kronecker delta.
This further implies
%By the orthogonality of the circular harmonics, $a^{n}_{n'}(|y|,|z|)=0$ for $n\neq n'$. For $n= n^{\prime}$, equation \eqref{coefficient-proof} implies 
\ben
a^{n}_{n'}(|y|,|z|)=0\quad\mbox{for}\quad n\neq n'\quad \mbox{and}\quad a^{n}_{n}(|y|,|z|)=a^{n}_{n}(|z|,|y|). 
\enn
Moreover, using the identities
\begin{equation}
f_n(t)=(-1)^nf_{-n}(t),\quad t>0\quad\mbox{for}\quad f_n=H_n^1, J_n^{'} \,\mbox{or}\, J_n,
%H_n^1(k|R|)=(-1)^nH_{-n}^1(k|R|),\quad J_n^{'}(kR)=(-1)^nJ_{-n}^{'}%(kR), 
%\nonumber
%\end{equation}
%\begin{equation}
%J_n(k|y|)=(-1)^nJ_{-n}(k|y|),\quad J_n(k|z|)=(-1)^nJ_{-n}(k|z|),
\nonumber
\end{equation}
we deduce from \eqref{coefficient-proof} that 
\begin{equation*}
    a^{n}_{n}(|y|,|z|)=a^{-n}_{-n}(|y|,|z|).
\end{equation*}
Finally, using the fact that $Y_{-n}(t)=\ov{Y_n(t)}$, we deduce that 
\begin{equation}
\begin{aligned}
I(y, z)&=\sum_{n=-\infty}^{+\infty} a^{n}_{n}(|y|,|z|)Y_n(\hat{y})\overline{Y_{n}(\hat{z})}\\
&=\sum_{n=-\infty}^{\infty} a^{-n}_{-n}(|y|,|z|)\overline{Y_{-n}(\hat{y})}Y_{-n}(\hat{z})\\
&=\sum_{n=-\infty}^{+\infty} a^{n}_{n}(|z|,|y|)Y_{n}(\hat{z})\overline{Y_n(\hat{y})}\\
&=I(z, y),\quad y,z\in B_R(0).
\end{aligned}
\nonumber
\end{equation}
\end{proof}

We define another indicator function
\begin{equation}
\begin{aligned}
\tilde{I}_S(z):=&\frac{R}{2\pi}\int _{0}^{2\pi}\int _{\R^+} k^2 \Bigg[ \frac{(z-x)}{|z-x|}\cdot(\cos\theta,\sin \theta)\Bigg]\Bigg[k^2 J_1(k|z-x|)u^s(x,k)\\
&-2k^2iH_1^1(k|z-x|) \Im(u^s(x,k))-J_1(k|z-x|)\Delta_x u^s(x,k) \Bigg]dk d\theta,\quad z\in B_R(0),
\label{S(y)-fomula}
\end{aligned}
\end{equation}
where $x=R(\cos\theta,\sin\theta)\in \pa B_R(0)$. The following result then holds.
\begin{theorem}\label{S(y)-fomula-proof111}
Let the real-valued source function $S(z)\in L^2_{\text{comp}}(\R^2)$ with support $\Omega \subset B_R(0)$. Then
\begin{equation}
\tilde{I}_S(z)=S(z),\, z\in B_R(0).
\nonumber
\end{equation}
\end{theorem}
\begin{proof}
Applying Green's formula to the biharmonic wave equation, we have
% \begin{equation}
% \begin{aligned}
\be
J_0(k|z-y|)&=&\int_{\pa B_R(0)}\Bigg[J_0(k|x-y|)\frac{\pa \Delta_x \Phi_k(x,z)}{\pa \nu_x}+\Delta_x J_0(k|x-y|)\frac{\pa \Phi_k(x,z)}{\pa \nu_x}\cr
&&-\left(\Delta_x\Phi_k(x,z)\frac{\pa J_0(k|y-x|)}{\pa \nu_x}+\Phi_k(x,z)\frac{\pa \Delta_x J_0(k|y-x|)}{\pa \nu_x}\right)\Bigg]ds(x)\cr
&=&\int_{\pa B_R(0)}\Bigg[J_0(k|x-y|)\frac{\pa \Delta_x \Phi_k(x,z)}{\pa \nu_x}+\Delta_x J_0(k|x-y|)\frac{\pa \Phi_k(x,z)}{\pa \nu_x}\cr
&&-\left(\Delta_x\Phi_k(x,y)\frac{\pa J_0(k|z-x|)}{\pa \nu_x}+\Phi_k(x,y)\frac{\pa \Delta_x J_0(k|z-x|)}{\pa \nu_x}\right)\Bigg]ds(x)\cr
&=&\int_{{\pa B_R(0)}}\Bigg[k^2\Phi_k(y,x)\frac{\pa J_0(k|z-x|)}{\pa \nu_x}-\frac{i}{4}J_0(k|x-y|)\frac{\pa H_0^1(k|z-x|)}{\pa \nu_x}\cr
&&-\Delta_x \Phi_k(y,x)\frac{\pa J_0(k|z-x|)}{\pa \nu_x}\Bigg]ds(x),\quad \forall y,z\in B_R(0),\, k\in \R^+,
\label{green-J}
% \end{aligned}
% \end{equation}
\en
where the second equality follows from Lemma \ref{I(y,z)=I(z,y)} and the last equality follows from \eqref{laplace-phi} and \eqref{laplace-J_0}. 
%Here, $\nu_x$ denotes the unit outward normal vector to $\pa B_R(0)$ at $x$. 
% Note that $\Omega \subset B_R(0)$ and the equation \eqref{Fourier-S}, we have $S(z)=0$, $|y|,|z|\geq 0$.
% To derive an explicit expression for the source function $S(y)$, we begin with its Fourier representation:
% \begin{equation}
% \begin{aligned}
% S(z)&=\frac{1}{(2\pi)^2}\int_{\R^2}\int_{\R^2} S(y) e^{-iy\cdot \xi}dy e^{iz\cdot \xi}d\xi\\
% &=\frac{1}{(2\pi)^2}\int _{\R^{+}}\int_{\R^2}\int_{\mathbb{S}^1}e^{ik(z-y)\cdot \hat{\xi}}d\hat{\xi}S(y)dy kdk\\
% &=\frac{1}{2\pi}\int _{\R^{+}}\int_{\R^2}J_0(k|z-y|)S(y)dy kdk,\quad \forall z\in \R^2.
% \label{Fourier-S}
% \end{aligned}
% \end{equation}

Using equation \eqref{green-J} and the well-known integral representation \cite{fourier-formula},
\begin{equation}
J_0(|y|)=\frac{1}{2\pi}\int_{\mathbb{S}^1}e^{iy\cdot \hat{\xi}}d\hat{\xi},\quad \forall y\in \R^2,
\nonumber
\end{equation}
% Define the inner integral in \eqref{Fourier-S} as
% \begin{equation}
% G(z,k)=\int_{\R^2}J_0(k|z-y|)S(y)dy,\quad z\in \R^2, k\in \R^+.
% \label{G}
% \end{equation}
% Our goal is to express the unknown data $G(y,k)$ in terms of the available scattered field data. 

% Substituting equation \eqref{green-J} into equation \eqref{G} and using the scattered field expression \eqref{Scattered field} with the variable substitution $x:=R(\cos \theta,\sin \theta)$, 
we derive
% \begin{equation}
% \begin{aligned}
\be
\tilde{I}_S(z)&=&\frac{R}{2\pi}\int _{0}^{2\pi}\int _{0}^{+\infty} k^2 \Bigg[ \frac{(z-x)}{|z-x|}\cdot(\cos\theta,\sin \theta)\Bigg]\Bigg[k^2 J_1(k|z-x|)u^s(x,k)\cr
&&-2k^2iH_1^1(k|z-x|) \Im(u^s(x,k))-J_1(k|z-x|)\Delta_x u^s(x,k) \Bigg]dk d\theta\cr
&=&-\frac{1}{2\pi}\text{div}\int_{\pa B_R(0)}\nu(x)\int_{0}^{+\infty} k\Big[k^2u^s(x,k)J_0(k|z-x|)\cr
&&-\left(2k^2iH_0^1(k|x-z|)\Im (u^s(x,k))+\Delta_xu^s(x,k)J_0(k|z-x|)\right)\Big]dkds(x)\cr
&=&\frac{1}{2\pi}\int _{\R^{+}}\int_{B_R(0)}J_0(k|z-y|)S(y)dy kdk\cr
&=&\frac{1}{(2\pi)^2}\int _{\R^{+}}\int_{\R^2}\int_{\mathbb{S}^1}e^{ik(z-y)\cdot \hat{\xi}}d\hat{\xi}S(y)dy kdk\cr
&=&\frac{1}{(2\pi)^2}\int_{\R^2}\int_{\R^2} S(y) e^{-iy\cdot \xi}dy e^{iz\cdot \xi}d\xi\cr
&=&S(z),\quad z\in B_R(0).
\label{S(y)-fomula-proof}
% \end{aligned}
% \end{equation}
\en
\end{proof}

Based on Theorem \ref{S(y)-fomula-proof111}, and given that the frequencies lie within the interval $(k_{-},k_{+})$, we employ the indicator function
\begin{equation}
\begin{aligned}
I^{(2)}_{S}(z):=&\frac{R}{2\pi}\int_0^{2\pi}\int_{k_{-}}^{k_{+}}k^2 \bigg[ \frac{(z-x)}{|z-x|}\cdot(\cos\theta,\sin \theta)\bigg]
\bigg[k^2 J_1(k|z-x|)u^{s}(x,k)\\
&-2k^2iH_1^1(k|x-z|) \Im(u^{s}(x,k))-J_1(k|x-z|)\Delta_x u^{s}(x,k) \bigg]dkd\theta,\quad z\in D,
\end{aligned}
\label{I-Omega-(2)}
\end{equation}
where $x=R(\cos \theta,\sin \theta)$, to calculate the source function. The procedure is formulated in the following algorithm.
%According to Theorem \ref{S(y)-fomula-proof111}, for sufficiently large $L$ and $k_+$, the indicator function $I^{(2)}_{S}(z)$ approximates the source function $S(z)$. Although the theoretical analysis assumes a large number of observation points to identify the extended source, our numerical experiments show that accurate identification of $S(z)$ is achievable with only a finite set of sensors.

\begin{algorithm}[h!]
\label{I-S-(2)}
\caption{Qualitative sampling method based on the indicator function $I^{(2)}_S$(z)}% 算法名字
\LinesNumbered %要求显示行号

1. Collect scattered field patterns $u^s(x,k)$ and $\Delta_x u^s(x,k)$ for all $x\in \Gamma_L$, $k\in(k_{-},k_{+})$; \\
2. Compute the indicator function $I^{(2)}_{S}(z)$ for all sampling points $z\in D$;\\
3. Plot the indicator function $I^{(2)}_{S}(z)$.\\
\end{algorithm}

%%%%%%%%%%%%%%%%%%%%%%%%%%%%%%%%%%%%%%%%%%%%%%%%%%%%%%%%%%%%%%%%%%%%%%%%%%%%%%%%%%%%%%%%%%%%%%%%%%%%%%%%%%%%%%%%%%%%%%%%%%%%%%
%%%%%%%%%%%%%%%%%%%%%%%%%%%%%%%%%%%%%%%%%%%%%%%%%%%%%%%%%%%%%%%%%%%%%%%%%%%%%%%%%%%%%%%%%%%%%%%%%%%%%%%%%%%%%%%%%%%%%%%%%%%%%%%
\section{Numerical examples and discussions}
Before presenting the numerical examples, we introduce the following notations. The scattered field is computed via \eqref{Scattered field} and then perturbed by random noise as follows
\begin{equation}
\begin{aligned}
u^{s,\delta}(x_l,k_m)&=u^s(x_l,k_m)(1+\delta \xi_{lm})\quad\mbox{and}\quad
\Delta u^{s,\delta}(x_l,k_m)=\Delta u^{s,\delta}(x_l,k_m)(1+\delta \xi_{lm}),
\label{numerical-u^{s,noise}}
\end{aligned}
\end{equation}
where $\delta=0.2$ is the relative noise level and $\xi_{lm}$ is a uniformly distributed random number between $-1$ and $1$. 
For inverse problems \textbf{IP(1)} and \textbf{IP(2)}, the numerical dataset is constructed as
\begin{equation}
\{u^{s,\delta}(x_l,k_m)\mid x_l\in \Gamma_L,\, k_m\in K \},
\label{numerical-data-u^s}
\end{equation}
whereas the indicator function $I_{S}^{(2)}$ proposed in Subsection \ref{indicator-3} requires extended data
\begin{equation}
\{u^{s,\delta}(x_l,k_m),\, \Delta u^{s,\delta}(x_l,k_m)\mid x_l\in \Gamma_L,\, k_m\in K\}.
\label{numerical-data-laplace-u^s}
\end{equation}
We define the sparse sensor array as
\begin{equation}
\Gamma_L:=\left\{3\left(\cos\left(\frac{2\pi l}{L}\right), \sin\left(\frac{2\pi l}{L}\right)\right)\Bigg | l=0,1,\ldots,L-1\right\}
\label{numerical-GammaL}
\end{equation}
which is uniformly distributed on a circle of radius 3. The wavenumber set is given by
%$K:=\left\{k_m=\frac{m}{2}\Bigg | \, m=1,2,\ldots,2k_{+}, k_{-}=k_1 \right\}$.
\begin{equation}
K:=\left\{k_m=k_{-}+(m-1)dk\big | \, m=1,2,\ldots,N_{\text{max}},\, k_{+}=k_{N_{\text{max}}} \right\}.
\label{numerical-wave-number}
\end{equation}
Unless otherwise specified, we set $L=30$, $k_-=0.5$, $k_+=30$ and $dk=0.5$.
Additionally, we assume a priori that the source is located within the search domain $D=[-2,2]\times [-2,2]$, which is discretized using a uniform  $401\times 401$ grid of $z$ values.

%\subsection{Three numerical examples}  

\begin{figure}[h!]
\centering
\begin{tabular}{cc}
\subfigure[Annular source]{
\label{num_ex_1}
\includegraphics[width=.30\textwidth]{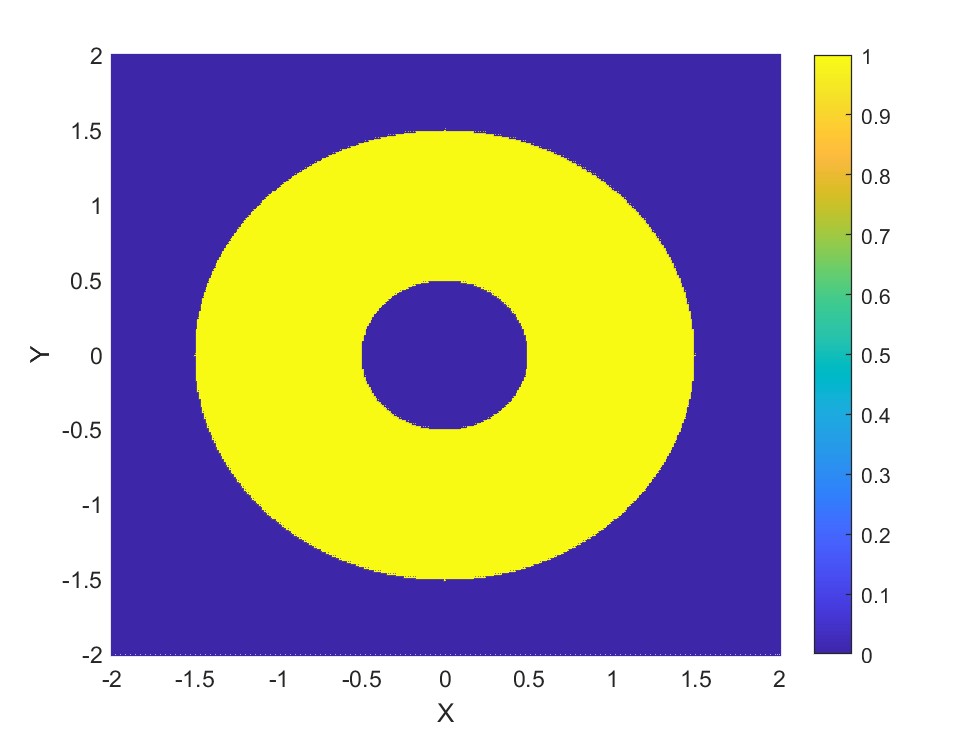}
}\hspace{0em} &
\subfigure[Polygonal source]{
\label{num_ex_2}
\includegraphics[width=.30\textwidth]{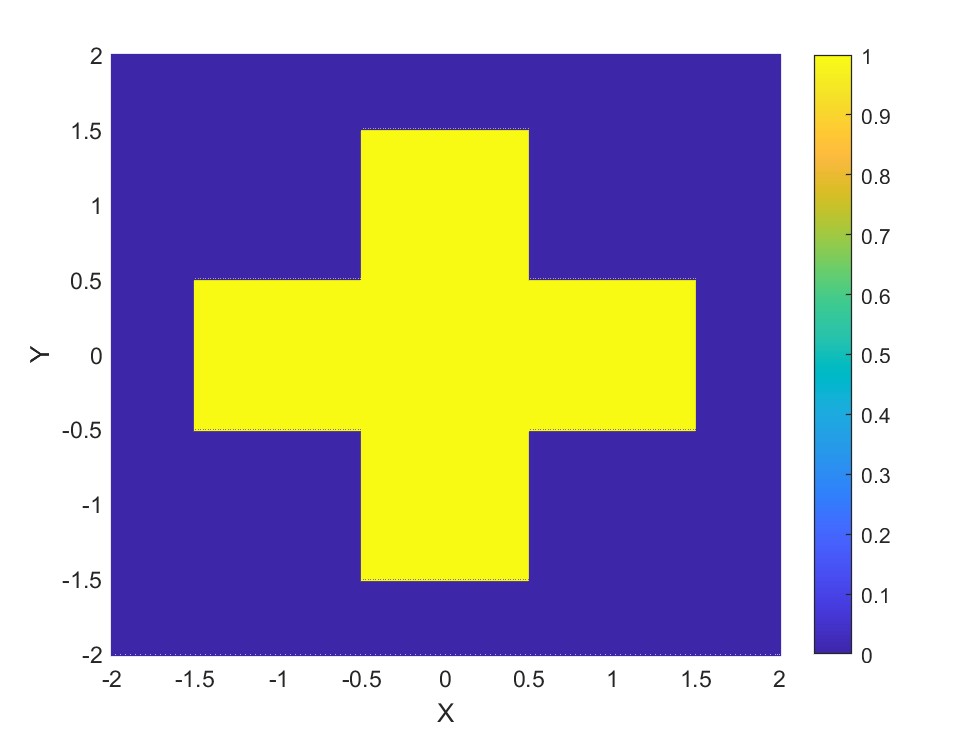}
}\\
\subfigure[Smiling bear]{
\label{num_ex_3}
\includegraphics[width=.30\textwidth]{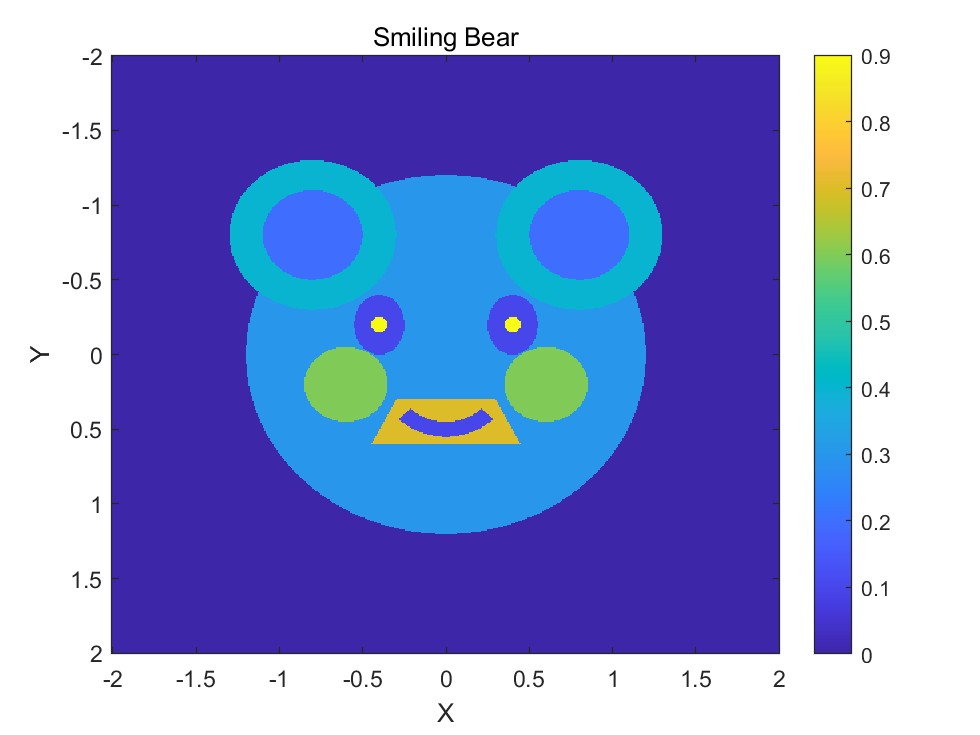}
}\hspace{0em} &
\subfigure[A smooth source]{
\label{num_ex_4}
\includegraphics[width=.30\textwidth]{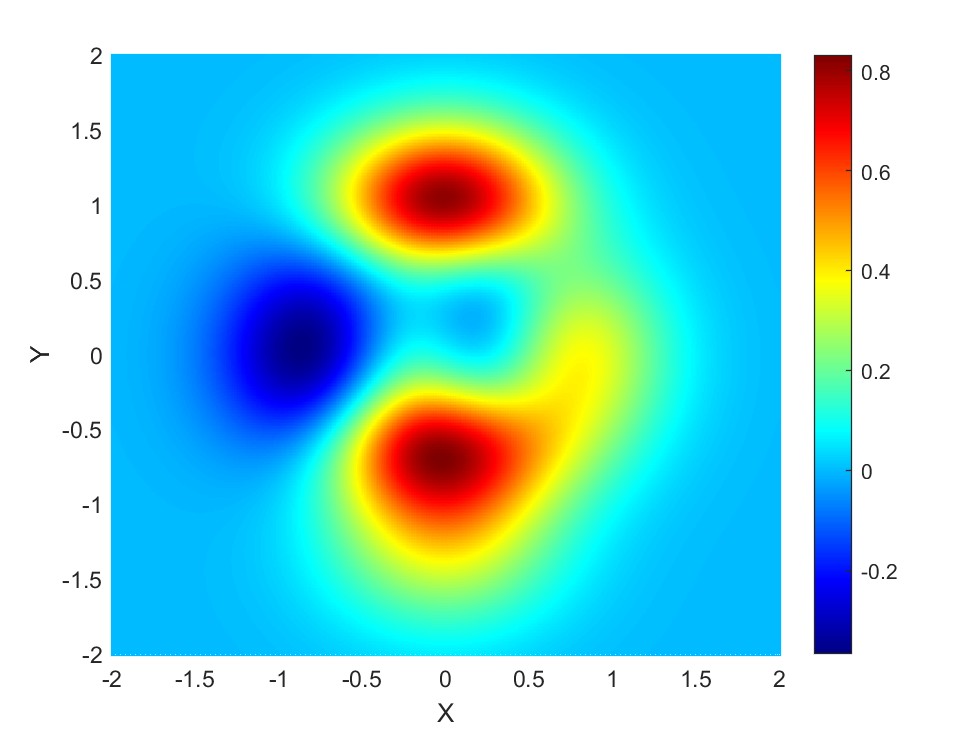}
}
\end{tabular}
\caption{The true sources.}
\label{true-sources}
\end{figure}

To verify the effectiveness and robustness of our indicator functions proposed in Section \ref{Numerical methods}, we designed the following four examples
\begin{itemize}
    \item $S=\chi_{\Omega}$ with $\Om:=\overline{B_{1.5}(0)\ba B_{0.5}(0)}$ being an annulus;
    \item $S=\chi_{\Omega}$ with $\Om$ being a cross-shaped polygon;
    \item A piecewise constant source function, supported in a smiling bear-shaped domain depicted in Figure \ref{num_ex_3};
    \item A smooth source function $S(y)$ defined by 
    \begin{equation*}
    \begin{aligned}
     S(y)=&0.3(1 - 1.5y_2)^2 e^{-[(1.5y_1)^2 + (1.5y_2 + 1)^2]}- 0.03 e^{-[(1.5y_1 + 1)^2 + (1.5y_2)^2]}\\
     &- [0.3y_1 - (1.5y_1)^3 - (1.5y_2)^5]e^{-[(1.5y_1)^2+ (1.5y_2)^2] },\quad y=(y_1,y_2)\in \R^2.
\end{aligned}
\end{equation*}
\end{itemize}
The true sources are illustrated in Figure \ref{true-sources}.
%In the first two examples, we take $S=\chi_{\Omega}$, where $\chi_{\Omega}$ is the characteristic function of an annulus and a cross-shaped polygon, respectively, as illustrated in Figure \ref{real annulus and polygon}. For the third example, we consider a general source function $S(y)$:
%\begin{equation}
%\begin{aligned}
%S(y)=&0.3(1 - 1.5y_2)^2 e^{-[(1.5y_1)^2 + (1.5y_2 + 1)^2]}- 0.03 e^{-[(1.5y_1 + 1)^2 + (1.5y_2)^2]}\\
%&- (0.3y_1 - (1.5y_1)^3 - (1.5y_2)^5)e^{-[(1.5y_1)^2+ (1.5y_2)^2] },
%\end{aligned}
%\label{third example}
%\end{equation}
%where $y:=(y_1,y_2)\in \R^2$.

\begin{figure}[h!]
\centering
\begin{tabular}{cc}
\subfigure{
\includegraphics[width=.30\textwidth]{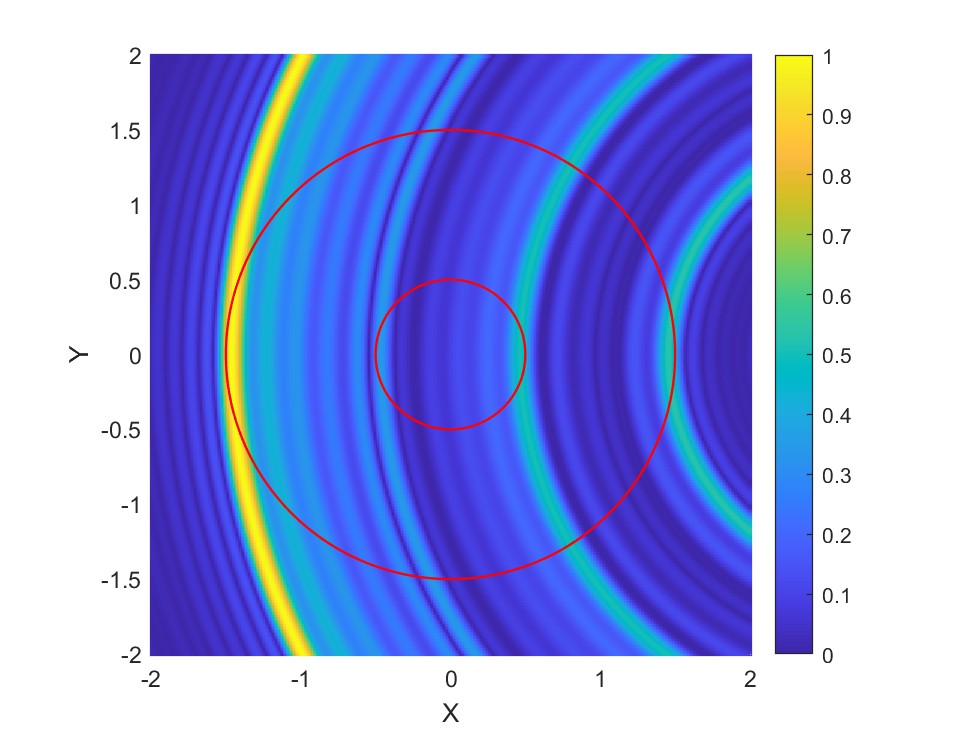}
}\hspace{0em} &
\subfigure{
\includegraphics[width=.30\textwidth]{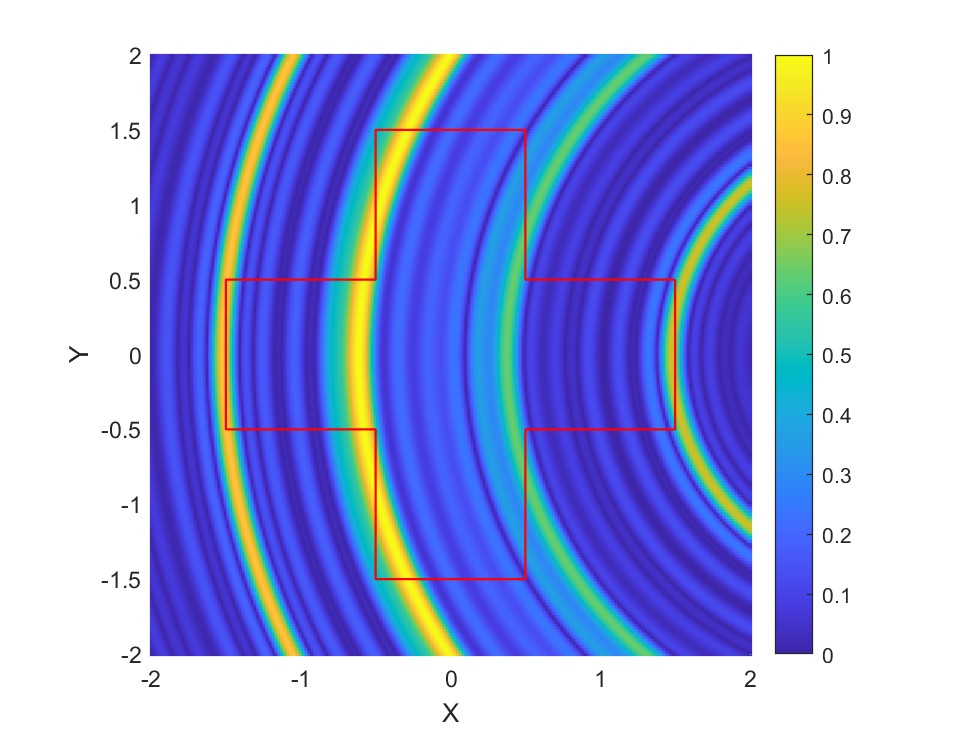}
}
\end{tabular}
\caption{Reconstructions of the support set by plotting $I_{\pa \Omega}$ with one sensor $(3,0)$. }
\label{L=1_annulus_+}
\end{figure}

\begin{figure}[h!]
\centering
\begin{tabular}{ccc}
\subfigure[$L=10$]{
\includegraphics[width=.30\textwidth]{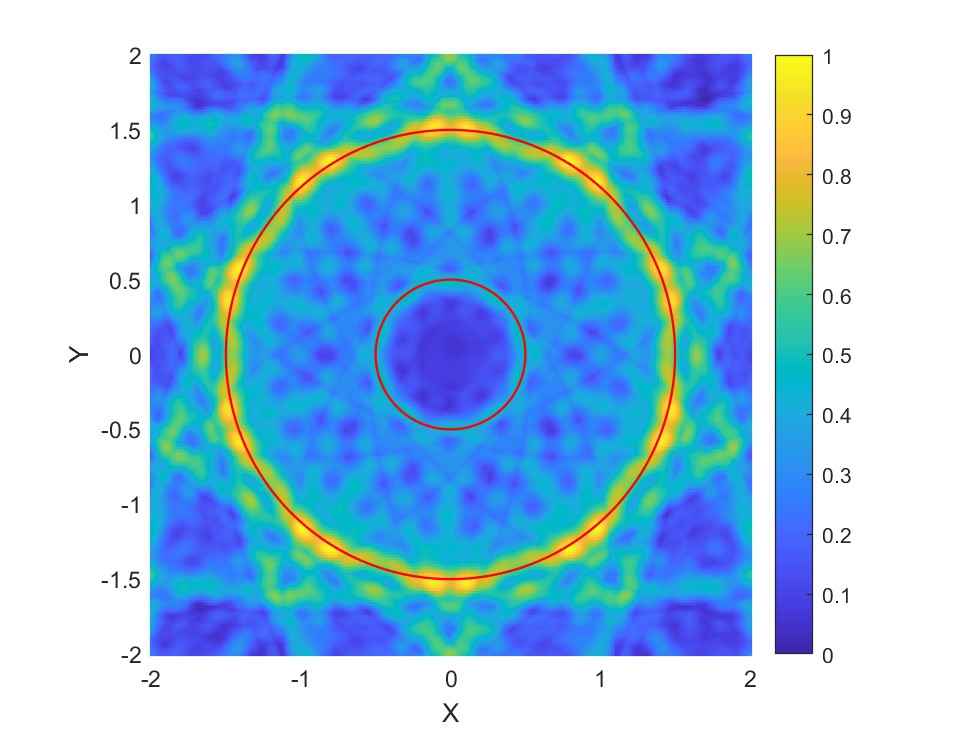}
}\hspace{0em} &
\subfigure[$L=20$]{
\includegraphics[width=.30\textwidth]{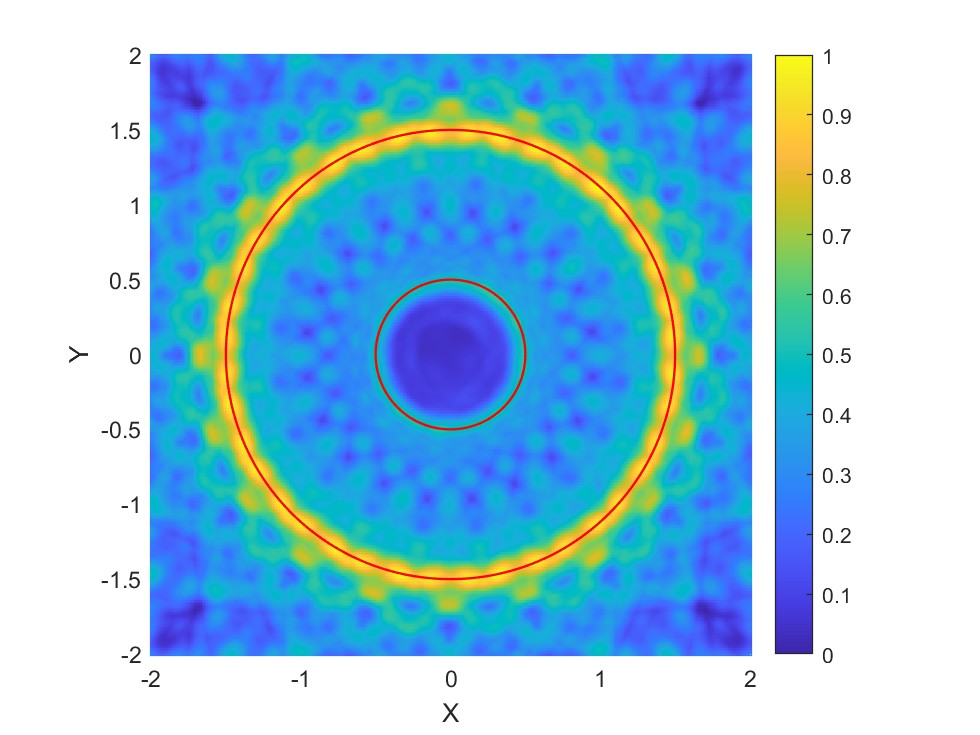}}&
\subfigure[$L=30$]{
\includegraphics[width=.30\textwidth]{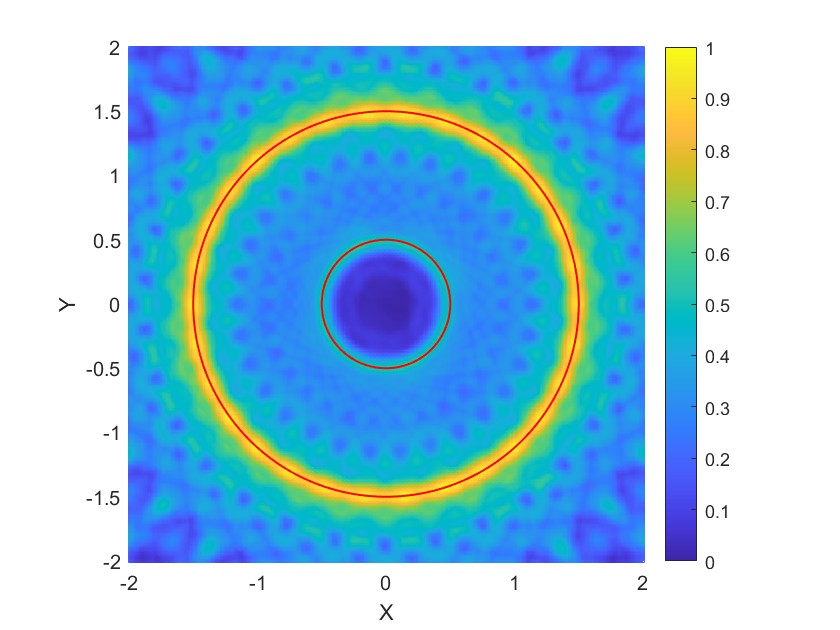}}
\\
\subfigure[$L=10$]{
\includegraphics[width=.30\textwidth]{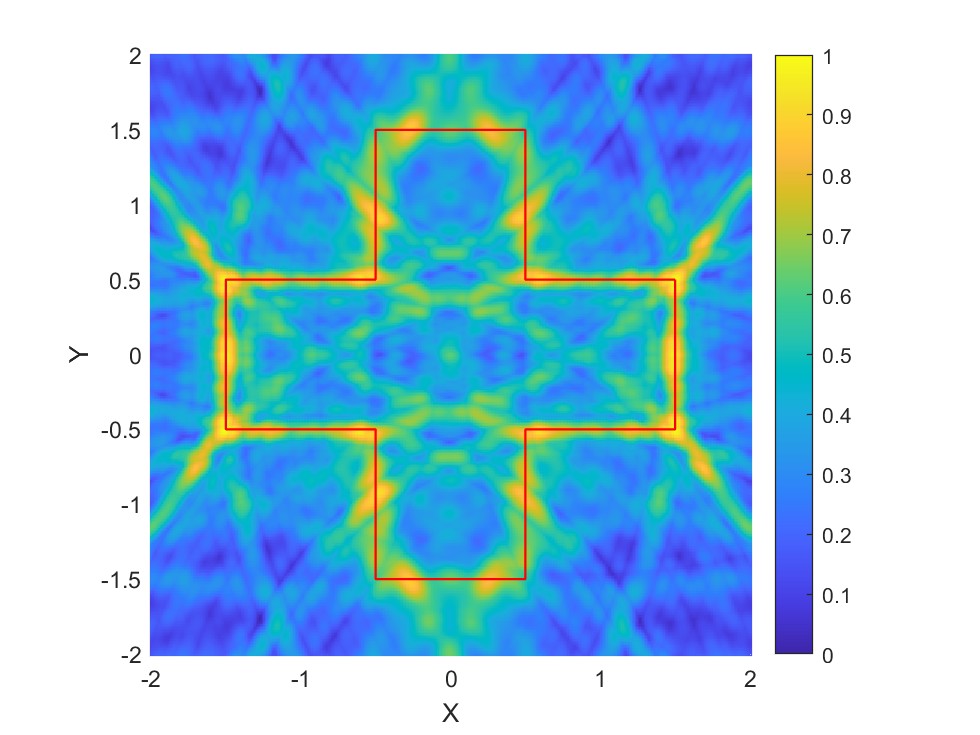}
}\hspace{0em} &
\subfigure[$L=20$]{
\includegraphics[width=.30\textwidth]{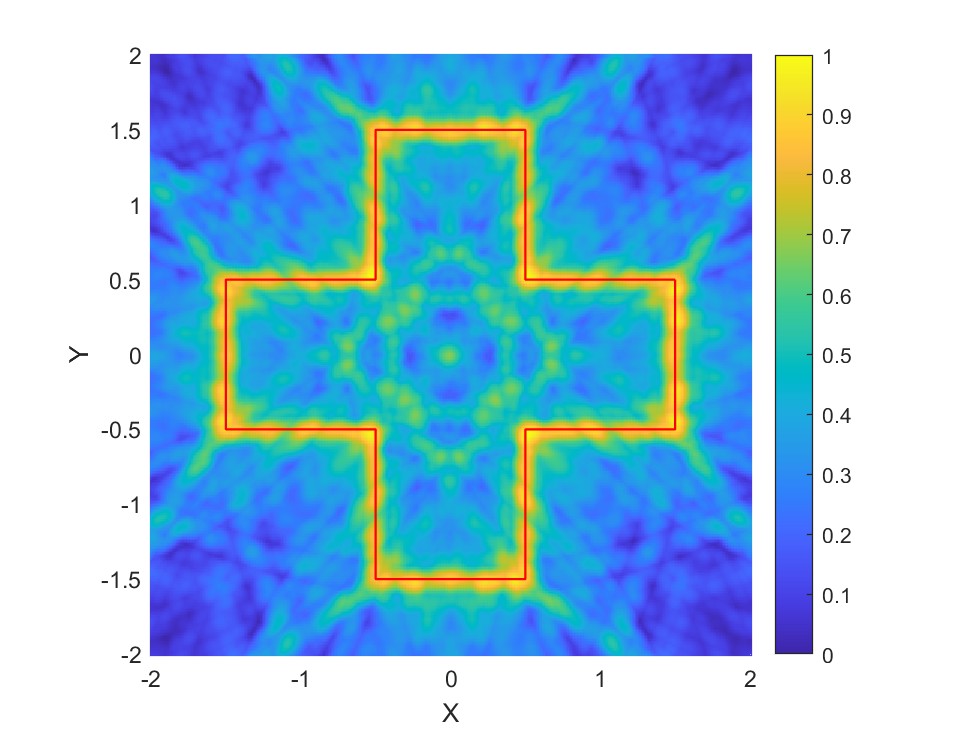}}&
\subfigure[$L=30$]{
\includegraphics[width=.30\textwidth]{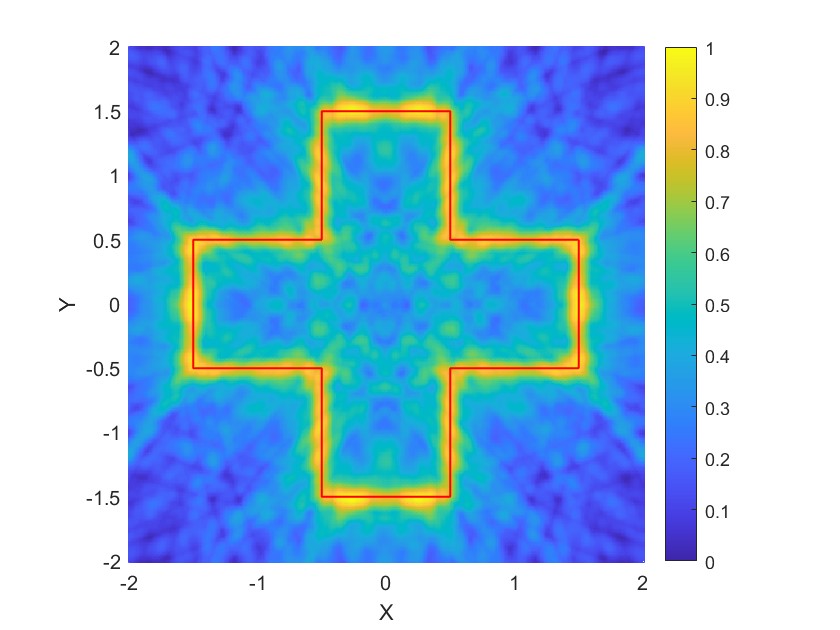}}
\end{tabular}
\caption{Reconstructions of the support set by plotting $I_{\pa \Omega}$. }
\label{L=10_20_30_annulus_+}
\end{figure}

Figure \ref{L=1_annulus_+} presents the reconstruction results obtained from scattered field measurements at a single sensor located at $x=(3,0)$. 
Consistent with theoretical predictions, the indicator function $I_{\pa \Omega}$ successfully identifies four circles centered at $x=(3,0)$ that are tangent to both the inner and outer boundaries of the annular source.
For the cross-shaped source, $I_{\pa \Omega}$ detects circles centered at $x=(3,0)$ that either intersect the polygon's vertices or are tangent to its edges.  

From the first row of Figure \ref{L=10_20_30_annulus_+}, we observe that both the inner and outer  boundaries of the annulus are successfully reconstructed as $L$ increases, with a superior effect achieved for the outer circle. This phenomenon is consistent with the asymptotic jump formulas \eqref{annular-dis1}--\eqref{annular-dis2}, where the indicator $|I_{\pa \Omega}|$ attains larger values along the outer boundary due to its greater radius.
We also observe that all the boundaries of the cross-shaped polygonal sources are progressively reconstructed with increasing number of sensors. Such a fact is consistent with the the theoretical analyses in the equalities \eqref{polygon-dis1}, \eqref{polygon-dis2}, and Theorem \ref{polygon_edges}.
\begin{figure}[h!]
\centering
\begin{tabular}{ccc}
\subfigure{
\includegraphics[width=.30\textwidth]{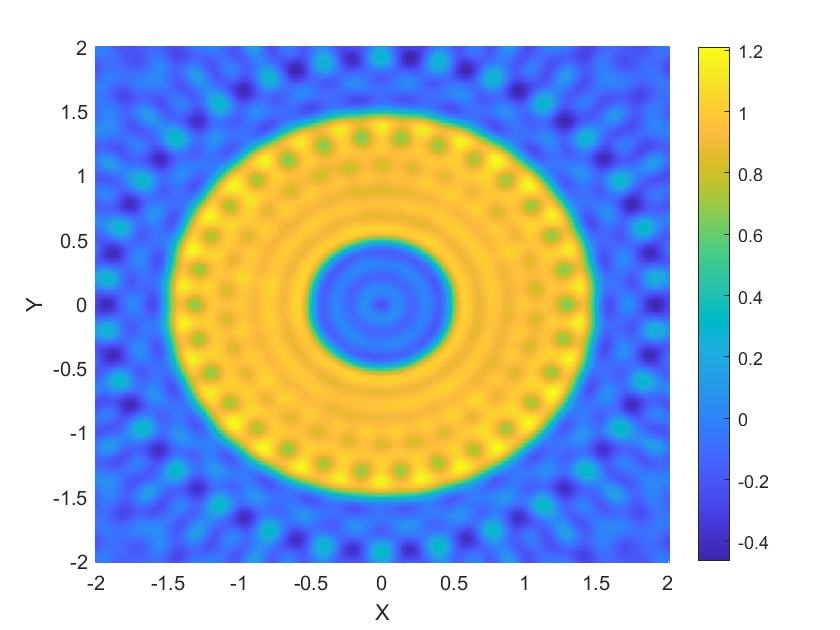}
}\hspace{0em} &
\subfigure{
\includegraphics[width=.30\textwidth]{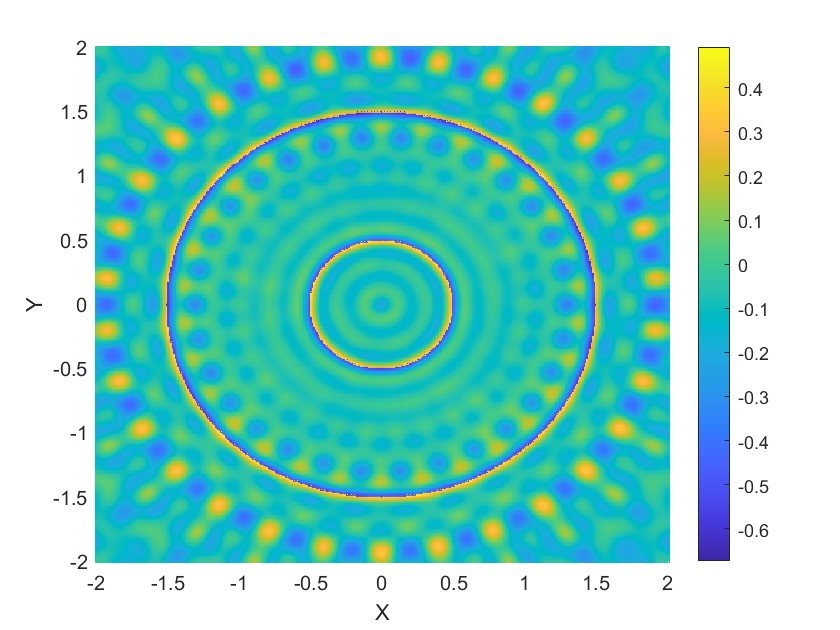}}&
\subfigure{
\includegraphics[width=.30\textwidth]{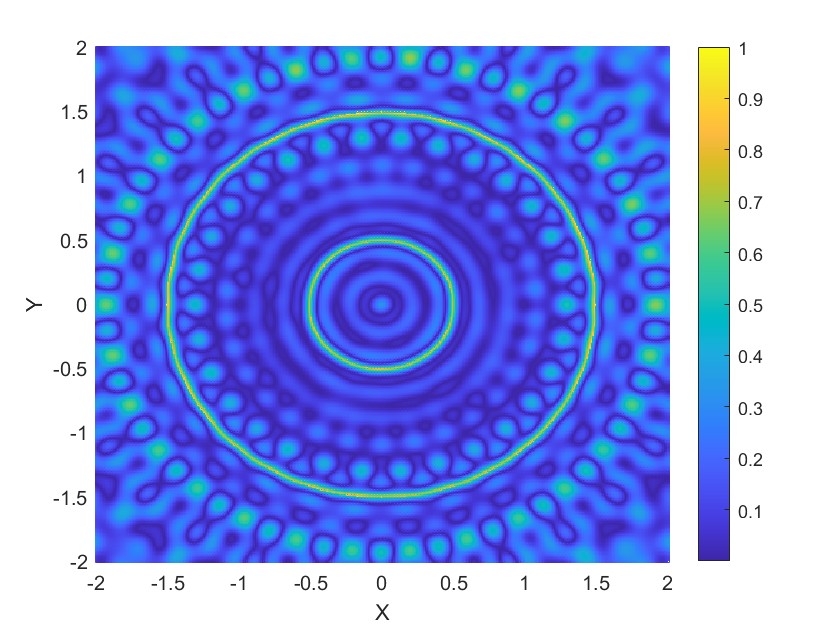}}\\
\subfigure{
\includegraphics[width=.30\textwidth]{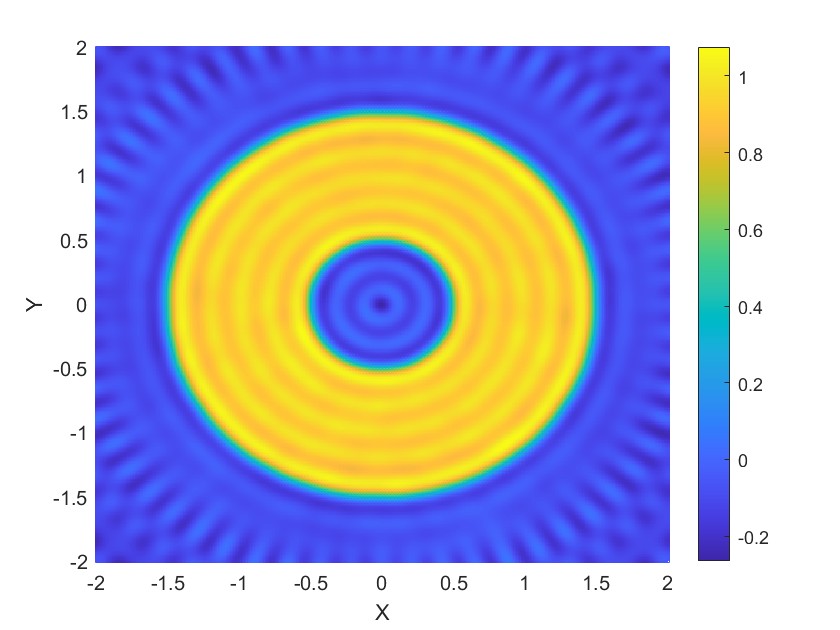}
}\hspace{0em} &
\subfigure{
\includegraphics[width=.30\textwidth]{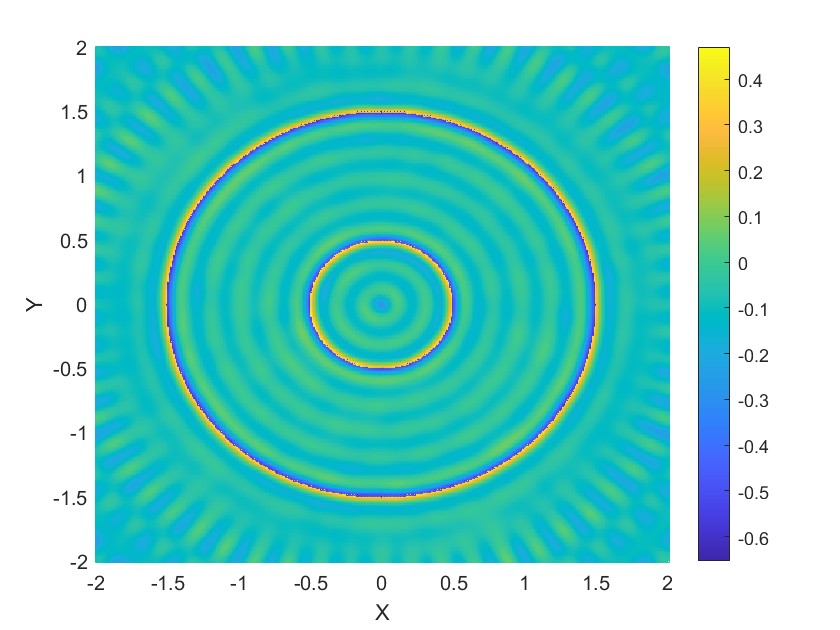}}&
\subfigure{
\includegraphics[width=.30\textwidth]{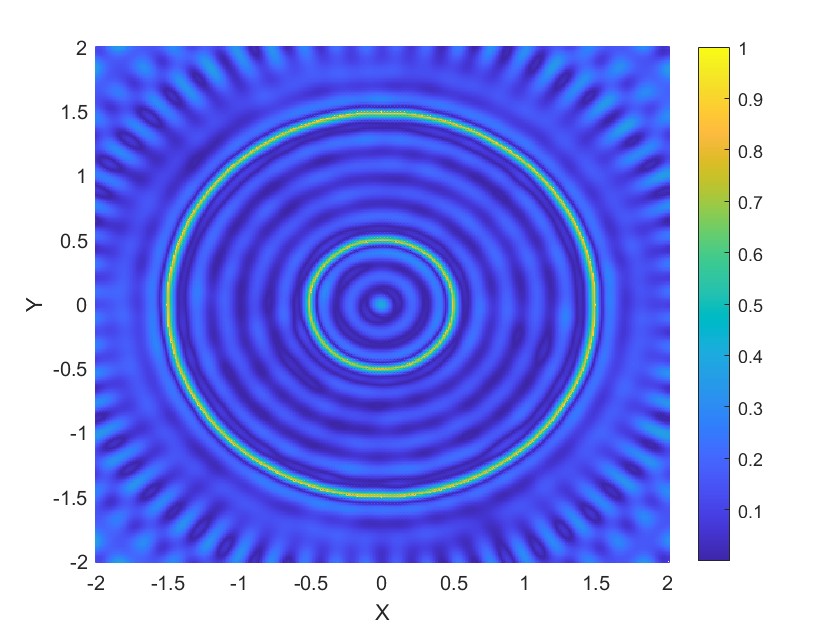}}
\end{tabular}
\caption{Reconstructions of the annular source using the indicator function $I^{(1)}_{S}$.\quad Left: Reconstructions with $L=30$ and $L=60$ respectively.\quad Middle: The error $I_S^{(1)}-S$.\quad Right: Normalized error.}
\label{annulus_07}
\end{figure}

\begin{figure}[h!]
\centering
\begin{tabular}{ccc}
\subfigure{
\includegraphics[width=.30\textwidth]{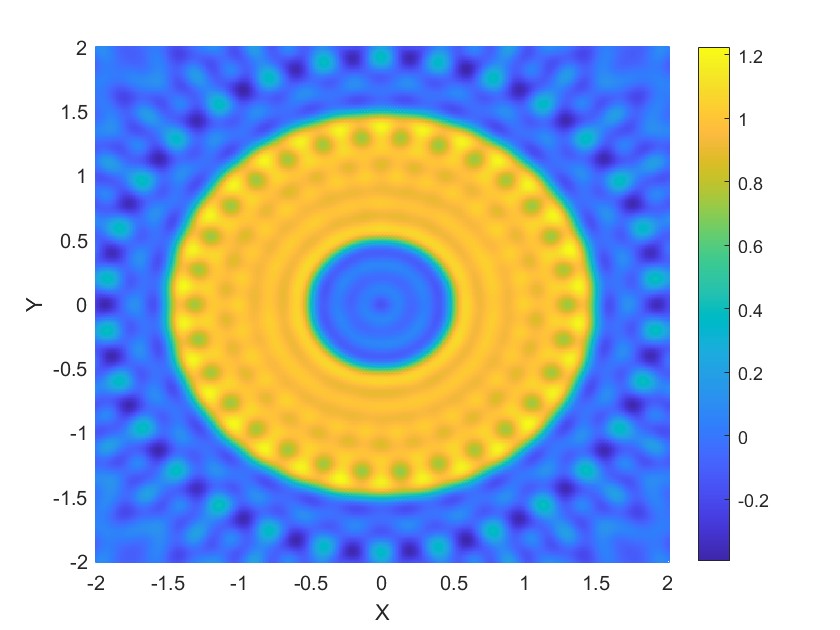}
}\hspace{0em} &
\subfigure{
\includegraphics[width=.30\textwidth]{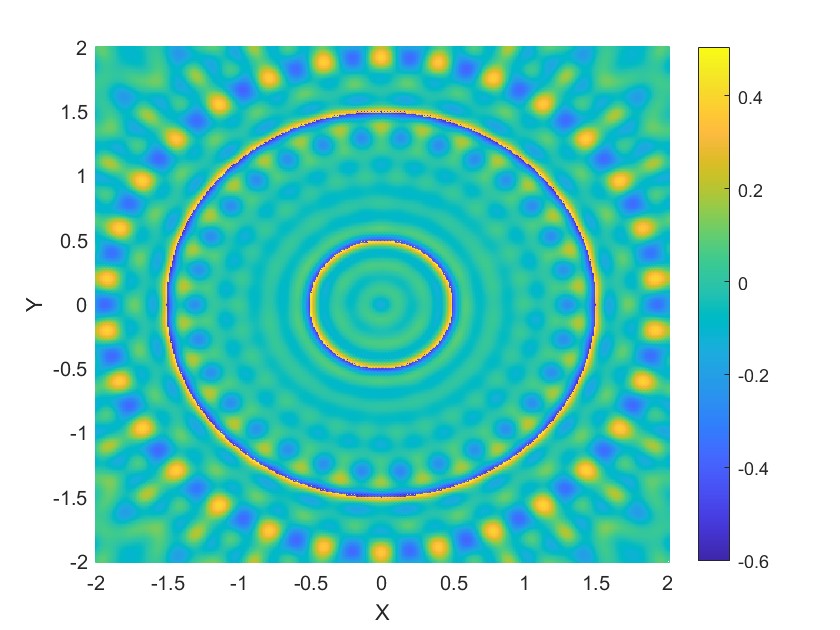}}&
\subfigure{
\includegraphics[width=.30\textwidth]{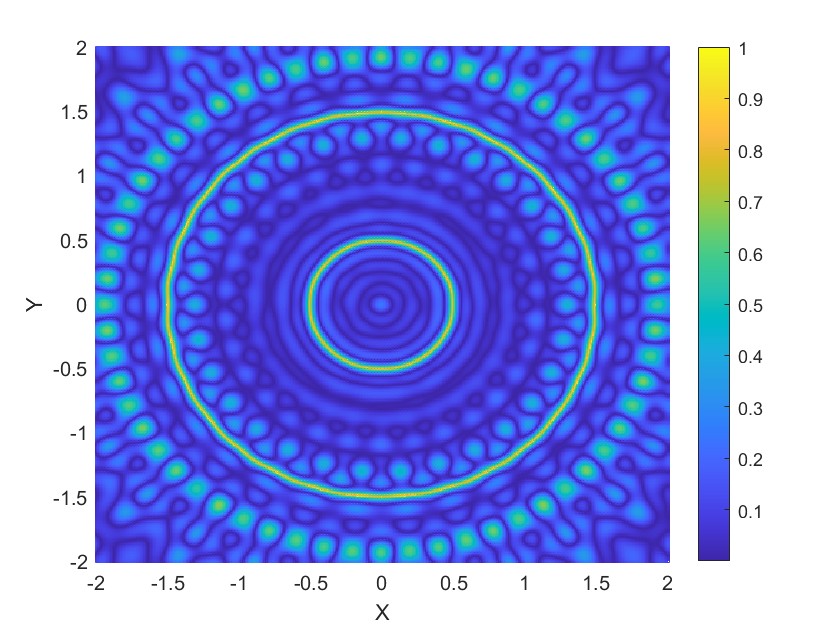}}\\
\subfigure{
\includegraphics[width=.30\textwidth]{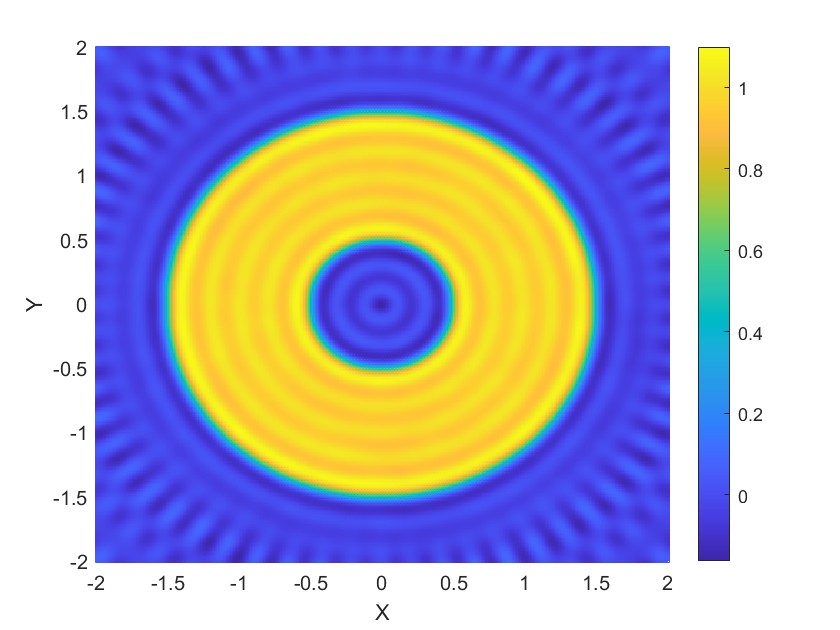}
}\hspace{0em} &
\subfigure{
\includegraphics[width=.30\textwidth]{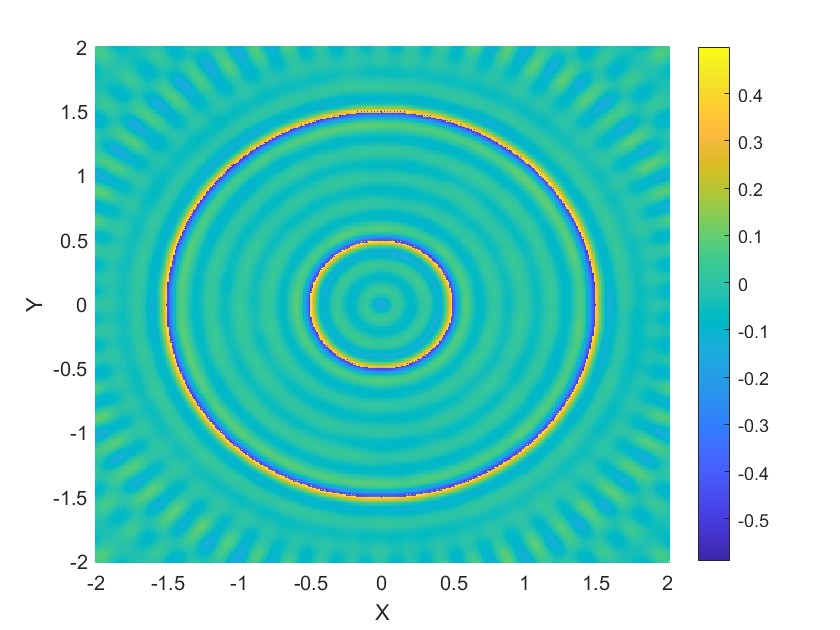}}&
\subfigure{
\includegraphics[width=.30\textwidth]{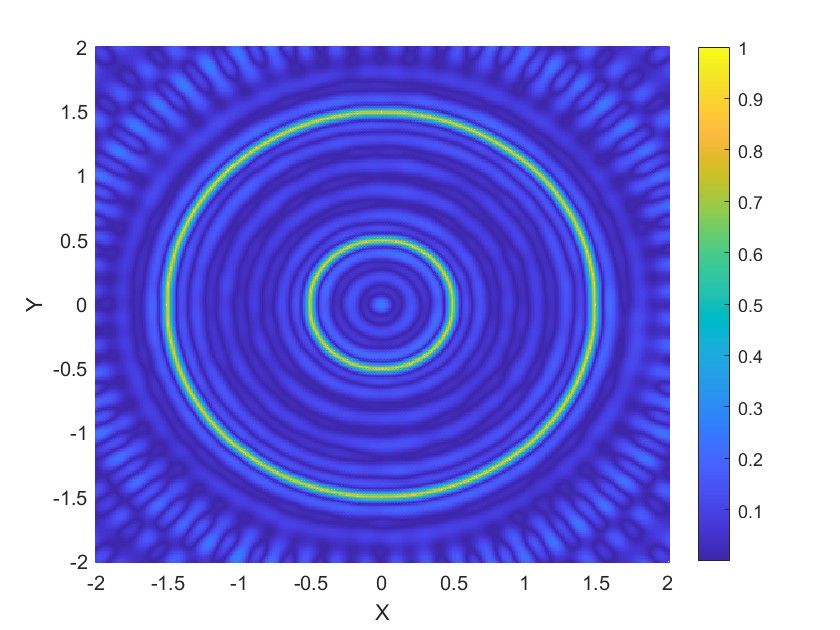}}
\end{tabular}
\caption{Reconstructions of the annular source using the indicator function $I^{(2)}_{S}$.\quad Left: Reconstructions with $L=30$ and $L=60$ respectively.\quad Middle: The error $I_S^{(2)}-S$.\quad Right: Normalized error.}
\label{annulus_two_layer}
\end{figure}

\begin{figure}[h!]
\centering
\begin{tabular}{ccc}
\subfigure{
\includegraphics[width=.30\textwidth]{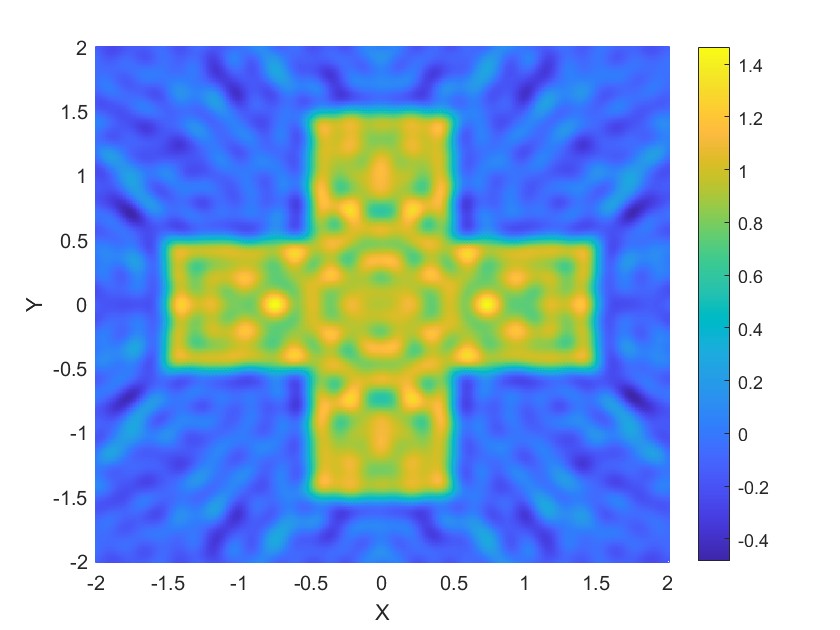}
}\hspace{0em} &
\subfigure{
\includegraphics[width=.30\textwidth]{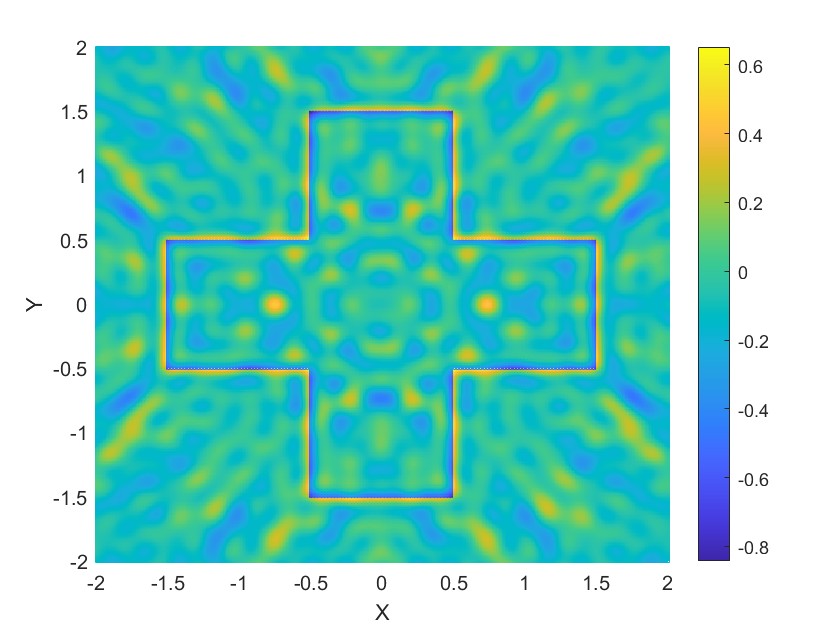}}&
\subfigure{
\includegraphics[width=.30\textwidth]{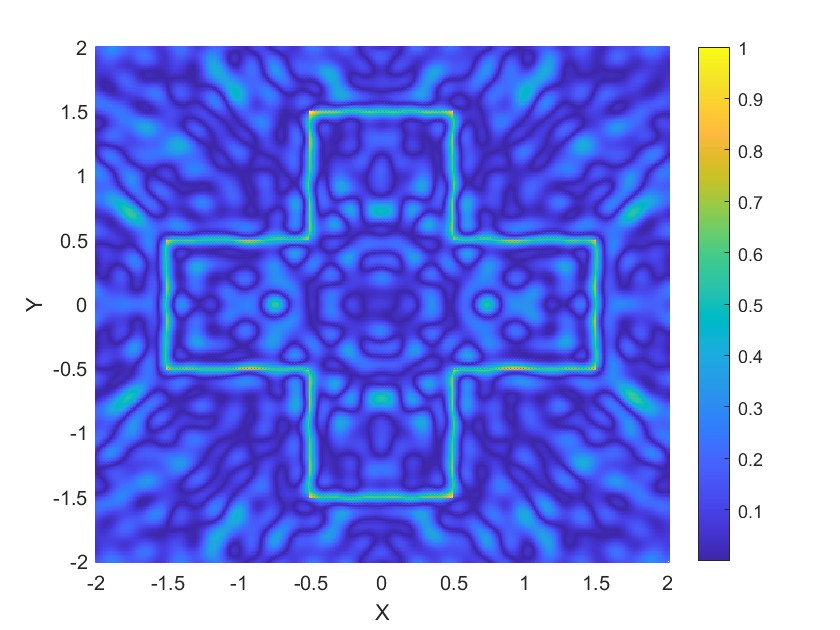}}\\
\subfigure{
\includegraphics[width=.30\textwidth]{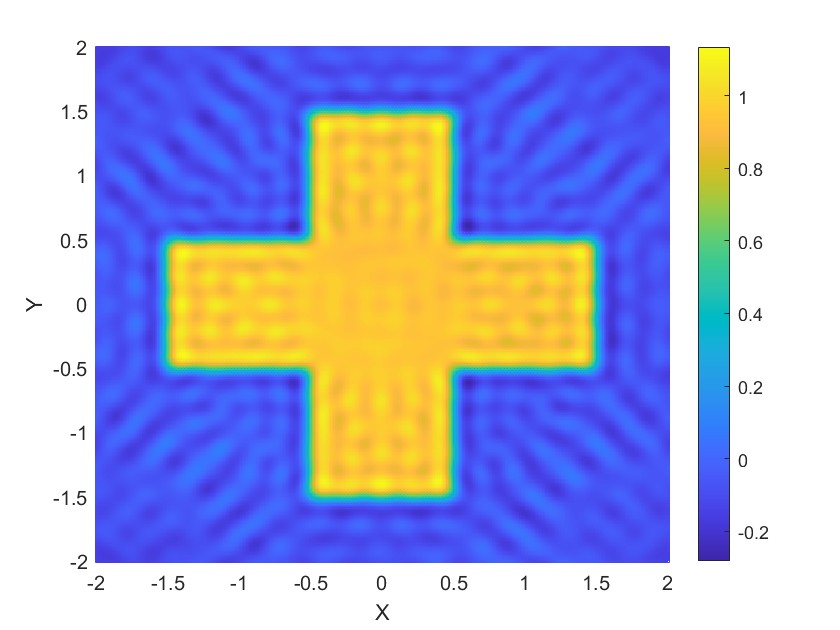}
}\hspace{0em} &
\subfigure{
\includegraphics[width=.30\textwidth]{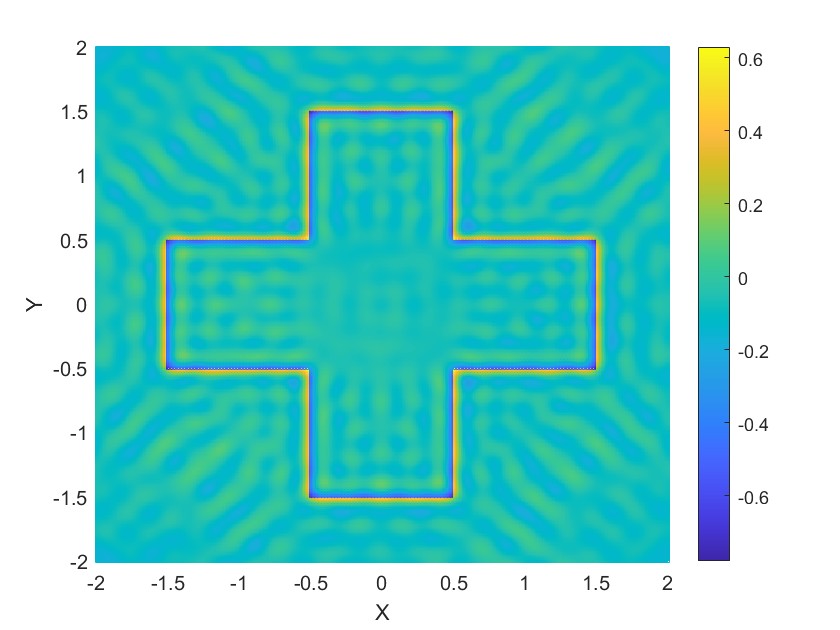}}&
\subfigure{
\includegraphics[width=.30\textwidth]{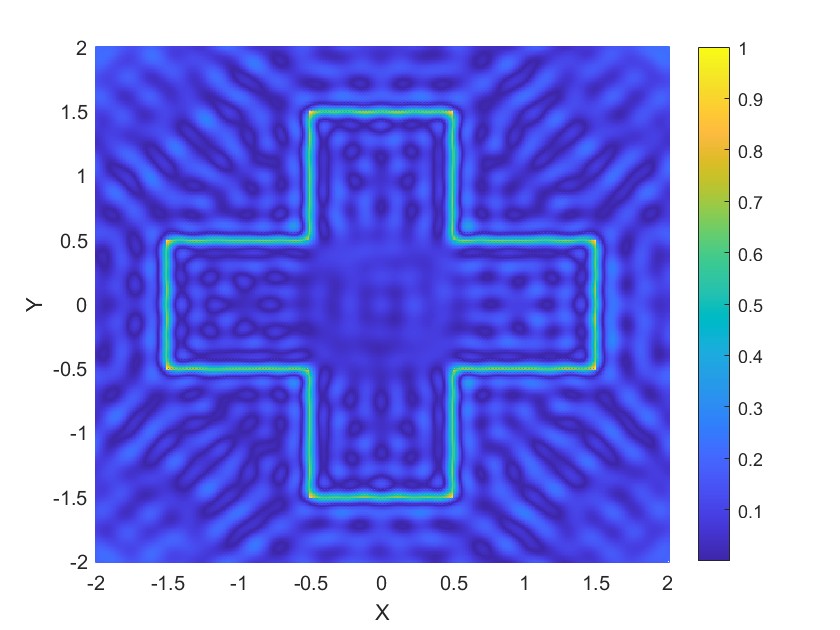}}
\end{tabular}
\caption{Reconstructions of the polygon source  using the indicator function $I^{(1)}_{S}$.\quad Left: Reconstructions with $L=30$ and $L=60$ respectively.\quad Middle: The error $I_S^{(1)}-S$.\quad Right: Normalized error.}
\label{polygon_07}
\end{figure}

\begin{figure}[h!]
\centering
\begin{tabular}{ccc}
\subfigure{
\includegraphics[width=.30\textwidth]{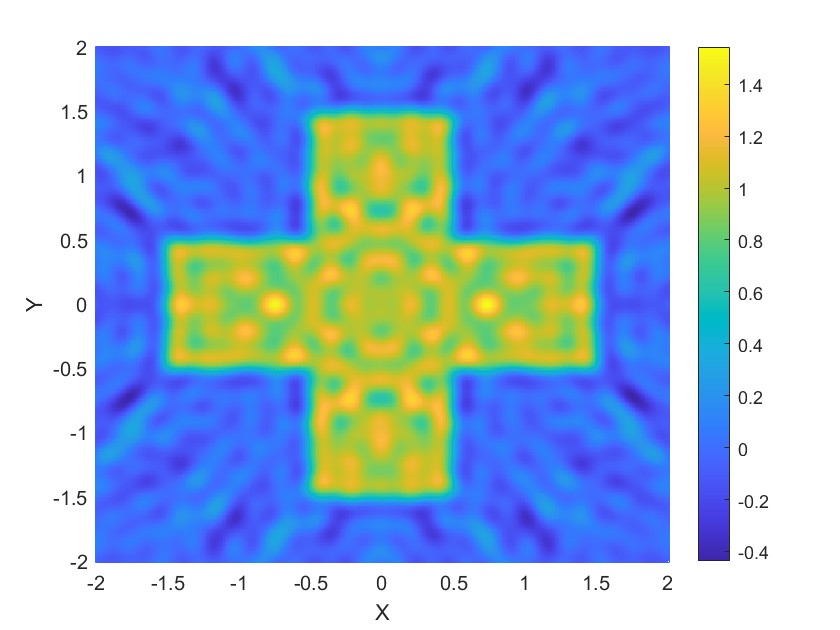}
}\hspace{0em} &
\subfigure{
\includegraphics[width=.30\textwidth]{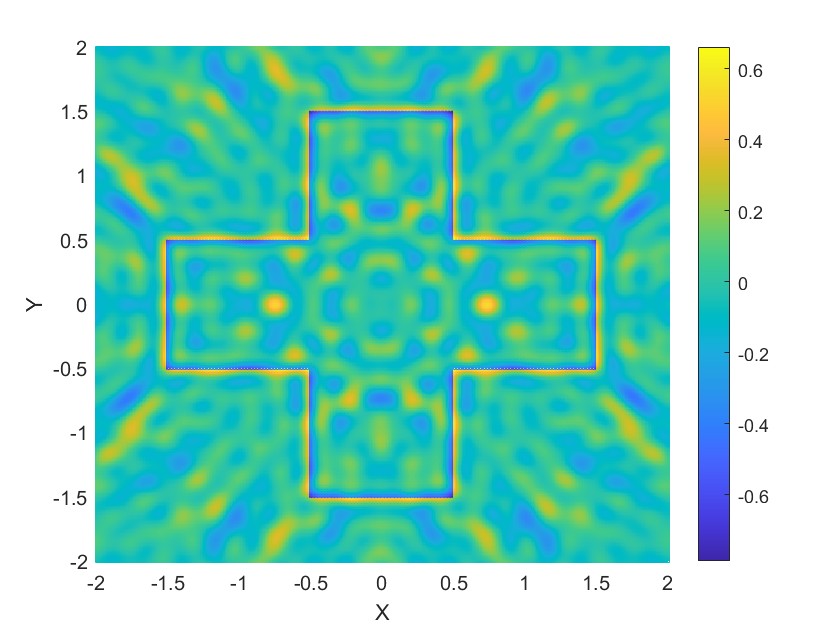}}&
\subfigure{
\includegraphics[width=.30\textwidth]{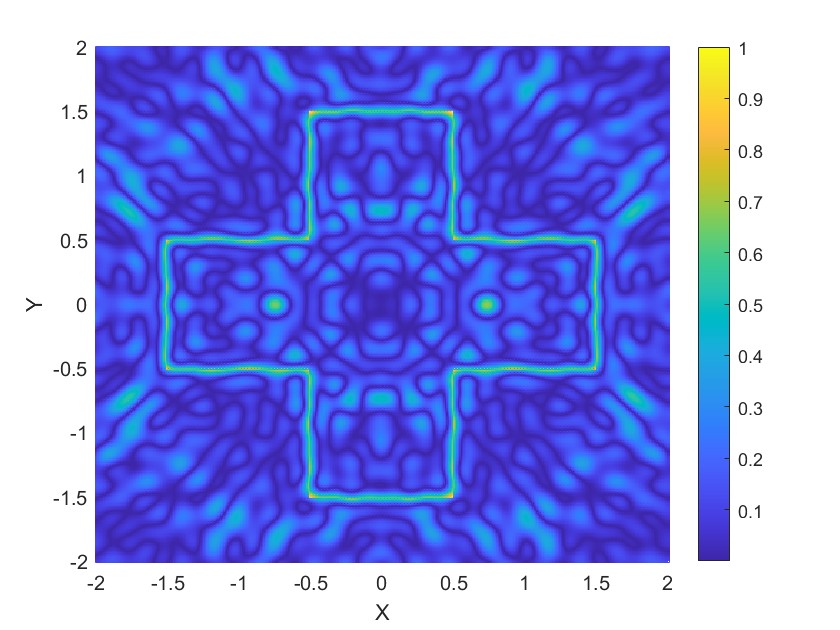}}\\
\subfigure{
\includegraphics[width=.30\textwidth]{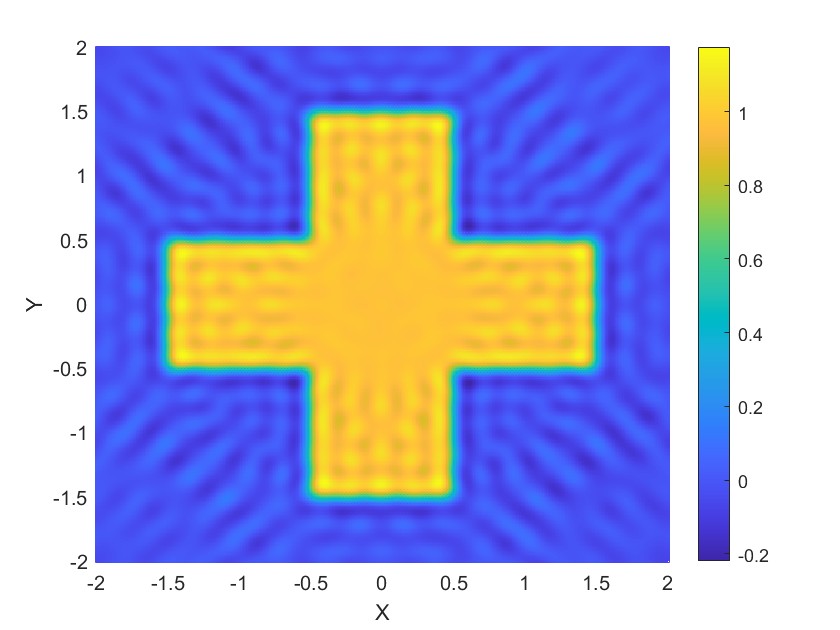}
}\hspace{0em} &
\subfigure{
\includegraphics[width=.30\textwidth]{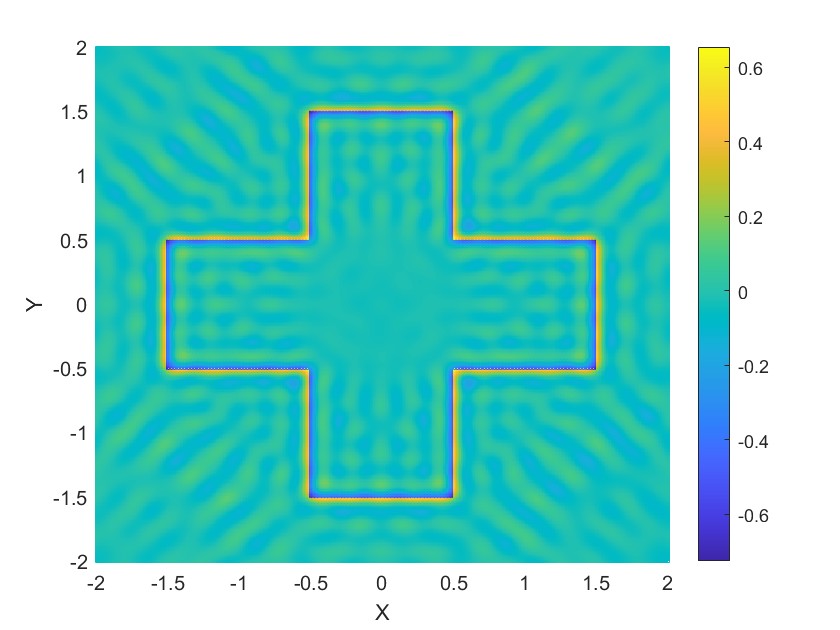}}&
\subfigure{
\includegraphics[width=.30\textwidth]{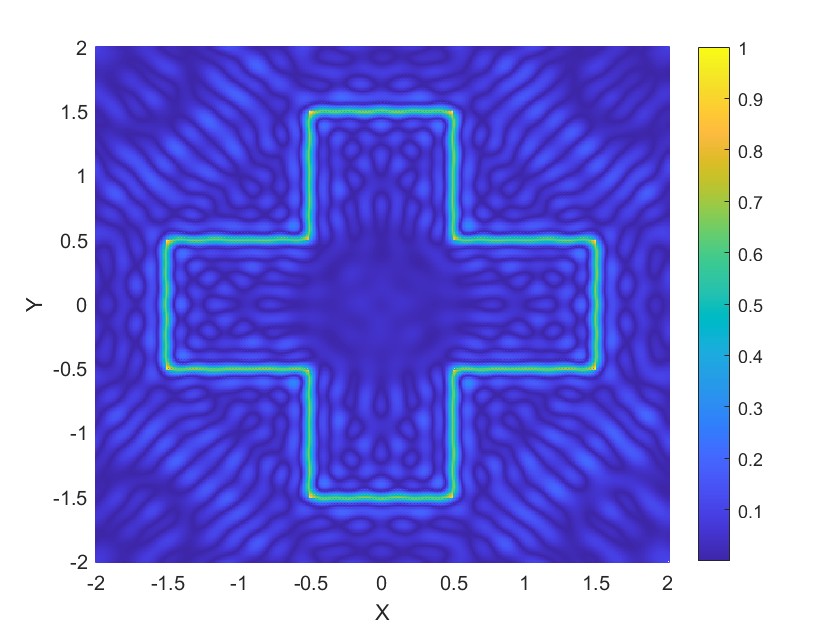}}
\end{tabular}
\caption{Reconstructions of the polygon source  using the indicator function $I^{(2)}_{S}$.\quad Left: Reconstructions with $L=30$ and $L=60$ respectively.\quad Middle: The error $I_S^{(2)}-S$.\quad Right: Normalized error.}
\label{polygon_two_layer}
\end{figure}

Figures \ref{annulus_07} and \ref{annulus_two_layer} show reconstructions of the annular source using the indicator functions $I^{(1)}_{S}$ and $I^{(2)}_{S}$, respectively. 
The left column presents the results for $L=30$ and $L=60$, illustrating that the error in the reconstructed source decreases as $L$ increases. The middle column shows the error distribution of the reconstructed source function. For $L=30$, in addition to the concentration of error along the boundary of the source support, there are point-like radial artifacts. These artifacts arise from the finite truncation applied to the Fourier transform in the numerical experiment, which induces the Gibbs phenomenon. 
When $L$ increases to 60, the point-like radial artifacts vanish and Gibbs oscillations become confined mainly to the boundary region. This effect is further highlighted in the right column of Figures \ref{annulus_07} and \ref{annulus_two_layer} through error normalization. These observations indicate that in our numerical reconstructions, errors are predominantly concentrated along the boundary. Moreover, as the amount of data increases, the boundary error does not vanish but stabilizes.
Similarly, Figures \ref{polygon_07} and \ref{polygon_two_layer} present the reconstruction of the polygonal source using the two indicator functions $I^{(1)}_{S}$ and $I^{(2)}_{S}$, respectively. 
% Notably, for both annular and polygonal sources, the indicator function $I^{(2)}_{S}$ demonstrates superior reconstruction accuracy with lower error rates compared to  $I^{(1)}_{S}$.

The third example involves a source function whose support is the smiling bear-shaped domain shown in Figure \ref{num_ex_3}. We perform reconstruction experiments using three distinct indicator functions $I_{\pa \Omega}$, $I_S^{(1)}$ and $I_{S}^{(2)}$, with the results presented in Figure \ref{Smiling_Bear}. Comparing the experimental results in the first and second rows of Figure \ref{Smiling_Bear}, we observe that both $I_S^{(1)}$ and $I_{S}^{(2)}$ successfully reconstruct the smiling bear as the data volume increases. In contrast, the boundary indicator function $I_{\pa \Omega}$ fails to accurately recover some internal morphological features of the smiling bear. This limitation stems from the non-simply connected nature of the function's support and the influence of the Gibbs phenomenon. Nevertheless, even under these conditions, the indicator function $I_{\pa \Omega}$ effectively reconstructs key characteristics such as the general contour and the eyes of the smiling bear, when sufficient data are provided.
\begin{figure}[h!]
\centering
\begin{tabular}{ccc}
\subfigure{
\includegraphics[width=.30\textwidth]{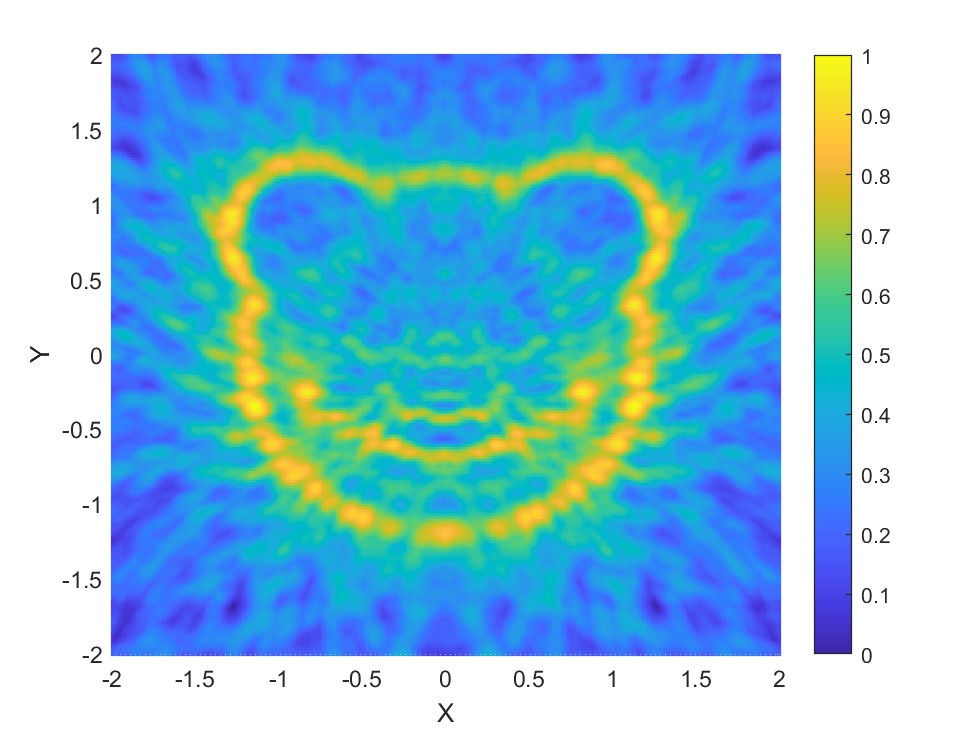}
}\hspace{0em} &
\subfigure{
\includegraphics[width=.30\textwidth]{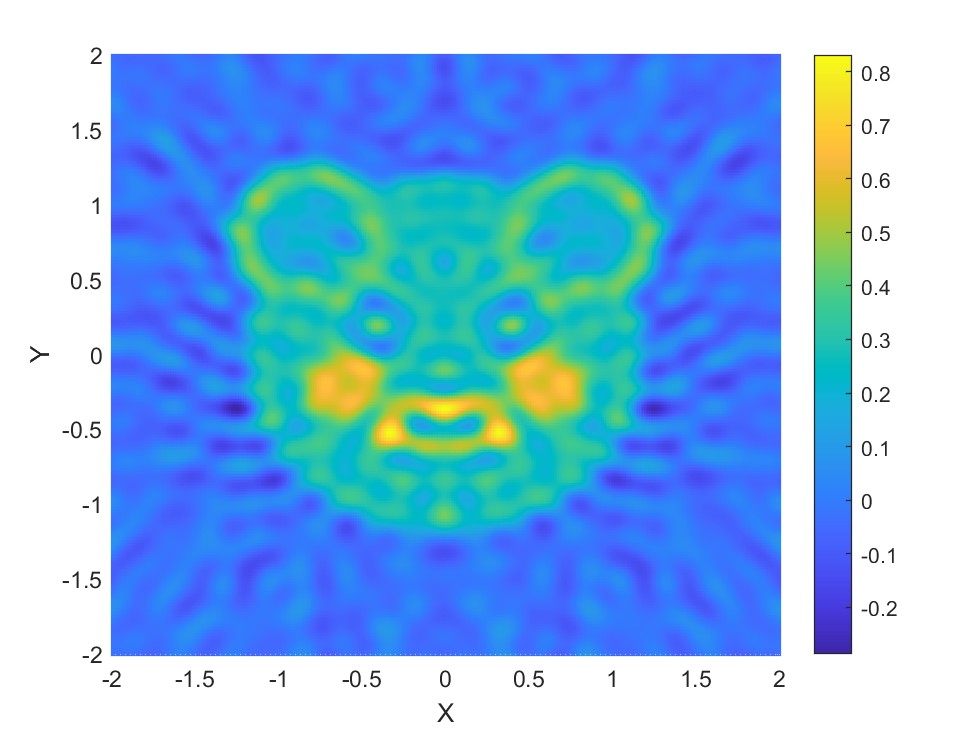}}&
\subfigure{
\includegraphics[width=.30\textwidth]{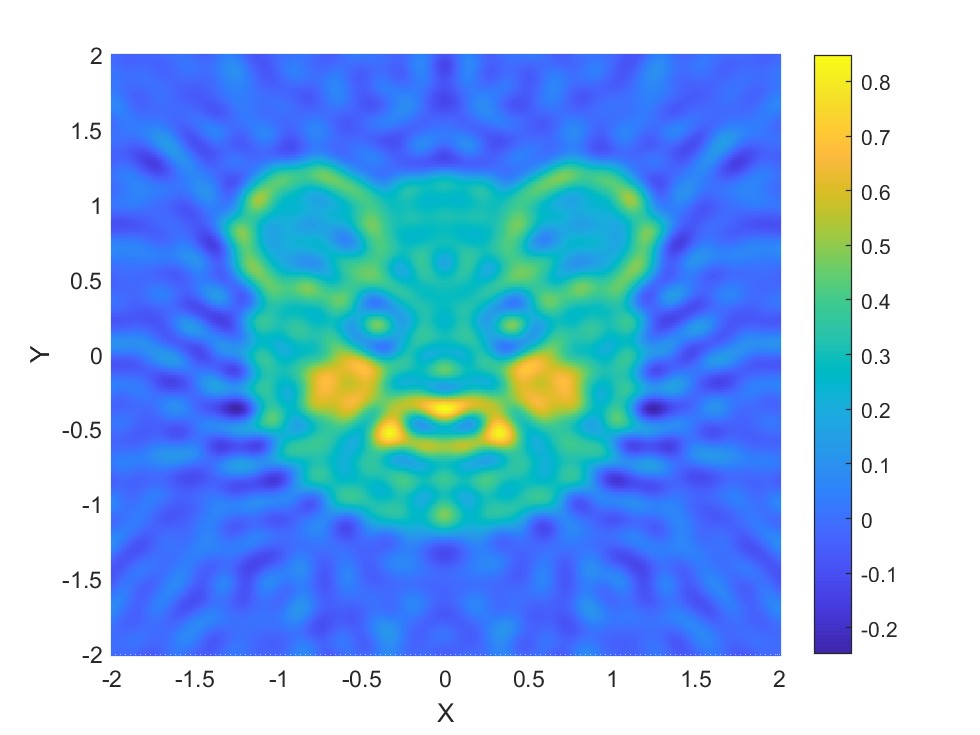}}\\
\subfigure{
\includegraphics[width=.30\textwidth]{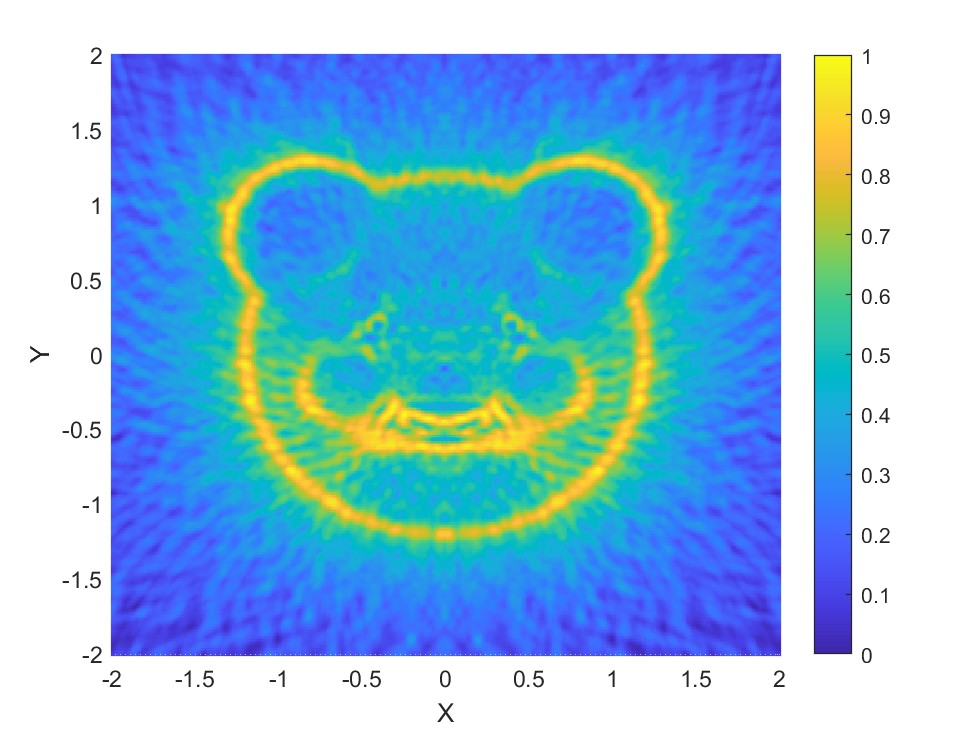}
}\hspace{0em} &
\subfigure{
\includegraphics[width=.30\textwidth]{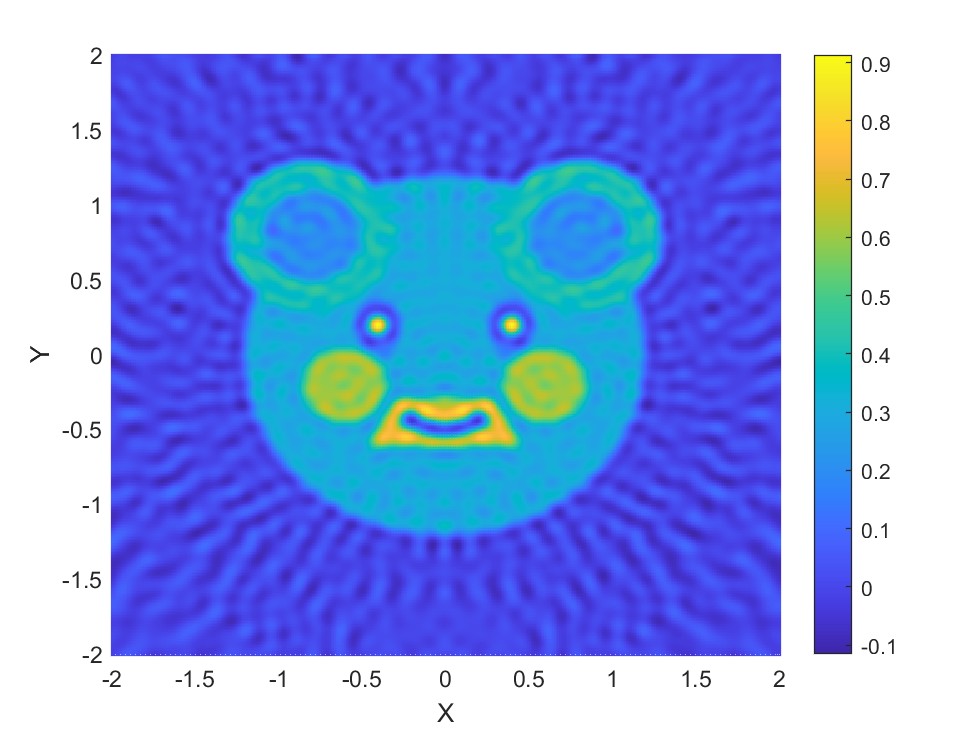}}&
\subfigure{
\includegraphics[width=.30\textwidth]{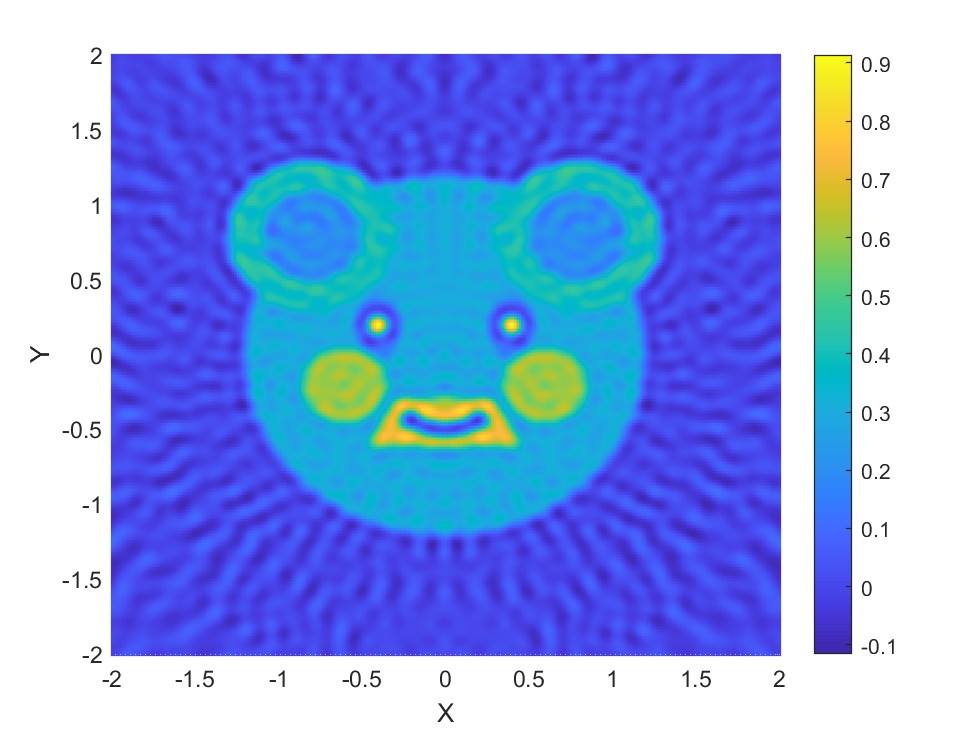}}
\end{tabular}
\caption{Reconstructions of the smiling bear: the first row employs parameters $k_-=0.5$, $k_+=30$, $dk=0.5$, $L=30$, while the second row utilizes $k_-=0.1$, $k_+=50$, $dk=0.1$, $L=60$.\quad Left: Reconstructions with the indicator function $I_{\pa \Omega}$.\quad Middle: Reconstructions with the indicator function $I_{S}^{(1)}$.\quad Right: Reconstructions with the indicator function $I_{S}^{(2)}$.}
\label{Smiling_Bear}
\end{figure}

Finally, we investigate the smooth source function, with the corresponding reconstructions shown in Figure \ref{L=30_k=30_source function0}. 
%Reconstruction performance is evaluated through both three-dimensional (3D) and two-dimensional (2D) visualizations in Figures \ref{L=30_k=30_source function1} and \ref{L=30_k=30_source function2}, using indicator functions $I^{(1)}_S$ and $I^{(2)}_S$, respectively. The middle columns of these Figures \ref{L=30_k=30_source function1} and \ref{L=30_k=30_source function2} present 3D and 2D reconstruction results obtained through sampling-based implementations of $I^{(1)}_S$ and $I^{(2)}_S$, while the right columns display the associated error distributions. 
Table \ref{tab:Comparison-I12} confirms the stability and high accuracy of reconstructions achieved by the proposed indicators.
%Under identical data parameters, the computational times consumed for computing the indicator functions $I_S^{(1)}$ and $I_S^{(2)}$ are $62.9777s$ and $51.0937s$, respectively. The corresponding relative errors, defined as $\frac{\|S(z)-I_S^{(i)}(z)\|_2}{\|S(z)\|_2}$, ($i=1,2$), are $0.4002$ and $0.2228$, respectively. Despite a ‌$20\%$ relative noise level‌ in the measurements, both reconstructions achieve ‌high accurarcy‌.
Moreover, comparative analysis reveals that  $I^{(2)}_S$ consistently surpasses  $I^{(1)}_S$ in both computational efficiency and reconstruction fidelity.
To illustrate this, we define a discrete approximation of the Radon transform in \eqref{I(r)} as
\begin{equation}
g(h):=\sum\limits_{j=1}^{60} 8k_j^3h \Im(u^s(x,k_j))J_0(k_j h),\quad x=(3 ,0), k_j=0.5j\,\, \text{or}\,\, 0.1j.
\nonumber
\end{equation}
% Prior to analyzing the more complex example, 
We conduct a numerical comparison of the two indicator functions $I^{(1)}_{S}$ and $I^{(2)}_{S}$.  $I^{(1)}_{S}$ (derived from Theorem \ref{source function-07}) involves a quadruple integral, whereas  $I^{(2)}_{S}$ (based on Theorem \ref{S(y)-fomula-proof111}) requires only a double integral.
In \eqref{07-I_S}, the innermost layer of the quadruple integral essentially generates Radon transform data of $S(z)$ from the scattered field measurements. However, this numerical integration introduces certain errors in the original scattered field data. As an illustration, Figure \ref{radon_error} shows error propagation in this procedure for the annular source. The red curve represents the Radon transform of the source function, while the blue curve depicts its numerical approximation based on the function $g(h)$.
Consequently, with 30 sensors and a frequency spacing of $dk=0.5$, the performance of the indicator function $I^{(1)}_S$ is observed to be inferior to that of $I^{(2)}_S$, attributable to the cumulative error introduced by multiple integration layers. 
As the frequency spacing becomes sufficiently small, the error introduced by the Radon transform decreases significantly, which can be observed in Figure \ref{k_0.1}. At the same time, the reconstruction error of the corresponding indicator function $I^{(1)}_S$ decreases as $d_k$ decreases.
%This phenomenon will be further validated through error analysis in the following example.

\begin{figure}[h!]
\centering
\begin{tabular}{cc}
\subfigure[$dk=0.5$]{
\includegraphics[width=.40\textwidth]{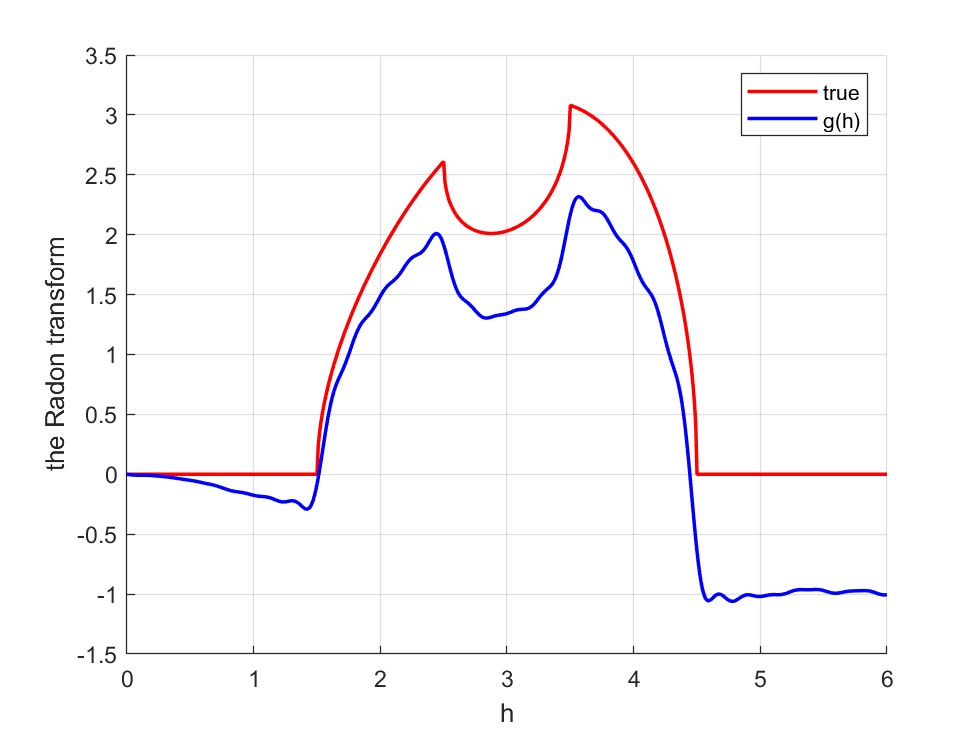}
}\hspace{0em} &
\subfigure[$dk=0.1$]{
\includegraphics[width=.40\textwidth]{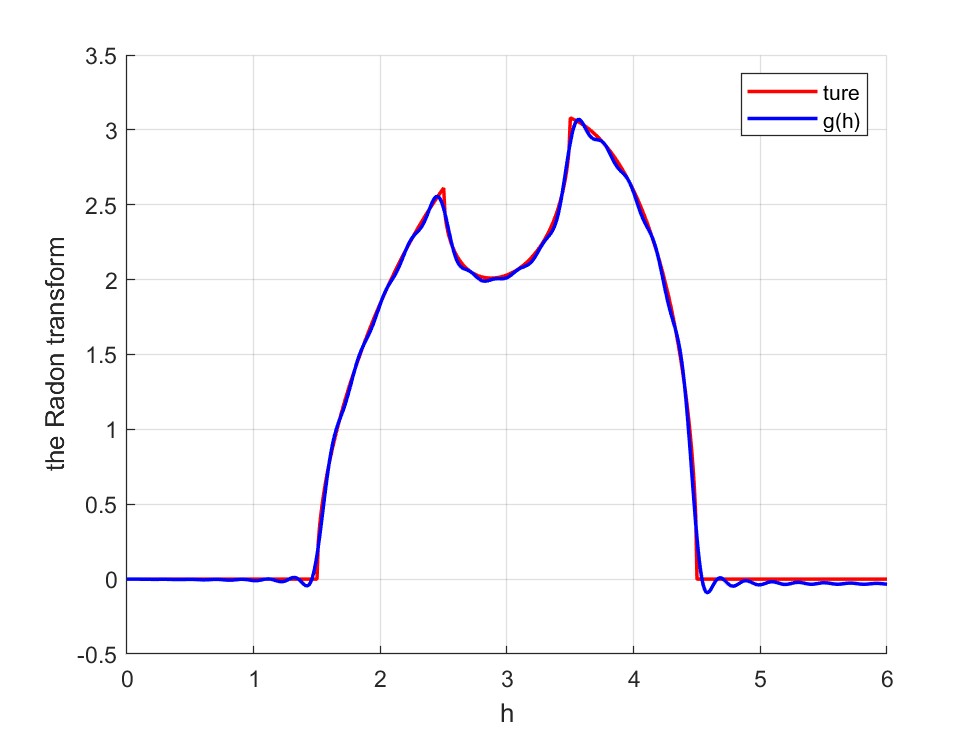}
\label{k_0.1}}
\end{tabular}
\caption{Comparison between the exact Radon transform and its approximation function $g(h)$.}
\label{radon_error}
\end{figure}
\begin{table}
    \centering
    \begin{tabular}{|c|c|c|} % 添加竖线 | 实现边框
        \hline % 添加水平线
                     & Relative error ($dk=0.5$)& Relative error ($dk=0.1$))\\
        \hline % 添加水平线
        $I_S^{(1)}(L=30)$  & 0.4002 & 0.2178  \\
        \hline % 添加水平线
        $I_S^{(2)}(L=30)$  & 0.2228 & 0.1029  \\
        \hline % 添加水平线
        $I_S^{(1)}(L=60)$  & 0.3992 & 0.1335  \\
        \hline % 添加水平线
        $I_S^{(2)}(L=60)$  & 0.1964 & 0.0997  \\
        \hline % 添加水平线
    \end{tabular}
    \caption{The computation is performed on a $401\times 401$ grid. The relative error is defined as $\frac{\|S(z)-I_S^{(i)}(z)\|_2}{\|S(z)\|_2}$, ($i=1,2$). Despite a ‌$20\%$ relative noise in the measurements, both reconstructions achieve high accuracy. In particular, the reconstruction error is greatly reduced with the decrease of $dk$. This further verify the observation from Figure \ref{radon_error}. }
    \label{tab:Comparison-I12}
\end{table}

\begin{figure}[h!]
\centering
\begin{tabular}{ccc}
\subfigure{
\includegraphics[width=.30\textwidth]{R2-1.jpg}
}\hspace{0em} &
\subfigure{
\includegraphics[width=.30\textwidth]{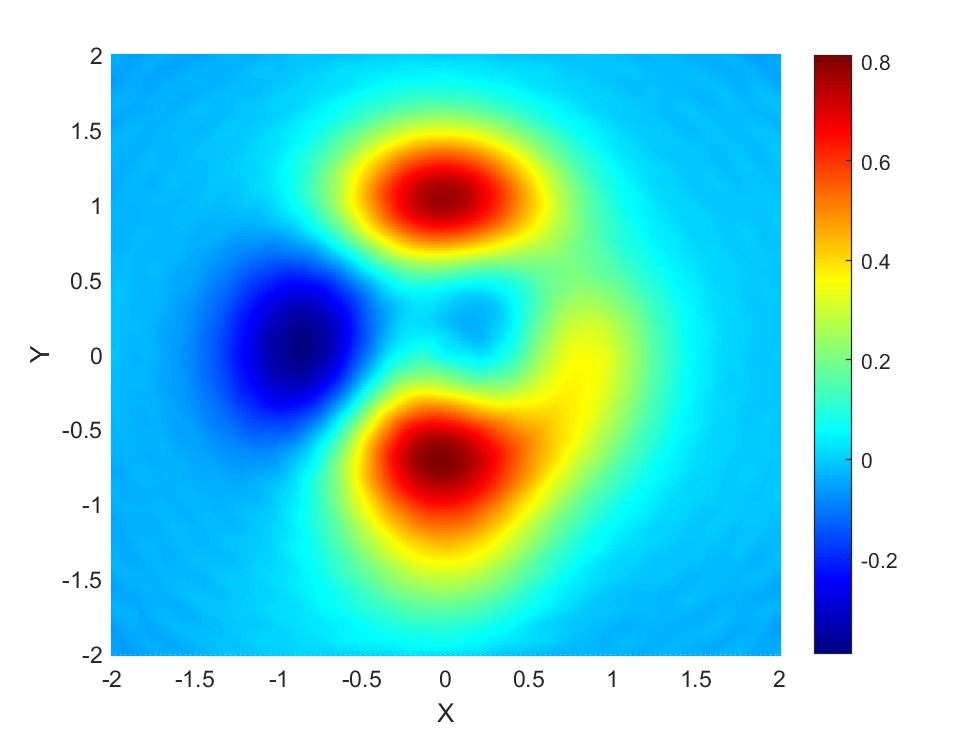}
}&
\subfigure{
\includegraphics[width=.30\textwidth]{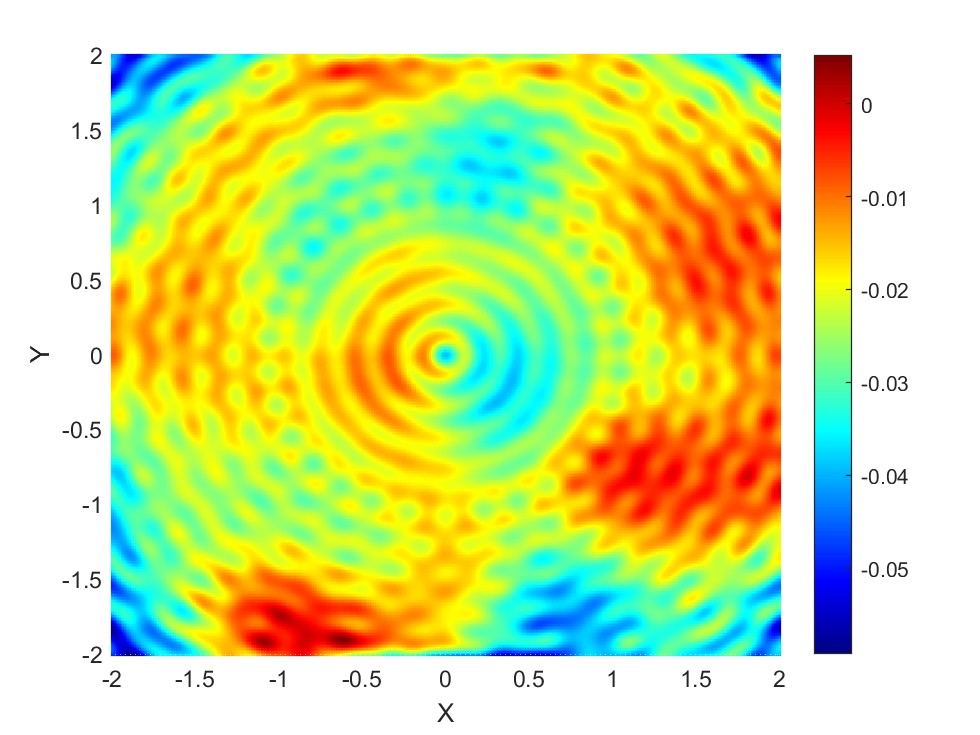}
}\\
\subfigure{
\includegraphics[width=.30\textwidth]{R2-1.jpg}
}\hspace{0em} &
\subfigure{
\includegraphics[width=.30\textwidth]{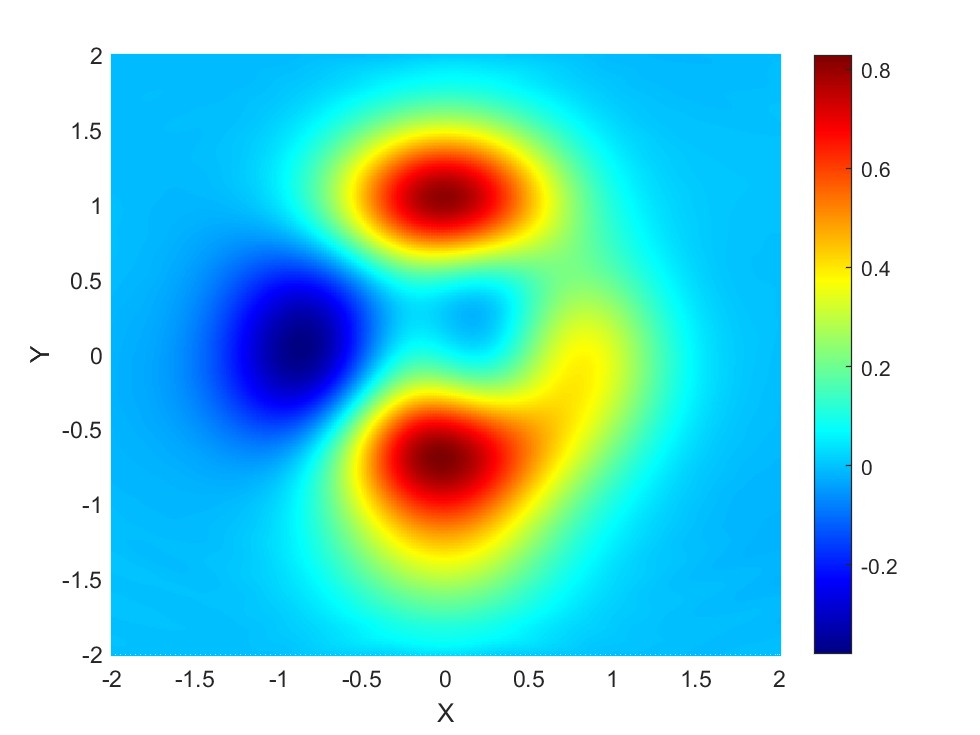}
}&
\subfigure{
\includegraphics[width=.30\textwidth]{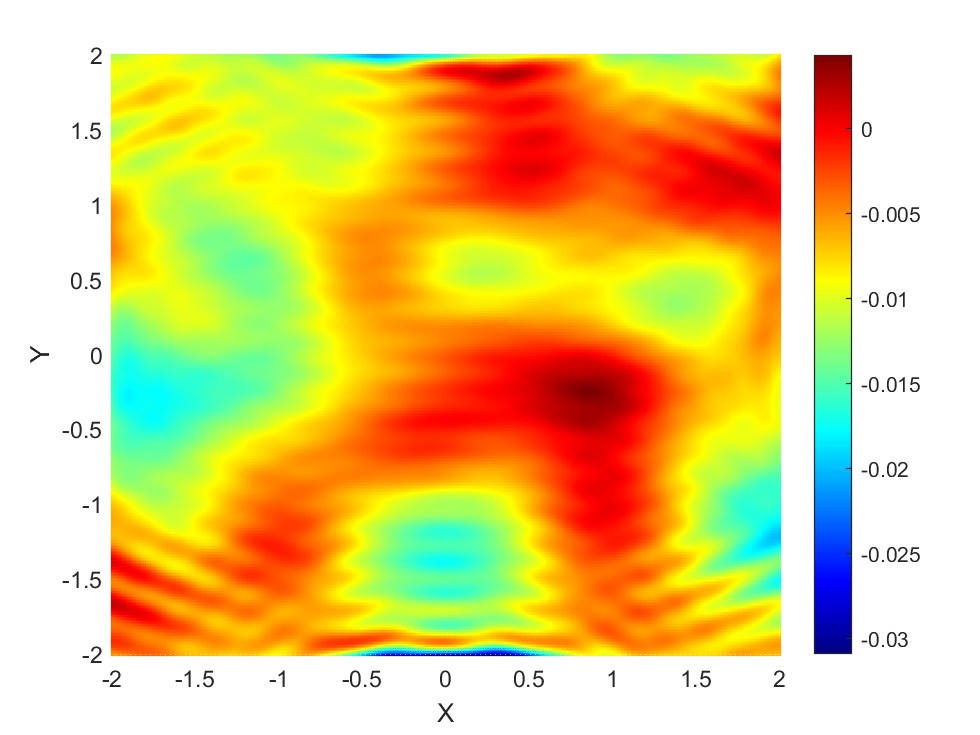}
}
\end{tabular}
\caption{Reconstrction of the smooth source function.\quad Left: The true source.\quad Middle: Reconstruction by plotting $I^{(1)}_{S}$ (top) and $I^{(2)}_{S}$ (bottom). \quad Right: The errors of $I^{(1)}_{S}$ (top) and $I^{(2)}_{S}$ (bottom).}
\label{L=30_k=30_source function0}
\end{figure}

\section{Concluding remarks}

This work focuses on characterizing the source support from the scattered fields taken at sparse sensors and developing fast, stable, and high-resolution numerical methods for identifying the source functions. Both objectives are highly relevant from a practical point of view. We demonstrate that the support of annular and polygonal sources can be uniquely determined using multi-frequency scattered fields measured at a finite number of sensors, where the required sensor count is explicitly determined by the number of unknown sources. Three sampling methods are proposed to reconstruct either the source support or the source functions themselves. 

To reconstruct annular and polygonal sources, we analyze the Radon transform of the source function to obtain the indicator function $I_{\pa \Omega}$ in \eqref{Ipa-Omega} that achieves large values exclusively on the boundary $\pa \Omega$. This allows for accurate shape recovery using only the imaginary part of the scattered field $\Im(u^s(x,k))$, measured at a sparse set of observation points. This represents an extremely small amount of data, which is of significant practical relevance. Moreover, by exploiting the Radon transform to establish a direct equality between $\Im(u^s(x,k))$ and the source function, the indicator function $I_S^{(1)}$ in \eqref{I-Omega-(1)} is based exclusively on $\Im(u^s(x,k))$ and completely avoids the need for $\Delta u^s(x,k)$, $\partial _{\nu}u^s(x,k)$ and $\partial _{\nu}\Delta_x u^s(x,k)$. In addition, with the inclusion of the data $\Delta u^s(x,k)$, the indicator function $I_S^{(2)}$ in \eqref{I-Omega-(2)} provides better reconstruction quality than $I_S^{(1)}$, as evidenced by our numerical experiments.

Although our numerical experiments demonstrate that the external boundaries of general domains can also be successfully reconstructed, the theoretical proof for general shapes and for complex-valued sources presents significant difficulties. This important and challenging problem will be a key focus of our future research.

\section*{Acknowledgment}
 The research of X. Liu is supported by the National Key R\&D Program of China under grant 2024YFA1012303  and the NNSF of China under grant 12371430.

\bibliographystyle{SIAM}

\end{document}